\documentclass[11pt]{article}
\usepackage[left=1in, right=1in, top=1in, bottom=1in, margin=1in]{geometry}

\usepackage{titlesec}
\usepackage{tcolorbox}
\usepackage{ninecolors}
\usepackage{xcolor}

\usepackage{amsmath,amssymb,amsthm,xspace,graphicx,relsize,bm,breqn,multirow}

\usepackage{tikz}
\usetikzlibrary{quantikz2}

\usepackage{adjustbox}

\usepackage{float}
\usepackage[linesnumbered,ruled,vlined]{algorithm2e}
\SetKwInput{KwInput}{Input}
\SetKwInput{KwOutput}{Output}
\SetKwInOut{Promise}{Promise}
\SetKwInput{Goal}{Goal}
\SetKwProg{Fn}{function}{}{}
\SetKwFor{RepTimes}{repeat}{times}{}
\SetKwFunction{Ver}{Verify}
\SetKwFunction{Prep}{PrepareState}
\SetKwComment{Comment}{/* }{ */}
\usepackage{algorithmic}

\usepackage{bbm}

\usepackage{thm-restate,mathrsfs}

\newtcolorbox{myalgorithm}[1][]{
    colback=gray!10, 
    colframe=black, 
    arc=5pt, 
    boxrule=0.5pt, 
    left=0pt, right=0pt, top=0pt, bottom=0pt 
}
\newcommand{\Exp}{\mathop{\mathbb{E}}}

\usepackage[margin=1in]{geometry}

\newcommand{\E}{\mathcal{E}}

\usepackage{amsthm}
\usepackage{dsfont}
\usepackage{array}
\usepackage{makecell}

\newcommand{\Sh}{\ensuremath{\mathsf{Stab}}}

\newcommand{\V}{\ensuremath{\mathcal{V}}}

\newcommand{\C}{\ensuremath{\mathcal{C}}}

\newcommand{\A}{\ensuremath{\mathcal{A}}}

\newcommand{\poly}{\ensuremath{\mathsf{poly}}}

\newcommand{\id}{\ensuremath{\mathbb{I}}}
\newcommand{\G}{\ensuremath{\mathcal{G}}}

\usepackage{upgreek}
\usepackage{enumerate}

\usepackage[pagebackref,hypertexnames=false]{hyperref}
\usepackage[usestackEOL]{stackengine}
\usepackage{thm-restate,mathrsfs}
\usepackage{enumerate}
\usepackage{array}
\usepackage{parskip}
\def\01{\{0,1\}}

\newcommand{\ketbra}[2]{|#1\rangle\langle#2|}

\newcommand{\be}{\begin{equation}}
\newcommand{\ee}{\end{equation}}
\newcommand{\ba}{\begin{array}}
\newcommand{\ea}{\end{array}}
\newcommand{\bea}{\begin{eqnarray}}
\newcommand{\eea}{\end{eqnarray}}

\newcommand{\Id}{\mathbbm{1}}

\usepackage{mathtools}
\DeclarePairedDelimiter\ceil{\lceil}{\rceil}

\DeclareMathOperator{\Tr}{Tr}
\newcommand{\ra}{\rangle}
\newcommand{\la}{\langle}

\newcommand{\opt}{\textsf{opt}}

\newcommand{\norm}[1]{\left\lVert#1\right\rVert}

\newcommand{\calA}{{\cal A }}
\newcommand{\calB}{{\cal B }}

\newcommand{\calL}{{\cal L }}
\newcommand{\calP}{{\cal P }}

\newcommand{\calF}{{\cal F }}
\newcommand{\calG}{{\cal G }}
\newcommand{\calV}{{\cal V }}
\newcommand{\calE}{{\cal E }}
\newcommand{\calU}{{\cal U }}
\newcommand{\calC}{{\cal C }}
\newcommand{\calS}{{\cal S }}

\newcommand{\calZ}{{\cal Z }}
\newcommand{\calT}{{\cal T }}
\newcommand{\FF}{\mathbb{F}}

\newcommand{\Selfcorrection}{\textsc{Self-Correction}}
\newcommand{\PFR}{\textsf{PFR}}
\newcommand{\APFR}{\textsf{APFR}}
\newcommand{\LCU}{\textsf{LCU}}
\newcommand{\BSG}{\textsf{BSG}}
\usepackage{setspace}

\newcommand{\PREP}{\mathrm{PREP}}
\newcommand{\SEL}{\mathrm{SEL}}

\newcommand{\gowers}[2]{\textsc{Gowers}\left({#1},{#2}\right)}

\newcommand{\weyl}[1]{\textsf{Weyl}\left({#1}\right)}
\newcommand{\stabfidelity}[1]{\calF_\calS\left({#1}\right)}

\definecolor{citegreen}{HTML}{208054}
\definecolor{citeblue}{HTML}{0055cc}

\NineColors{saturation=high}
\hypersetup{
    breaklinks=true,   
    colorlinks=true, 
    linkcolor=blue3, 
    citecolor=green5, 
    urlcolor=blue3, 
}

\newtheorem{theorem}{Theorem}[section]
\newtheorem{definition}[theorem]{Definition}

\newtheorem{prop}[theorem]{Proposition}
\newtheorem{lemma}[theorem]{Lemma}
\newtheorem{remark}{Remark}
\newtheorem{corollary}[theorem]{Corollary}

\newtheorem{fact}[theorem]{Fact}
\newtheorem{claim}[theorem]{Claim}

\newtheorem{problem}[theorem]{{Problem}}
\newtheorem{result}[theorem]{{Result}}

\usepackage{thm-restate} 
\newcommand{\customlabel}[2]{%
\protected@write \@auxout {}{\string \newlabel {#1}{{#2}{}}}}
\makeatother

\global\long\def\argmax{\operatornamewithlimits{argmax}}

\newcommand{\Dpsi}{D_\Psi}
\newcommand{\bsgtest}{\hyperref[algo:bsg_test]{\textsf{BSG-TEST}}}
\newcommand{\edgetest}{\hyperref[algo:edge_test]{\textsf{EDGE-TEST}}}
\newcommand{\sample}{\hyperref[algo:sample]{\textsf{SAMPLE}}}
\newcommand{\sampleAone}{\hyperref[algo:sample_small_doubling_set]{\textsf{Sample-Small-Doubling-Set}}}
\newcommand{\pfrsubgroup}{\hyperref[algo:obtain_pfr_subgroup_V]{\textsf{PFR-Subgroup}}}
\newcommand{\findstabilizer}{\hyperref[algo:find_good_stab]{\textsf{Find-Stabilizer}}}
\newcommand{\symgramschmidt}{\hyperref[alg:sym_gram_schmidt]{\textsf{Symplectic Gram-Schmidt}}}
\makeatletter
\def\widebreve{\mathpalette\wide@breve}
\def\wide@breve#1#2{\sbox\z@{$#1#2$}%
     \mathop{\vbox{\m@th\ialign{##\crcr
\kern0.08em\brevefill#1{0.8\wd\z@}\crcr\noalign{\nointerlineskip}%
                    $\hss#1#2\hss$\crcr}}}\nolimits}
\def\brevefill#1#2{$\m@th\sbox\tw@{$#1($}%
  \hss\resizebox{#2}{\wd\tw@}{\rotatebox[origin=c]{90}{\upshape(}}\hss$}
\makeatletter

\title{Learning stabilizer structure of quantum states}

\begin{document}

\author{
Srinivasan Arunachalam\\[2mm]
IBM Quantum\\
\small Almaden Research Center\\
\small \texttt{Srinivasan.Arunachalam@ibm.com}
\and
Arkopal Dutt\\[2mm]
IBM Quantum\\
\small   Cambridge, Massachusetts\\
\small \texttt{arkopal@ibm.com}
}

\date{\today}

\maketitle

\begin{abstract}
We consider the task of learning a \emph{structured stabilizer decomposition} of an arbitrary $n$-qubit quantum state $\ket{\psi}$:  for every $\varepsilon > 0$, output a succinctly describable state $\ket{\phi}$ with stabilizer-rank $\poly(1/\varepsilon)$ such that $\ket{\psi}=\ket{\phi}+\ket{\phi'}$ where $\ket{\phi'}$ has stabilizer fidelity at most $\varepsilon$. We firstly show the existence of such decompositions using the inverse theorem for the Gowers-$3$ norm of quantum states that was recently established~\cite[STOC'25]{ad2024tolerant}. 

\vspace{2mm}

Algorithmizing the inverse theorem is key to \emph{learning} such a decomposition. To this end, we initiate the task of \emph{self-correction} of a state $\ket{\psi}$ with respect to the class of states $\C$: given copies of $\ket{\psi}$ which has fidelity $\geq \uptau$ with a state in~$\calC$, output  $\ket{\phi} \in \calC$ with fidelity $|\la \phi | \psi \ra|^2 \geq \Omega(\uptau^C)$ for some constant $C>1$. {Assuming} the \emph{algorithmic} polynomial Frieman-Rusza ($\textsf{APFR}$) conjecture in the high-doubling regime (whose combinatorial version was resolved in a recent breakthrough~\cite[Annals of Math.'25]{gowers2023conjecture}), we give a $\poly(n,1/\varepsilon)$-time algorithm for self-correction of stabilizer states. 

\vspace{2mm}

Given access to the state preparation unitary $U_\psi$ for $\ket{\psi}$ and its controlled version $\textsf{con}U_\psi$, we  give a $\poly(n,1/\varepsilon)$-time protocol that learns a structured stabilizer decomposition of $\ket{\psi}$. Without assuming $\textsf{APFR}$, we give a $\poly(n,(1/\varepsilon)^{\log 1/\varepsilon})$-time protocol for the same task. Our techniques extend to finding structured decompositions over high stabilizer-dimension states, by giving a new tolerant tester for these states.

\vspace{2mm}

As our main application, we give learning algorithms for states $\ket{\psi}$ promised to have \emph{stabilizer extent} $\xi$, given access to $U_\psi$ and $\textsf{con}U_\psi$. We give a protocol that outputs $\ket{\phi}$ which is constant-close to~$\ket{\psi}$ in time $\poly(n,\xi^{\log \xi})$, which can be improved to $\poly(n,\xi)$ assuming $\textsf{APFR}$. This gives an unconditional learning algorithm for stabilizer-rank $\kappa$ states in time $\poly(n,\kappa^{\kappa^2})$. As far as we know, efficient learning arbitrary states with even stabilizer-rank~$\kappa\geq 2$ was unknown.
\end{abstract}

\newpage

\setcounter{tocdepth}{2}
{\small \tableofcontents}

\newpage 

\part{Introduction}

\section{Introduction}\label{sec:intro}

Over the last few decades, Gowers norms for functions have played a central role  in understanding the structure of classical functions and patterns in a sequence of integers~\cite{green2008primes,roth1953certain,tao2007structure, Gowers1998,GreenTao08_U3Inverse,samorodnitsky2007low,HatamiHatamiLovett2019HOFA}.  Influential works involving Gowers norm have qualitatively shown that, if the Gowers norm of a function $f$ is large, then  $f$ is ``structured" and if the norm is small, then $f$ is ``pseudorandom". This \emph{structure vs.~randomness} paradigm (beyond Gowers norm) has been instrumental in obtaining influential results in mathematics and theoretical computer science~\cite{tao2007dichotomy,
cohen2021structure,chen2022beyond,santhanam2007circuit,oliveira2017pseudodeterministic,alweiss2020improved,alon2003testing,bhattacharyya2010optimal}.  In this theme, seminal works of Gowers, Green and Wolf~\cite{green2006montreal,gowers2010true} have shown how to ``extract structure" by decomposing an arbitrary $f$ with high Gowers-$3$ norm into a \emph{sum} of structured objects, which was subsequently algorithmized by Tulsiani and Wolf~\cite{tulsiani2014quadratic}.  Here, we consider a similar ``structure vs.~randomness" paradigm for \emph{quantum states}. 

In a recent work,~\cite{ad2024tolerant} considered the notion of Gowers-$3$ norm for a  quantum state $\ket{\psi}$  and showed $(i)$ it is closely related to the well-known \emph{stabilizer fidelity}, the maximum fidelity of $\ket{\psi}$ with stabilizer states $\Sh$ denoted $\calF_{\mathcal{S}}(\ket{\psi})$; and $(ii)$ proved an \emph{inverse theorem}, i.e., if  Gowers-$3$ norm of $\ket{\psi}$ is $\geq \tau$, then   $\calF_{\mathcal{S}}(\ket{\psi})\geq \tau^C$~\cite{mehraban2024improved,bao2024tolerant} for some $C>1$. Although inverse Gowers theorems for states~\cite{ad2024tolerant} hints at the possibility of there being ``some structure" in states with high Gowers-$3$ norm, it is far from clear if this is structure is \emph{efficiently obtainable}. This brings us to the first task that we  initiate in this work, which we call \emph{self-correction}\footnote{We chose self-correction for the name of this task, inspired by classical literature~\cite{tulsiani2014quadratic,blum1990self}. Although~self-correction has been used in the context of quantum memories~\cite{bacon2006operator,bravyi2013quantum}, we are not aware if they are related.} for stabilizer states (which is unexplored in quantum computing as far as we know):
\begin{quote}
\begin{center}
    \emph{If $\ket{\psi}$ is promised to have $\calF_{\mathcal{S}}(\ket{\psi})\geq \tau$, output a stabilizer $\ket{s}$ with $|\langle s|\psi\rangle|^2\geq \tau^C$}.
\end{center}
\end{quote}
One can view self-correction as a weaker form of \emph{agnostic learning} wherein the goal is to output a $\ket{\phi}$ for which  $|\langle \phi|\psi\rangle|\geq \tau -\varepsilon$ for $\varepsilon>0$. Agnostic learning  of stabilizers was recently considered~\cite{grewal2023improved,chen2024stabilizer} where they gave quasi-polynomial time algorithms. However, it is often not necessary to be $(\tau-\varepsilon)$-close but being $\poly(\tau)$-close (as in self-correction is sufficient), motivating our question.

Beyond the similarity to agnostic learning, a natural followup question is, can we extract \emph{further} structure, i.e., similar to the classical works~\cite{green2006montreal,gowers2010true},  can we write $\ket{\psi}$ as a sum of ``structured objects" which in this case would be \emph{stabilizer states}? In quantum computing, expressing a quantum state $\ket{\psi}$ as a  sum of stabilizer states is well-known as a \emph{stabilizer rank decomposition} and this has played an important role~\cite{bravyi2016trading,bravyi2019simulation} in understanding the simulability limits of circuits and states. In fact, one of the fastest classical simulation methods for quantum computing is based on stabilizer rank decompositions. This naturally motivates the question
\begin{quote}
\centering
\emph{Can we efficiently decompose every $\ket{\psi}$ into a structured and unstructured part?}
\end{quote}
One can envision a few applications of this: $(i)$ understanding if an unknown $\ket{\psi}$ admits a succinct description, $(ii)$ using the structured part to learn \emph{properties} of $\ket{\psi}$, and $(iii)$ learning low stabilizer-rank~states. The last application being a long-standing open question in quantum learning theory. 

\subsection{Main results}\label{subsec:main_results}
We give a positive answer to the two questions posed above: $(i)$ a \emph{polynomial-time} self-correction algorithm for stabilizers assuming the algorithmic \emph{polynomial Freiman-Ruzsa} $(\PFR)$ conjecture, and $(ii)$ a polynomial-time algorithm  for learning structured stabilizer-rank decompositions (assuming algorithmic $\PFR$ conjecture) along with a quasi-polynomial time algorithm without this assumption.   Before giving more details about these results, we discuss our conjecture in additive combinatorics.

\textbf{Polynomial Freiman-Ruzsa conjecture.} The Freiman-Ruzsa theorem \cite{freiman1987structure,ruzsa1999analog} is a cornerstone of additive combinatorics with diverse applications to theoretical computer science~\cite{lovett2015exposition}. To state the conjecture, we say a set $A$ has \emph{doubling constant} $K$ if $|A+A|\le K|A|$, where  $A+A = \{a+a' \;;\; a,a' \in A\}$. 
In this setting, the $\PFR$ conjecture states that sets $A \subseteq \FF_2^n$ with $|A+A|\le K|A|$ is covered by $\poly(K)$ translates of a subspace $V \subset \FF_2^n$ of size $|V| \leq |A|$. This conjecture was open for decades before being resolved recently in a seminal work of Gowers, Green, Manners and Tao~\cite{gowers2023conjecture}, who showed that
\begin{restatable}{theorem}{combpfrtheorem}(Combinatorial $\PFR$ theorem)
\label{thm:marton_conjecture}
Suppose $A \subseteq\mathbb{F}_2^{n}$ has doubling constant $K$, then $A$ is covered by at most $2K^{9}$ cosets of some subgroup $H \subset \textsf{span}(A)$ of size $|H| \leq |A|$.
\end{restatable}
A natural question is can this be algorithmized, i.e., can one \emph{find the subspace $V$} or output an oracle for $V$? This has not received much attention  (in part because $\PFR$ itself was open until recently). Our result relies on the fact that this combinatorial theorem can be algorithmized.
\begin{restatable}{conjecture}{algopfrconjecture}
(High-doubling algorithmic $\PFR$ conjecture)
\label{conj:algopfrconjecture}
Let $K \geq 1$. Suppose  $A \subseteq\mathbb{F}_2^{2n}$ has doubling constant~$K$. Given random samples from $A$ and membership oracle for $A$ (i.e., on input $x$ outputs if $x\in A$ or not), there is a $\poly(n,K)$-time  procedure that outputs a membership oracle for the subgroup $H$ (whose size is at most $|A|$), such that $A$ is covered by $\poly(K)$-many cosets of~$H$.
\end{restatable}
Recently, \cite{algopfr} gave a $\poly(n,2^K)$-time algorithm for this task. In additive combinatorics, $K$ is generally assumed to be a constant (or at most poly-logarithmic in $n$), in which case the result is efficient and settles the above conjecture in the small-doubling regime. In quantum learning, we are interested in the high-doubling regime of $K=\poly(n)$, in which case this conjecture is open. \textcolor{black}{We believe this regime will also spur new ideas for structurally understanding these  high-doubling sets in additive combinatorics.}  We refer to the open questions section for further discussion.
\paragraph{Our results.} We are now ready to state our results. Along the way we briefly discuss their significance and context, before giving the proof outlines after that.

\noindent \textbf{\emph{1. Self correction}.} Our first result is a  polynomial-time algorithm that solves the self-correction task that we defined above. 
Throughout the paper, we  define the \emph{stabilizer fidelity} of a state $\ket{\psi}$ as its maximal fidelity over all stabilizer states, i.e.,
$$
\stabfidelity{\ket{\psi}}=\max_{\ket{\phi}\in \textsf{Stab}}\{|\langle \phi|\psi\rangle|^2\},
$$
where $\textsf{Stab}$ is the class of $n$-qubit stabilizer states. 
\begin{restatable}{theorem}
{selfcorrectionstatement}
\label{thm:self_correction} ($\Selfcorrection$)
Let $\uptau>0$. Let~$\ket{\psi}$ be an unknown $n$-qubit quantum state such that $\stabfidelity{\ket{\psi}} \geq \uptau$. Assuming the high-doubling algorithmic $\PFR$ Conjecture~\ref{conj:algopfrconjecture}, there is a protocol that with probability $1-\delta$, outputs a $\ket{\phi}\in \Sh$ such that $|\langle\phi |\psi\rangle|^2 \geq  \tau^C$ (for a universal constant $C>1$) using $\poly(n,1/\uptau,\log(1/\delta))$ time and copies of~$\ket{\psi}$.
\end{restatable}

We provide the formal statement of this result in Section~\ref{sec:selfcorrectionprotocol}. In the table below, we compare our work with the state-of-the-art results in this direction.
\begin{table}[!ht]
\small
\centering
\begin{tabular}{|c | c  | c  |c|} 
\hline
\makecell{$\calF_\calS(\ket{\psi}) = \opt$} & \makecell{Output} & \makecell{Access\\ model} &\makecell{Time\\ complexity} \\ [0.5ex] 
\hline
\makecell{\cite{grewal2023improved}\\$\opt \geq \tau$} &  \makecell{Stabilizer  state\\[1mm] with fidelity $\geq \tau-\varepsilon$}  &  Sample access & $2^{n/\tau^4}/\varepsilon^2$  \\ 
\hline
\makecell{\cite{chen2024stabilizer}\\ $\opt \geq \tau$} &  \makecell{Stabilizer state\\[1mm] with fidelity $\geq \tau-\varepsilon$}  &  Sample access & $\poly(n,1/\varepsilon)\cdot  (1/\tau)^{\log (1/\tau)}$  \\ 
\hline
\makecell{\emph{This work}$^*$\\ $\opt \geq \tau$} &  \makecell{Stabilizer state \\[1mm] with fidelity $\geq \tau^C$}  &  Sample access & $\poly(n,1/\tau)$  \\ 
\hline
\end{tabular}
\caption{Summary of results. In the access model, sample access refers to an algorithm which is only given copies of $\ket{\psi}$ and $U_\psi,\textsf{con}U_\psi$ access refers to an algorithm which is given access to state preparation unitary $U_\psi$ (and its controlled version $\textsf{con}U_\psi$) such that $U_\psi\ket{0^n}=\ket{\psi}$. The $^*$ on this work is to indicate that it relies on the high-doubling algorithmic $\PFR$ conjecture.} 
\label{tab:summary_results_paper}
\end{table}
Like we mentioned in the introduction, in comparison to the harder task of agnostic learning~\cite{grewal2023improved,chen2024stabilizer} for which we have a quasi-polynomial algorithm, our work can be viewed as answering a weaker question, but on the flip-side we are able give a polynomial-time algorithm for this task (assuming the high-doubling algorithmic $\PFR$ conjecture).  Furthermore, we remark that our algorithm involves only single-copy and two-copy measurements of the unknown quantum state. 

\noindent \textbf{\emph{2. Structured decomposition}.} 
In order to prove our self-correction result, our main contribution is to algorithmize the inverse Gowers-$3$ theorem of quantum states that was established in~\cite{ad2024tolerant}. In additive combinatorics and higher-order Fourier analysis, inverse theorems often imply \emph{decomposition theorems}~\cite{green2006montreal,gowers2010true,gowers2011linearFpn,gowers2011linearZn} where the goal is to express a bounded function $f:\FF_2^n \rightarrow \mathbb{C}$ as
$$
f = \sum_{i=1}^k \beta_i g_i + f_{\text{unstruct}}
$$
where $g_i$ for all $i \in [k]$ are \emph{succinctly} representable functions from a defined class of functions $\calF$, $\beta_i\in \mathbb{C}$ and and $f_{\text{unstruct}}$ is nearly orthogonal to all functions in $\calF$ i.e., $\la f, g \ra \leq \varepsilon$ for all $g \in \calF$ (for some error parameter $\varepsilon > 0$). Notably, the seminal Goldreich-Levin algorithm~\cite{goldreich1989decomp} works with $\calF$ being the set of parity functions and outputs in $\poly(n)$ time all the large Fourier characters of $f$ and $f_{\text{unstruct}}$ which has low correlation with any parity function. In quadratic Fourier analysis, the inverse theorem of the Gowers-$3$ norm of Boolean functions~\cite{samorodnitsky2007low,GreenTao08_U3Inverse} is leveraged to obtain a structured decomposition of $f$ in terms of quadratic phase polynomials, and an unstructured part involving a function with low $\ell_1$ norm and another which has low correlation with any quadratic function. From an algorithmic point of view, these proofs were originally algorithmized by Tulsiani and Wolf~\cite{tulsiani2014quadratic}, and improved recently by Briet and Castro-Silva~\cite{briet2025near} who gave a classical algorithm for this task, surprisingly via \emph{dequantized} stabilizer bootstrapping~\cite{chen2024stabilizer}!

Taking inspiration from these results, the natural question  in the quantum setting is: Can we use the inverse theorem of Gowers-$3$ norm of quantum states to give a structured stabilizer \emph{decomposition theorem} for quantum states and furthermore does it admit an \emph{efficient} {algorithm}? Our second main result answers this affirmatively by showing that any self-corrector of stabilizer states with runtime $T$ can be iteratively applied to an arbitrary state $\ket{\psi}$ to ``find" a decomposition of $\ket{\psi}$ in terms of a structured part and an unstructured part in time $\poly(n,T)$. In order to perform this recursive procedure, we will require access to the state preparation unitary for the unknown state $\ket{\psi}$. Throughout this paper, for an $n$-qubit quantum state $\ket{\psi}$, we let $U_\psi$ be a unitary that prepares $\ket{\psi}$ and $\textsf{con}U_\psi$ be the controlled-version~of~$U_\psi$. We now state this result more formally.

\begin{result}
\label{result:learn_decomposition}
\textcolor{black}{ Let $\calA$ be an algorithm such that: for every $\varepsilon$, given copies of an unknown $n$-qubit state $\ket{\varphi}$ with $\stabfidelity{\ket{\varphi}} \geq \varepsilon$, outputs a stabilizer state $\ket{\phi}$ such that $|\la \phi| \varphi \ra|^2 \geq \eta(\varepsilon)$.
Then, there is an algorithm $\calA'$ such that: for every $\ket{\psi}$, given access to $U_\psi,\textsf{con}U_\psi$, invokes $\A$ $k\leq O(1/\eta(\varepsilon)^2)$ many times and outputs 
$\beta \in \calB_\infty^k,\alpha \in \calB_\infty$,\footnote{Here $\calB_\infty^k$ refers to the unit ball, i.e., $\calB_\infty^k=\{a\in \mathbb{C}^k:|a_i|\leq 1 \text{ for all } i \in [k]\}$.} stabilizers $\{\ket{\phi_i}\}_{i\in [k]}$ such that one can write $\ket{\psi}$ as}
$$
  \ket{\psi}=\sum_{i\in [k]} \beta_i \ket{\phi_i}+\alpha\ket{\phi^\perp},
$$
where the residual state $\ket{\phi^\perp}$ satisfies $|\alpha|^2\cdot \stabfidelity{\ket{\phi^\perp}} < \varepsilon$.
\end{result}
The above result can be viewed as a quantum analogue of quadratic Fourier analysis in additive combinatorics. We emphasize that the result implicitly shows: $(i)$ the \emph{existence} of a stabilizer decomposition for every quantum state with non-negligible stabilizer fidelity, and $(ii)$ gives an algorithm that outputs this structure. 
The above result can be applied onto arbitrary base~subroutines~$\calA$ such as the $\Selfcorrection$ algorithm of Theorem~\ref{thm:self_correction} or the algorithm of Chen et al.~\cite{chen2024stabilizer}. We informally give the consequence below (stating it formally in Section~\ref{sec:iteratedselfcorrection}).
\begin{corollary}
For an unknown $n$-qubit state $\ket{\psi}$, given access to $U_\psi,\textsf{con}U_\psi$, there is a
\begin{enumerate}[$\qquad 1.$]
    \item $\poly(n,\varepsilon)$ algorithm using $\Selfcorrection$ (assuming $\APFR$~conjecture~\ref{conj:algopfrconjecture}), 
    \item $\poly(n,(1/\varepsilon)^{\log 1/\varepsilon})$-time algorithm using stabilizer bootstrapping,
\end{enumerate}
that outputs  
$\beta \in \calB_\infty^{k},\alpha \in \calB_\infty$ and stabilizer states $\{\ket{\phi_i}\}_{i\in [k]}$ such that one can write $\ket{\psi}$ as
$$
\ket{\psi}=\sum_{i\in [k]} \beta_i \ket{\phi_i}+\alpha\ket{\phi^\perp},
$$
where $k \leq \poly(1/\varepsilon)$ and the residual  state $\ket{\phi^\perp}$ satisfies $|\alpha|^2\cdot \stabfidelity{\ket{\phi^\perp}} \leq \varepsilon$. 
\end{corollary}

\paragraph{\emph{3. Learning states with low stabilizer extent.}} The notion of \emph{stabilizer extent} was introduced by Bravyi et al.~\cite{bravyi2019simulation} as a way to operationally quantify the ``non-stabilizerness" of a quantum state. On a high-level, the stabilizer extent of a quantum state $\ket{\psi}$ is defined as the minimal $\sum_i |c_i|$ when optimized over all possible stabilizer decompositions of $\ket{\psi}=\sum_i c_i\ket{\phi_i}$ where $\ket{\phi_i}$ are stabilizer states. \textcolor{black}{ These results gave simulation algorithms for $\ket{\psi}$ with runtimes scaling polynomially in the stabilizer extent of $\ket{\psi}$. A natural question is, given low stabilizer extent states are simulatable efficiently, are they also \emph{learnable} efficiently? This has been the primary quest of a sequence of works (which we discuss below)~\cite{aaronsontalk,montanaro2017learning,arunachalam2022optimal,grewal2023efficient,leone2024learningstab,hangleiter2024bell,leone2024learning}.}  
Our main application answers this positively by giving a learning algorithm for states with bounded stabilizer~extent. 

\begin{result}\label{result:learn_stab_extent}
Given access to $U_\psi,\textsf{con}U_\psi$ for an unknown $\ket{\psi}$ with stabilizer extent $\xi$, there is a
\begin{enumerate}[$\qquad 1.$]
\item $\poly(n,\xi)$ algorithm using $\Selfcorrection$ (assuming $\APFR$~conjecture~\ref{conj:algopfrconjecture}),
\item $\poly(n,\xi^{\log \xi})$ algorithm using stabilizer bootstrapping,
\end{enumerate}
which outputs $\ket{\phi}$ that is close to $\ket{\psi}$ upto constant trace distance.
\end{result}
We provide the formal statement of this result in Section~\ref{sec:learn_stab_extent}. \textcolor{black}{ Prior to this we aren't aware even of quasi-polynomial algorithms for learning these states (even if given access to $U_\psi,\textsf{con}U_\psi$); and assuming $\APFR$ Conjecture~\ref{conj:algopfrconjecture}, we resolve the main open question asked by Grewal et al.~\cite{grewal2022low}.}

Additionally, using recent results relating stabilizer extent and stabilizer rank, we also give learning algorithms for states with bounded \emph{stabilizer rank}. Learning quantum states with stabilizer structure has a rich history. The first works on learning stabilizer states were by~\cite{aaronsontalk,montanaro2017learning} who showed that stabilizer states were learnable efficiently. Subsequently, Lai and Cheng~\cite{lai2022learning} showed that states prepared by Clifford circuits with one layer of $O(\log n)$ many $T$ gates can be learned efficiently. A sequence of works has since appeared considering the more general question of learning states produced by arbitrary Clifford circuits with $O(\log n)$ many $T$ gates~\cite{grewal2023efficient,leone2024learningstab,hangleiter2024bell,leone2024learning} and have been able to give efficient algorithms for this task. The key insight into these algorithms is that the states produced by a Clifford circuit with $t$ many $T$ gates are stabilized by an Abelian group of $2^{n-2t}$ Paulis or in other words, have a \emph{stabilizer dimension} of at least $n-2t$. These learning algorithms then exploit this structure. Such states are known to have  stabilizer extent at most $\poly(n)$~\cite{grewal2022low}. 

However, $n$-qubit states with stabilizer extent $\poly(n)$ are more general and could have $\poly(n)$ stabilizer dimension, leading to exponential-time learning algorithms using~\cite{grewal2023efficient}.\footnote{An example is as follows: for $m\geq 1$, let $\mathcal{C}=\{\ket{0}_{S}\otimes \ket{W_m}_{\ket{\bar{S}}}\}_{|S|=m}$, where $\ket{W_m}=\frac{1}{\sqrt{m}}\sum_{i=1}^m\ket{e_i}$ and $e_i\in \FF_2^m$ is the  basis state with $1$ in the $i$th coordinate and $0$s elsewhere. The stabilizer extent of $\ket{\psi}\in \mathcal{C}$ equals $\sqrt{m}$ but the stabilizer dimension of these states is $n-m+1$ (the stabilizers are $(\{\Id,Z\}^{S}\otimes \{\Id^{\bar{S}},Z^{\bar{S}}\})$ which has dimension~$n-m+1$.}
It is then desirable to have efficient  learning algorithms for states with bounded stabilizer extent, an outstanding question in quantum learning,  explicitly raised in~\cite{arunachalam2022optimal,grewal2022low,anshu2024survey}. \textcolor{black}
{Here, we give a polynomial-time learning algorithm using the state preparation unitary and its controlled version, assuming $\APFR$ Conjecture~\ref{conj:algopfrconjecture}, along with an \emph{unconditional} quasi-polynomial time algorithm.
For stabilizer-rank $k$ states, we use the recently established result of $\xi\leq k^k$~\cite{kalra2025stabilizer} to give an unconditional $\poly(n,k^{k^2})$ learning algorithm for stabilizer-rank $k$ states. Prior to our work, we were not aware of efficient learning algorithms for even states promised to have stabilizer~rank~$2$!} 

\paragraph{\textbf{\emph{4. Further implications}}} We have two further applications of our main result.\\[2mm]
\textbf{\emph{$(i)$}} \textcolor{black}{\textbf{\emph{Estimating stabilizer-extent fidelities}}.} A natural task considered in~\cite{flammia2011direct,aaronson:shadow,huang2020predicting} was the estimation of fidelities of an unknown quantum state $\ket{\psi}$ with classes of succinctly describable states. The classical shadows protocol of \cite{huang2020predicting} can be used to perform these predictions efficiently. However,  can one \emph{extract a succinctly describable} ``mimicking state" $\ket{\sigma}$ such that this state can be used to obtain an $\varepsilon$-approximation of the overlap of the unknown $\ket{\psi}$ with any low stabilizer-extent state $\ket{s}$ i.e., $|\langle s | \sigma \rangle - \langle s | \psi \ra| \leq \varepsilon$? One can view our iterative self-correction result as saying: for every $\ket{\psi}$, given access to $U_\psi,\textsf{con}U_\psi$, one can extract a  succinctly describable``mimicking state"  $\ket{\sigma} = \sum_i \beta_i\ket{\phi_i}$ in $\poly(n)$-time (assuming $\APFR$ Conjecture~\ref{conj:algopfrconjecture}) and quasi-polynomial otherwise.
Recent works of~\cite{king2024triply,chen2024optimal,tran2025one}  considered the task of preparing a mimicking state for an arbitrary quantum state in order to estimate $\langle \psi|P|\psi\rangle$ for all Paulis $P$~\cite{huang2020predicting}. Unfortunately, the works of~\cite{king2024triply,tran2025one,chen2024optimal} construct the mimicking state via a feasibility problem over the convex set of $n$-qubit states which may take \emph{exponential time}. In contrast, for the task of estimating inner products of an unknown state $\ket{\psi}$ with low stabilizer-extent states, we give a procedure that is efficient (assuming $\APFR$ Conjecture~\ref{conj:algopfrconjecture}).  We give more details in Section~\ref{sec:mimickingfidelity};\\[2mm]
\textbf{\emph{$(ii)$}} \textcolor{black}{\textbf{\emph{Testing and learning more general states}}.}  We also consider  states of the form $\ket{\psi}=\sum_i c_i\ket{\phi_i}$ where $\ket{\phi_i}$ has stabilizer dimension $n-t$ (this generalizes the usual definition of stabilizer rank where $t=0$). We give learning algorithms for these quantum states using the same access model as above running in $\poly(n,2^t)$-time (assuming $\APFR$ Conjecture~\ref{conj:algopfrconjecture}). This is achieved by algorithmizing a new local Gowers-$3$ inverse theorem which we show here: If the Gowers-$3$ norm of $\ket{\psi}$ is $\geq \gamma$ then it has $\geq \poly(\gamma)$ fidelity with an $n-O(\log(1/\gamma))$ stabilizer dimension state. As a result, we  give a tolerant tester for states with high stabilizer dimension which extends \cite{ad2024tolerant}.
\vspace{-1.5mm}
\subsection{Proof sketch of self-correction}\label{sub_sec:proof_sketches}
In this section, we will give a proof sketch of our $\Selfcorrection$ result.
The starting point of proving Theorem~\ref{thm:self_correction} is the tolerant testing of stabilizer states' algorithm presented in~\cite{ad2024tolerant}, where they showed an inverse theorem for the Gowers-$3$ norm of quantum states. Below, we will state our main results in the context of stabilizer fidelity (for simplicity in presentation) which we know is polynomially related to Gowers-$3$ norm by~\cite{ad2024tolerant}. Before that, we define some notation: for a quantum state $\ket{\psi}$, the characteristic distribution is defined as $p_\Psi(x)=2^{-n}\cdot |\langle \psi|W_x|\psi\rangle|^2$, and the Weyl distribution is defined as its convolution $q_\Psi=4^n (p_\Psi\star p_\Psi)$. It is well-known that Bell-difference sampling~\cite{gross2021schur} allows us to sample from $q_\Psi$ and estimate the expression $\Exp_{x\sim q_\Psi}[|\langle \psi|W_x|\psi\rangle|^2]$. 
\begin{restatable}[\cite{ad2024tolerant}]{theorem}{gowersstates}\label{thm:inversegowersstates}
    Let $\gamma\in [0,1]$. If $\ket{\psi}$ is an $n$-qubit state such that $\Exp_{x\sim q_\Psi}[2^np_\Psi(x)] \geq \gamma$, then there is an $n$-qubit stabilizer state $\ket{\phi}$ such that $|\la \psi | \phi \ra|^2 \geq \Omega(\gamma^C)$ for some constant $C>1$.
\end{restatable}

The main contribution here is to algorithmize the proof of the above theorem in order to \emph{find} a stabilizer state with the guarantee as stated, and thereby obtain our self-correction protocol. We emphasize that \emph{almost all} steps involved in the proof of the theorem above were existential and algorithmizing each one of them is the main technical work involved in the $\Selfcorrection$ algorithm. We describe the algorithm by first discussing the algorithmic primitives that we use. 
\vspace{-1.5mm}
\subsubsection{Algorithmic components}
In order to discuss the algorithmic subroutines, we first give a high-level idea of the proof of Theorem~\ref{thm:inversegowersstates} and along the way motivate these  subroutines. 

\textbf{\emph{(i) Bell sampling.}} 
In~\cite{ad2024tolerant}, they observed that if $\Exp_{x\sim q_\Psi}[|\langle \psi|W_x|\psi\rangle|^2] \geq \gamma$, then there is a large approximate subgroup $S \subseteq \{x \in \FF_2^{2n}: 2^n p_\Psi(x) \geq \gamma/4\}$ with size $|S|\in [\gamma/2 , 2/\gamma ]\cdot 2^n$~satisfying
$$
\Pr_{x,y \in S}[ x+y \in S]\geq \poly(\gamma).
$$
Note that for stabilizer states, $\gamma=1$, and $S$ would be the stabilizer subgroup, but in~\cite{ad2024tolerant} they showed that a $\gamma$-lower bound implies the existence of a large approximate group $S$. Unfortunately, $S$ is exponentially sized which is too expensive to store in memory. Our first simple observation is that the well-known Bell difference sampling subroutine can be utilized to sample elements from $S$ efficiently and with $\poly(\gamma)$ probability, so we do have ``access" to this approximate subgroup $S$ via samples but do not have an explicit description of $S$. From here onwards, we \emph{condition} on the event of sampling from $S$ and we will crucially use the fact that we are sampling from this corresponding distribution in order to establish a few concentration properties, which we discuss next.
     
\textbf{\emph{(ii) $\BSG$ Test.}}
The next step in the proof of~\cite{ad2024tolerant} was to apply the Balog-Szemeredi-Gowers ($\BSG$) theorem~\cite{balog1994statistical,gowers2001new} to $S$, which implies the existence of a large subset $S' \subseteq S$ which has a small doubling constant. In particular,~\cite{ad2024tolerant} showed that 
$$
|S' + S'| \leq \poly(1/\gamma) |S|, \enspace |S'| \geq \poly(\gamma)\cdot  |S|.
$$
Again since $S'$ is exponentially large, we will not hold this nor construct it. It would be ideal to have access to this $S'$ and for this one would like to algorithmize the $\BSG$ theorem. Our main contribution here will be to describe a membership test for $S'$. Along with the sampling subroutine from $S$ in $(i)$, this allows us to obtain samples from $S'$. To achieve this we use the main ideas by Tulsiani and Wolf~\cite{tulsiani2014quadratic} who showed how to algorithmize the $\BSG$ theorem for Boolean functions. There are two fundamental issues in using~\cite{tulsiani2014quadratic} as a blackbox which we state first:  $(i)$ in~\cite{tulsiani2014quadratic}, they fix a good \emph{choice function} $\phi:\FF_2^n\rightarrow \FF_2^n$ and define the set $S_\phi=\{(x,\phi(x)):x\in \FF_2^n\}$ for which they prove properties about. However for us, our set $S$ need not necessarily have this product structure. 
$(ii)$ More importantly, in~\cite{tulsiani2014quadratic} once they have $\phi$, one can sample from the set $S_\phi=\{(x,\phi(x)):x\in \FF_2^n\}$ \emph{uniformly} by first sampling $x\in \FF_2^n$ and then outputting $(x,\phi(x))$. In contrast, the samples that we obtain in our choice set, i.e., the approximate subgroup $S$, comes from the Bell difference sampling distribution, requiring us to reanalyze all their claims.

Our main contribution is, we are able to  show that all their arguments also goes through when we are handed samples from the Weyl distribution $q_\Psi$ corresponding to Bell difference sampling. In order to port the uniform-distribution arguments (conditioned on sampling from the approximate subgroup $S$) from~\cite{tulsiani2014quadratic} to ours where we sample from the Weyl distribution, we open up several  arguments in~\cite{sudakov2005question,tulsiani2014quadratic} and use the analytic properties of the graphs used in Balog-Szemeredi~\cite{balog1994statistical} to show new concentration inequalities of the Weyl distribution along the way (which may be of independent~interest). Crucial to the the proof of~\cite{tulsiani2014quadratic} is that every point in the choice set has a distributional weight of exactly $1/|S|$ (since it is the uniform distribution), whereas for us, we work with the Weyl distribution which does not have this property. Note that if one had access to copies of the state's conjugate $\ket{\psi^*}$, then we could have obtained the choice set via Bell sampling whose distribution corresponds to $p_\Psi$ and then, we would have the promise that $p_\Psi(x) \geq \gamma\cdot 2^{-n},\,\, \forall x \in S$. However, as having access to $\ket{\psi^*}$ is unnatural, we limit ourselves to having only access to samples from the Weyl distribution $q_\Psi$. Proving analogous statements under the latter distribution is significantly more challenging due to the lack of lower bounds on $q_\Psi(x)$ for all $x\in S$, and circumventing this is the main technical contribution of our~ $\BSG$~test.
 
\textbf{\emph{ (iii) Algorithmic $\PFR$ conjecture.}}
The next step in~\cite{ad2024tolerant} is to apply the recently proven $\PFR$ theorem~(Theorem~\ref{thm:marton_conjecture})~\cite{gowers2023conjecture} to~$S'$. Using this, one can observe that $S'$ is covered by few translates of a subgroup $V \subseteq S'$ and with some analysis,~\cite{ad2024tolerant} showed that 
$$
\Exp_{x \in V}\left[2^n p_\Psi(x) \right] \geq \poly(\gamma).
$$
In~\cite{ad2024tolerant}, the structural statement of the $\PFR$ statement was sufficient. As in the previous steps, we would ideally like an algorithmic version of~\cite{gowers2023conjecture} but this seems non-trivial to obtain.
Hence, we assume the algorithmic $\PFR$ conjecture in the high-doubling regime (Conjecture~\ref{conj:algopfrconjecture}). 
Since we have already constructed sample and membership access to $S'$ as described in $(i,ii)$, we will later show that $\APFR$ Conjecture~\ref{conj:algopfrconjecture} can be utilized to construct a basis for $V$.
    
\textbf{\emph{(iv) Clifford synthesis for subgroup transformation.}} In \cite{ad2024tolerant}, it was then shown that the subgroup $V$ can be covered by a set of $\poly(1/\gamma)$ unsigned stabilizer subgroups. This allowed them to conclude that there is a stabilizer state $\ket{\phi}$ such that $|\la \psi | \phi \ra|^2 \geq \poly(\gamma)$. We now aim to algorithmize this: \emph{determine} the stabilizer group and thus the stabilizer state.
To aid us, we seek to find a Clifford unitary $U$ that 
\begin{align}
\label{eq:introsgs}
U V U^\dagger = \la Z_1, X_1, \ldots, Z_k, X_k, Z_{k+1}, Z_{k+2}, \ldots, Z_{k+m} \ra = \calP^k \times \calP_\calZ^m,
\end{align}
where $\calP^k$ is the $k$-qubit Pauli group. Note that there exist algorithms for synthesizing such a unitary when $V$ is an isotropic subspace~\cite{grewal2023efficient} but the $V$ in question here is product of an isotropic subspace and the $k$-qubit Pauli group. We use the following procedure for synthesizing $U$. We first use the \emph{symplectic} Gram-Schmidt procedure  to obtain the centralizer $C_V$ of $V$ which is the subset of Paulis in $V$ which commute with all of $V$ and the anti-commutant $A_V = V \backslash \la C_V \ra$. We then show that given bases of $C_V$ and $A_V$, one can determine a Clifford circuit with $O(n^2)$ gates implementing $U$ in $O(n^3)$ time. There are still three issues which we need to tackle in order to find the stabilizer state, $(a)$ the stabilizer group is not-fully specified (since $k+m\leq n$), $(b)$ the group above is an \emph{unsigned} stabilizer covering and $(c)$ even if we brute force over the first $k$-qubits (which is admissible since $\exp(k)=\poly(1/\gamma)$), it is unclear how one should view the last $n-k$ bits. We describe how these are circumvented in the next section.

We are done describing the algorithmic primitives of $\Selfcorrection$. Essentially, with  these subroutines we have avoided using exponential memory (and time) in \emph{storing} the approximate subgroups crucial in the proof of Theorem~\ref{thm:inversegowersstates} and instead have oracle access to these~sets.

\subsubsection{$\Selfcorrection$ algorithm}
Describing the algorithmic subroutines by themselves does not immediately give a $\Selfcorrection$ protocol. We describe in three steps below how these subroutines are combined together.

$(i)$ First, given copies of $\ket{\psi}$, the algorithm uses Bell difference sampling combined with the $\BSG$ test to obtain samples from a set $S'$ which has a small doubling constant.\footnote{This is not quite what the $\BSG$ test guarantees, but for the proof overview this picture simplifies the presentation.} As mentioned earlier, proving the correctness  requires reanalyzing~\cite{tulsiani2014quadratic} when given samples from this~distribution.

$(ii)$ The goal is to  construct a basis for the subspace $V$ whose translates covers $S'$. Now, the algorithmic $\PFR$ Conjecture~\ref{conj:algopfrconjecture} only gives us \emph{membership oracle} access, as is common in additive combinatorics and also we know only translates of $V$ covers $S'$, but we need a subspace that captures a large fraction of the ``weight'' of the Weyl distribution $p_\Psi$. The algorithm now samples $t = O(n^2)$ points from this set $S'$  (according to the Weyl distribution) and takes their span and one can show that this span covers at least a $\gamma$-fraction of elements of $S'$. Using a membership oracle by algorithmic $\PFR$, we can retain the points which are in the subspace. After this we use several linear algebraic observations to obtain a subgroup $V$ that is small subgroup and satisfies $\Exp_{x\in V}[2^n p_\Psi(x)]\geq \gamma$. The core of this  step crucially gives the  combinatorial-to-algebraic bridge in going from the the dense but unstructured set $S'$ to a succinctly describable~subgroup.

$(iii)$ Once the subgroup $V$ is known, the goal is to construct an  stabilizer state $\ket{\phi}$ with  $|\langle \psi|\phi\rangle|^2\geq \poly(\gamma)$. We apply the symplectic Gram--Schmidt procedure to $V$ to get a Clifford
unitary $U$ as in Eq.~\eqref{eq:introsgs}
where $k = O(\log(1/\gamma))$ and $k+m \leq n$. 
 One can  show that $UVU^\dagger$ can be covered using $\poly(1/\gamma)$ unsigned stabilizer groups corresponding to mutually unbiased bases (MUBs) on the first $k$-qubits and $\{I,Z\}^{\otimes (n-k)}$ on the last $(n-k)$ qubits. Moreover, due to the promise of a stabilizer state having high fidelity with $\ket{\psi}$, we observe that there is an $n$-qubit stabilizer \emph{product} state promised to have $\geq \poly(\gamma)$ fidelity with $U\ket{\psi}$, of the form $\ket{\varphi_z}\otimes\ket{z}$ where $\ket{\varphi_z}$ is a $k$-qubit stabilizer corresponding to an MUB and $\ket{z}$ is an $(n-k)$-qubit basis state. To determine this state, one option is to then construct all the stabilizer states of this form with all possible phase choices $\{-1,1\}^{2n}$ for the different generators. However, this is an exponentially large set of candidate states. We instead enumerate over all possible phases over the first $k$ qubits corresponding to the MUBs leading to a  set of $\poly(1/\gamma)$ candidate states for $\ket{\varphi_z}$. We then show that $\ket{\varphi_z}$ can be determined  by measuring the first $k$ qubits of $U \ket{\psi}$ in the corresponding basis and conditioned on this, the computational basis state $\ket{z}$ can be  obtained by measuring the last $(n-k)$ qubits of $U\ket{\psi}$ in the computational basis. We then estimate its fidelity with $U\ket{\psi}$ using classical shadows and output $U^\dagger(\ket{\varphi_{z^*}\otimes \ket{z^*}})$ where the stabilizer state $\ket{\varphi_{z^*}}\otimes \ket{z^*}$ has the highest fidelity with $U\ket{\psi}$.

\subsection{Proof sketch of iterated Self-Correction}
So far, we saw that $\Selfcorrection$, on input $\ket{\varphi}$ satisfying $\stabfidelity{\ket{\varphi}}\geq \gamma$, produced a stabilizer state $\ket{\phi_1}$ such that $|\langle \varphi|\phi_1\rangle|^2\geq \gamma^C$.  A natural question is, can we recurse this further to find \emph{another} stabilizer state $\ket{\phi_2}$ such that $|\langle \varphi|\big(\beta_1\ket{\phi_1} + \beta_2\ket{\phi_2}\big)\rangle|^2\geq \gamma'$ (for some $\beta_1,\beta_2 \in \mathbb{C}$) such that $\gamma'\geq \gamma^C$? This is precisely what we do below by giving an algorithm that learns a structured stabilizer decomposition of an arbitrary state $\ket{\psi}$.
In particular, given access to $U_\psi,\textsf{con}U_\psi$, we output a structured decomposition of $\ket{\psi}$ where the structured part is a stabilizer-rank $k\leq 1/\eta^2$ state and the unstructured part has low stabilizer fidelity.
To this end, let us denote the $\Selfcorrection$ algorithm as $\calA$ and denote the iterative learning algorithm as $\calL$. In the first step, $\calL$ checks if $\Exp_{x \sim q_\Psi}[|\la \psi | W_x | \psi \ra|^2] \geq  \varepsilon^6$ (which is a proxy for the Gowers-$3$ norm) or not by using the protocol of \cite{ad2024tolerant}. If so, $\calL$ proceeds and if not, it stops and outputs $\ket{\psi}$. When $\calL$ proceeds, it uses $\calA$ on copies of $\ket{\psi}$ to learn $\ket{\phi_1} \in \Sh$ such that $| \la \phi_1 | \psi \ra|^2 \geq \eta$ (where $\eta = \varepsilon^{6C}$). We can then~write
$$
\ket{\psi} = c_1 \ket{\phi_1} + r_1 \ket{\phi_1^\perp},
$$
where $\ket{\phi_1^\perp}$ is a state orthogonal to $\ket{\phi_1}$, and $c_1$ (resp.~$r_1$) are the corresponding coefficients  of $\ket{\phi_1}$ (resp.~$\ket{\phi_1^\perp}$). Note that these coefficients satisfy
$
c_1 = \la \phi_1 | \psi \ra$ and $|r_1| = \sqrt{1 - |c_1|^2}$. 
Let us denote the new state of interest as $\ket{\psi_2}=\ket{\phi_1^\perp}$ which is nothing but
$$
\ket{\psi_2} = \left(\ket{\psi} - c_1 \ket{\phi_1} \right)/|r_1|,
$$
In the second step of $\calL$, we first check if $|r_1| < \varepsilon$ which would indicate we have accomplished state tomography of $\ket{\psi}$ with $\ket{\phi_1}$. If not, we then proceed to prepare the state $\ket{\psi_2}$ via linear combination of unitaries ($\LCU$) to implement $(U_\psi - c_1 W_1)/|r_1|$ where $W_1$ is the Clifford unitary preparing the stabilizer state $\ket{\phi_1}$. It can be shown that the corresponding success probability of preparation is high if $|r_1|\geq \varepsilon$. At this point, we estimate the Gowers-$3$ norm of $\ket{\psi_2}$: $(i)$ if it is at most $\varepsilon^6$, we terminate and output $\ket{\phi^\perp}$ in our theorem statement as $\ket{\psi_2}$, and $(ii)$ if  Gowers-$3$ norm of $\ket{\psi_2}$ is high i.e., $\Exp_{x \sim q_{\Psi_2}}[|\la \psi_2 | W_x | \psi_2 \ra|^2] \geq  \varepsilon^6$, we run $\Selfcorrection$ on $\ket{\psi_2}$ using $\calA$ to output a stabilizer state $\ket{\phi_2}$ such that $|\la \psi_2 | \phi_2 \ra|^2 \geq \eta$. At this point, one can then express
$$
\ket{\psi_2} = \ket{\phi_1^\perp} = c_2 \ket{\phi_2} + r_2 \ket{\phi_2^\perp},
$$
where $c_2 = \la \phi_2 | \phi_1^\perp \ra$ with the promise that $|c_2|^2 \geq \eta$, and $\ket{\phi_2^\perp}$ is some state that is orthogonal to $\ket{\phi_2}$. We have so far then expressed our original state $\ket{\psi}$ as
$$
\ket{\psi} = c_1 \ket{\phi_1} + r_1 c_2 \ket{\phi_2} + r_1 r_2 \ket{\phi_2^\perp},
$$
where  we have the promise that the coefficients $|c_1|^2 \geq \eta$ and $|r_1|^2|c_2|^2 \geq \eta^2$. The above equation illustrates the beginning of learning a {structured} decomposition of $\ket{\psi}$ in terms of stabilizer states $\ket{\phi_1}$ and $\ket{\phi_2}$.  Recursing this procedure, we produce a list of stabilizer states $\ket{\phi_i}$, $\beta_i\in \mathbb{C}$ where
$$
\ket{\psi}=\sum_{i=1}^k \beta_i \ket{\phi_i}+ \alpha_{k+1}\ket{\phi^\perp}.
$$
The algorithm stops when either one of the following two conditions are met, either $\prod_{i=1}^t |r_t|^2 < \varepsilon$ or if $\Exp_{x \sim q_{\Psi_t}}[|\la \psi_t | W_x | \psi_t \ra|^2] <\varepsilon^6$. This in turn implies that the residual state $\alpha_{k+1}\ket{\phi^\perp}$ satisfies the guarantee of the theorem.  
For simplicity, let us denote $\ket{\widehat{\psi}_t}:=\sum_i \beta_i \ket{\phi_i}$. In order to understand the complexity of the iterated $\Selfcorrection$ algorithm $\calL$, it remains to argue the following:

\textbf{\emph{1.}} \emph{Upper bound on the number of iterations $k$.} To argue that the process above stops after not-too-many steps, we consider the $\ell_2$ norms of  $\ket{\Psi_t}=\ket{\psi}-\ket{\widehat{\psi}_t}$. Using the orthogonality of $\ket{\phi_t}$ and $\ket{\Psi_{t+1}}$ (by construction of the $\ket{\psi_i}$ states above) as well as stopping criteria requirement that $\Exp_{x \sim q_{\Psi_t}}[|\la \psi_t | W_x | \psi_t \ra|^2] \geq \varepsilon^6$ and $\prod_{i \in [t-1]} |r_i|^2 \geq \varepsilon$ for all $t\leq k$, we show that
$$
\norm{\Psi_t}_2^2 - \norm{\Psi_{t+1}}_2^2 \geq \eta^2, \quad \text{ for all } t\leq k
$$
and summing this from $t=1,\ldots,k$, we get 
$
k\eta^2\leq \norm{\Psi_{1}}_2^2 - \norm{\Psi_{k}}_2^2 \leq 1,$ implying that $k\leq 1/\eta^2$.

\textbf{\emph{2.}} \emph{Algorithms to implement each step.} As part of $\calL$, we need the ability to estimate Gowers-$3$ norm of states, prepare the intermediate state $\ket{\psi_i}$ and computation of coefficients $\{c_i\}_i$. To $\varepsilon$-approximate the Gowers-$3$ norm, we use the $\poly(n,1/\varepsilon)$-time protocol from~\cite{ad2024tolerant}. To compute the coefficients $c_i$, we use the Hadamard test (which also allows us to track phase information). Finally,  to prepare the intermediate states $\ket{\psi_i}$ in each iteration, we use the well-known linear combination of unitaries ($\LCU$) method~\cite{childs2012lcu} and for this step, we require the ability to implement $U_\psi,\textsf{con}U_\psi$. Since one can construct the Clifford unitary $U_{\phi_i}$ which prepares the stabilizer state $\ket{\phi_i}$ efficiently, preparing the state $\ket{\psi_i}$ using $U_\psi,\textsf{con}U_\psi,U_{\phi_i},\textsf{con}U_{\phi_i}$ can be done efficiently. 
   
\textbf{\emph{3.}} \emph{Errors in subroutines.} The two steps above  are made formal in Section~\ref{sec:errorfreeSelfcorrection} (assuming all the subroutines are noiseless). Unfortunately, when we take into account the errors in estimating the Gowers norm, Hadamard test and $\LCU$ implementation, almost \emph{all} the identities that we use above to argue the correctness of the algorithm and to obtain an upper bound on $k$ do not hold. In the process, we  modify the algorithm (one of the concerns with iterated $\Selfcorrection$ is we do not know what \emph{is} the stopping point, so the accuracy up to which the parameters are estimated needs to be recalibrated across different iterations). Our final algorithm (which is resilient to these errors) incorporates all these details and making this rigorous is the most technical part of our work. Proving the correctness and convergence of the modified algorithm with noisy estimates is done using fairly involved calculations and is proven in Section~\ref{sec:errorSCsection}. With this, we prove Result~\ref{result:learn_decomposition}. 

While our discussion so far has focused on applying the $\Selfcorrection$ protocol iteratively to learn a structured decomposition, we could have replaced it with the stabilizer bootstrapping algorithm of Chen et al.~\cite{chen2024stabilizer} that would have allowed us, on input of a state $\ket{\psi_t}$ with the promise $\calF_{\calS}(\ket{\psi_t}) \geq \varepsilon$, to output a stabilizer state $\ket{\phi_t}$ with the promise that $|\la \psi_t | \phi_t \ra|^2 \geq \varepsilon/2$ in $\poly(n,(1/\varepsilon)^{\log 1/\varepsilon})$ time. So in every iteration, we would have $\eta = \varepsilon/2$ and the previous arguments hold for this value of $\eta$. The complexity here is worse than the one where we used $\Selfcorrection$ as the base algorithm, but it doesn't require the $\APFR$ Conjecture~\ref{conj:algopfrconjecture}.

\paragraph{Application.} We now give a learning algorithm for  states with low stabilizer extent. Let $\ket{\psi}=\sum_i c_i\ket{\phi_i}$ be the unknown  state with $\sum_i |c_i|\leq \xi$. By running our iterative procedure on $\ket{\psi}$ (with error set to $\varepsilon'/\xi$), we  obtain a $\ket{\widetilde{\psi}}$ such that stabilizer fidelity of $\ket{\psi}-\ket{\widetilde{\psi}}$ is $\leq \varepsilon'/\xi$. In particular, 
$$
|1-\langle \psi|\widetilde{\psi}\rangle|=|\langle \psi|{\psi}\rangle-\langle \psi|\widetilde{\psi}\rangle|\leq \sum_i |c_i|\cdot |\langle \phi_i| \psi-\widetilde{\psi}\rangle|\leq \xi\cdot \varepsilon=\varepsilon',
$$
which is precisely the requirement for tomography of states with low stabilizer extent. However, recall that $\ket{\widetilde{\psi}}$ isn't a valid quantum state, so we renormalize it and that gives a constant-distance tomography protocol.  Our results for learning low stabilizer rank states follows immediately by using recent results~\cite{mehraban2024improved,kalra2025stabilizer} who showed that stabilizer rank $k$ states satisfy $\xi(\ket{\psi})\leq k^{O(k)}$.  Using similar ideas as above, we also give learning algorithms for states which can be expressed as $\ket{\psi}=\sum_i c_i\ket{\phi_i}$ where $\ket{\phi_i}$ have stabilizer dimension $n-t$,  with a $2^t$ overhead in the complexity.
\subsection{Open questions} 
 Our work opens up several interesting directions for future work.

\begin{enumerate}
\item \emph{Complexity of algorithmic $\PFR$:} The first question that arises from this work is, what is the complexity of algorithmizing $\PFR$? In a recent work, \cite{algopfr} gave an algorithmic $\PFR$ result whose complexity scaled as $\poly(n,2^K)$. Often in additive combinatorics, $K$ is treated as a constant, but in this work, $K$ could scale with $n$. 
Can we improve this complexity to $\poly(n,K^{\log K})$ (using ideas from~\cite{ben2014sampling}), or $n^{O(\log \log K)}$ (since the proof of~\cite{gowers2023conjecture} involves an iterative with $\log \log K$ iterations) or even $\poly(n,K)$? If so, the latter results would immediately have consequences for our results. 
\item \emph{Remove the need of unitary access}: Our iterated $\Selfcorrection$ protocol required access to the unitary preparing $\ket{\psi}$ and its controlled version. Can we remove the latter requirement and only work in the setting where we have copies of $\ket{\psi}$?
\item \emph{Further generalizations.} A natural extension is the one of considering mixed states. Recently,~\cite{iyer2024tolerant} extended the tolerant tester of \cite{ad2024tolerant} to \emph{mixed states}. Few results here naturally extend to the setting of mixed states (we mention it as part of our statements), but extending the applications of our main results is unclear. In another direction, recent work has extended Bell sampling and the notion of Gowers norms to qudits~\cite{allcock2024beyond,bu2025quantum,bu2025stabilizer}. This raises the question if one can extend the results in this paper to that for qudit~systems?

\item \emph{Implications to complexity theory?} Classically, agnostic learning and self-correction has been shown to have implications to list-decoding error correcting codes and also in the PCP theorem~\cite{tulsiani2014quadratic,ben2014sampling,hatami2018cubic}.  
Our iterated protocol could potentially allow to extract the ``classical part" of an arbitrary quantum object, leaving as a residual the ``truly quantum" aspect of the object. Does this have implications beyond learning?
\item \emph{Self-correction beyond stabilizers?}  A natural question is, can one obtain an efficient self-correction algorithm for other classes of states? For example, there has been recent progress in improved tomography protocols for free-fermionic states~\cite{bittel2025optimal}. Could one then extend the results in this work to obtain analogous results for free-fermionic states? 
\end{enumerate}

\subsection{Organization of paper}
In Part I, Section~\ref{sec:prelims}  of the paper we  state and prove several lemmas and describe the notational convention that we use throughout the paper. In Part II of the paper, we give our self correction results:  in Section~\ref{sec:subroutines} we give our algorithmic components for self correction and in Sections~\ref{sec:selfcorrectionprotocol} and~\ref{sec:improperselfcorrection}, we give our proper and improper self correction algorithms respectively. In Part III of the paper, we give our iterative procedure which could use an arbitrary base algorithm:  in Section~\ref{sec:iteratedselfcorrection} we give our entire iterative self correction procedure and in Section~\ref{sec:applications} we present our applications for iterated $\Selfcorrection$ procedure.   

We have written this paper in a modular fashion to aid the reader. Since the main technical contributions are two-fold, a self-correction algorithm and an iterative  procedure (which could use the base algorithm as either $\Selfcorrection$ or stabilizer bootstrapping~\cite{chen2024stabilizer}), a reader who would like to just understand the iterative procedure along with the now-well-known stabilizer bootstrapping framework can directly read Part III of the paper (and skip Part II of the~paper).

\paragraph{Acknowledgments.} SA and AD thank Sergey Bravyi for multiple useful discussions, and Theodore (Ted) J. Yoder for introducing us to the symplectic Gram-Schmidt procedure. We thank Madhur Tulsiani for helpful clarifications on~\cite{tulsiani2014quadratic}. We thank Sabee Grewal and Pulkit Sinha for several discussions. We also thank Jop Briet, David Gosset,  Patrick Rall, and Ewout van den Berg for discussions. SA and AD thank the Institute for Pure and Applied Mathematics (IPAM) for its hospitality throughout the long program “Mathematical and Computational Challenges in Quantum Computing” in Fall 2023 during which part of this work was initiated. 

\section{Preliminaries}
\label{sec:prelims}
For $n\geq 1$, let $[n]=\{1,\ldots,n\}$.  For a set $A\subseteq \FF_2^n$ and $k\geq 1$, define $kA=\{\sum_{i\in S}a_i:a_i\in A\}_{|S|=k}$, so $2A=\{a_1+a_2:a_1,a_2\in A\}$. We define the \emph{doubling constant} of $A$ as the smallest integer $t$ such that $|2A|\leq t|A|$. Throughout this paper, for a quantum state $\ket{\psi}$, we will denote $U_{\psi}$ to denote a unitary that prepares $\ket{\psi}$, i.e., $U_{\psi}\ket{0^n}=\ket{\psi}$.  For $\varepsilon\in (0,1)$, we say $f(\varepsilon)=\poly(\varepsilon)$ if there exists a constant $c_1,c_2\geq 1$ such that $f(\varepsilon)=c_1\varepsilon^{c_2}$.\footnote{In this paper, there are several polynomial factors that we have not explicitly optimized, so we use the convention $\poly(\varepsilon)$ to make the exposition easier to follow.}  We define $\calB_\infty^k$ as the unit complex ball, i.e., $x\in \calB_\infty^k$ if $x_i \in \mathbb{C}$ for all $i\in [k]$ and $|x_i| \in (0,1)$. 

\subsection{Weyl operators}
The single-qubit Pauli matrices matrices are defined as follows
$$\id=\begin{pmatrix}
1 & 0\\
0 & 1
\end{pmatrix}, X=\begin{pmatrix}
0 & 1\\
1 & 0
\end{pmatrix}, Y=\begin{pmatrix}
0 & -i\\
i & 0
\end{pmatrix},Z=\begin{pmatrix}
1 & 0\\
0 & -1
\end{pmatrix}
$$
It is well-known that the $n$-qubit Pauli matrices $\{\id,X,Y,Z\}^n$ form an  {orthonormal basis} for $\mathbb{C}^n$.    In particular, for every $x=(a,b)\in \mathbb{F}_2^{2n}$, one can define the \emph{Weyl operator}
$$
W_x=i^{a\cdot b} (X^{a_1}Z^{b_1}\otimes X^{a_2}Z^{b_2} \otimes \cdots \otimes X^{a_n}Z^{b_n}).
$$  
and these operators $\{W_x\}_{x \in \FF_2^{2n}}$ are orthonormal. Note that each Weyl operator is a Pauli operator and indeed, every Pauli operator is a Weyl operator up to a phase. Throughout the paper we will denote  $\calP^k$ to be the $k$-qubit Pauli group and $\calP_Z^m = \{I,Z\}^{\otimes m}$. 
 For $x,y \in \mathbb{F}_2^{2n}$, we write $x=(x_1, x_2)$ with $x_1$ denoting the first $n$ bits of $x$ and $x_2$ denoting the last $n$ bits (similarly for $y=(y_1,y_2)$). We define the \emph{symplectic inner product} as
\begin{equation}
    [x,y] = \la x_1, y_2 \ra + \la x_2, y_1 \ra \mod 2. 
    \label{eq:symplectic_inner_product}
\end{equation}
We define the commutation relations of Paulis $W_x,W_y$ for all $x,y \in \FF_2^{2n}$ using the symplectic product of their corresponding strings $x,y$. In particular, we say that two Paulis $W_x, W_y$ commute if $[x,y]=0$ and $W_x, W_y$ anti-commute if $[x,y]=1$. It is evident that $W_x W_y= (-1)^{[x,y]} W_y W_x$. Furthermore every   $\ket{\psi}$ can be written~as
$$
\ketbra{\psi}{\psi}=\frac{1}{2^n}\sum_{x\in \mathbb{F}_2^n} \alpha_x \cdot W_x,
$$
where   
$$
\alpha_x=\Tr(W_x \ketbra{\psi}{\psi}), \qquad \frac{1}{2^n}\sum_x \alpha_x^2=1.
$$   
Below we will use  {$p_\Psi(x)=\alpha_x^2/2^n$}, so that $\sum_x p_\Psi(x)=1$.    Since $n$-qubit Paulis can be associated with ($2n$)-bit strings, we will often refer to a Pauli $P$ by $x\in \01^{2n}$ by which we mean $P=W_x$. 

We define the \emph{characteristic distribution}~as
\begin{equation}
\label{eq:defnofppsi}
    p_\Psi(x) = \frac{|\la \psi | W_x | \psi \ra|^2}{2^n},
\end{equation}
which satisfies $\sum_{x \in \FF_2^{2n}} p_\Psi(x)=1$. It is well-known that one can \emph{sample} from the characteristic distribution by carrying out Bell sampling on $\ket{\psi} \otimes \ket{\psi^\star}$, where $\ket{\psi^\star}$ is the conjugate of $\ket{\psi}$. The \emph{Weyl distribution}~\cite{gross2021schur} is defined as $q_\Psi$ as
\begin{equation}
    q_\Psi(x) = \sum_{y \in \mathbb{F}_2^{2n}} p_\Psi(y) p_\Psi(x+y).
    \label{eq:weyl_distribution}
\end{equation}
We will require the following upper bounds on $q_\Psi(x),p_\Psi(x)$ for any $x \in \FF_2^{2n}$ (following the presentation of~\cite{chen2024stabilizer}).
\begin{fact}
\label{lem:ub_qPsi}
For any $x \in \FF_2^{2n}$, $q_\Psi(x),p_\Psi(x) \leq 2^{-n}$.
\end{fact}
\begin{proof}
The bound on $p_\Psi(x)$ is trivial by definition of Eq.~\eqref{eq:defnofppsi} (since the number of that equation is at most $1$). The upper bound on $q_\Psi(x)$ is folklore.
\begin{align}
 q_\Psi(x) = \sum_{y \in \mathbb{F}_2^{2n}} p_\Psi(y) p_\Psi(x+y)
 &= \sum_{y \in \mathbb{F}_2^{2n}} (-1)^{[x,y]}p_\Psi(y) p_\Psi(y)\\
 &=\frac{1}{2^{2n}}\sum_y (-1)^{[x,y]}|\langle \psi|W_y|\psi\rangle|^2|\langle \psi|W_{y}|\psi\rangle|^2\\
 &\leq \frac{1}{2^{2n}}|\langle \psi^{\otimes 2}|\sum_y W_y^{\otimes 2}|\psi^{\otimes 2}\rangle|^2=\frac{1}{2^{n}}|\langle \psi^{\otimes 2}|\textsf{SWAP}|\psi^{\otimes 2}\rangle|^2=2^{-n},
\end{align}
where the second equality is by
the penultimate equality is by definition that $\sum_y W_y^{\otimes 4}=2^n \textsf{SWAP}$ (where $\textsf{SWAP}$ is the usual swap operator on a bipartite system.
\end{proof}

We will use the following two facts that were established by prior work~\cite{ad2024tolerant}. 
\begin{fact}[{\cite[Lemma~2.9]{ad2024tolerant}}]
\label{fact:relation_qPsi_and_pPsi} 
For every $\ket{\psi}$ we have that
\begin{equation*}
    \Exp_{x \sim q_\Psi}\left[|\la \psi | W_x | \psi \ra|^2 \right] = 2^{3n} \Exp_{y \in \mathbb{F}_2^n} \left[\sum_{\alpha \in \mathbb{F}_2^{n}} {p}_\Psi^3(y,\alpha)\right].
\end{equation*}    
\end{fact}

\begin{lemma}[{\cite[Lemma~3.8]{ad2024tolerant}}]\label{lem:est_weyl_exp}
    Let $\ket{\psi}$ be an arbitrary $n$-qubit quantum state. One can estimate $\Exp_{x \sim q_\Psi}\left[|\la \psi | W_x | \psi \ra|^2 \right]$ upto additive error $\delta$ using $O(1/\delta^2)$ copies of $\ket{\psi}$ and $O(n/\delta^2)$ gates.
\end{lemma}

\subsection{Stabilizer subgroups and stabilizer states}

\paragraph{Pauli gates.} The $2$-qubit Pauli matrices matrices are defined as follows
$$\id=\begin{pmatrix}
1 & 0\\
0 & 1
\end{pmatrix}, X=\begin{pmatrix}
0 & 1\\
1 & 0
\end{pmatrix}, Y=\begin{pmatrix}
0 & -i\\
i & 0
\end{pmatrix},Z=\begin{pmatrix}
1 & 0\\
0 & -1
\end{pmatrix}
$$
It is well-known that the $n$-qubit Pauli matrices $\{\id,X,Y,Z\}^n$ form an  {orthonormal basis} for $\mathbb{C}^n$.    In particular, for every $x=(a,b)\in \mathbb{F}_2^{2n}$, one can define the \emph{Weyl operator}
$$
W_x=i^{a\cdot b} (X^{a_1}Z^{b_1}\otimes X^{a_2}Z^{b_2} \otimes \cdots \otimes X^{a_n}Z^{b_n}).
$$  
and these operators $\{W_x\}_{x \in \FF_2^{2n}}$ are orthonormal. 

\paragraph{Clifford and $T$ gates.}  Clifford unitaries are those generated by Hadamard gate $\textsf{Had}=\frac{1}{\sqrt{2}}\begin{pmatrix}
1 & 1\\
1 & -1
\end{pmatrix}$, controlled-$X$ gate and $S=\begin{pmatrix}
1 & 0\\
0 & i
\end{pmatrix}$ gate. In order to achieve universal quantum computing, we also need the  $T$ gate is defined as $T=\begin{pmatrix}
1 & 0\\
0 & i
\end{pmatrix}$.

\paragraph{Stabilizer fidelity.}  We denote $\calF_\calS(\ket{\psi})$ to be the maximum \emph{stabilizer fidelity} of a quantum state $\ket{\psi}$, i.e., the overlap between $\ket{\psi}$ and the ``closest" (to $\ket{\psi})$) stabilizer state $\ket{s}$. 
More formally,
$$
\calF_\calS(\ket{\psi})=\max_{\ket{s}\in \textsf{Stab}} |\langle s|\psi\rangle|^2,
$$

\paragraph{Stabilizer dimension.} We first define the \emph{unsigned stabilizer group} as 
$$
\weyl{\ket{\psi}}=\{x\in \FF_2^{2n}: \langle \psi|W_x|\psi\rangle \in \{-1,1\}\}
$$ to be the Pauli matrices that stabilize $\ket{\psi}$.  We say that an $n$-qubit pure quantum state $\ket{\psi}$ has stabilizer dimension of $k$ if $\ket{\psi}$ is stabilized by an Abelian group of $2^k$ Pauli operators, in other words $\dim(\weyl{\ket{\psi}}) = k$. A stabilizer state has the maximal stabilizer dimension of $n$. Let $\Sh(n-t)$ be the states with stabilizer dimension $n-t$, i.e., if $\ket{\psi}\in \Sh(n-t)$, then $\dim(\weyl{\ket{\psi}})\geq n-t$.

\paragraph{Lagrangian subspace.} We say a subspace $S \subset \FF_2^{2n}$ is isotropic when $[x,y]=0$ for all $x,y \in S$ i.e., all the Weyl operators corresponding to the strings in $S$ \emph{commute} with each other. We say that an isotropic subspace $S$ is a Lagrangian subspace when it is of \emph{maximal size} $2^n$ i.e., $|S| = 2^n$. This ties in with the fact that a maximal set of $n$-qubit commuting Paulis is of size $2^n$. We will often use the following well-known theorem that relates stabilizer fidelity to the weight of the Weyl distribution inside a Lagrangian subspace. 
\begin{theorem}[\cite{gross2021schur,grewal2024agnostic}]
\label{fact:lower_bound_stabilizer_fidelity}
For every $\ket{\psi}$ and Lagrangian subspace $T\subseteq \mathcal{P}_n$, we have $$
\stabfidelity{\ket{\psi}}\geq \Exp_{P\in T}[|\langle \psi|P|\psi\rangle|^2].
$$
Furthermore, if $\calF_\calS(\ket{\psi}) \geq \eta$, then
\begin{equation*}
  \Exp_{x \sim q_\Psi}\left[|\la \psi | W_x | \psi \ra|^2 \right] \geq \eta^6.
\end{equation*}
\end{theorem}

\begin{lemma}[{\cite[Lemma~4.6]{grewal2023efficient}}]
\label{lem:state_isotropic_space_high_mass}
Let $A$ be an isotropic subspace of dimension $n-t$, and suppose that
$$
\sum_{x \in A}p_\Psi(x) \geq 2^{-t} \eta
$$
Then, there exists a state $\ket{\phi}$ such that $A \subseteq \weyl{\ket{\phi}}$ and $|\la \phi | \psi \ra|^2 \geq \eta$. In particular, $\ket{\phi} = U^\dagger \ket{\varphi}\ket{x}$, where $\ket{x}$ is an $(n-t)$-qubit basis state,
$$
\ket{\varphi} := \frac{(I \otimes \bra{x})U\ket{\psi}}{\norm{(I \otimes \bra{x})U\ket{\psi}}_2}
$$
is a $t$-qubit quantum state, and $U$ is a Clifford circuit mapping $A$ to $\la Z_{t+1},\ldots,Z_{n}\ra$.
\end{lemma}

\begin{definition}[$t$-doped states]
A $t$-doped Clifford circuit is a quantum circuit that consists of Clifford gates and at most $t$ many single-qubit non-Clifford gates. A $t$-doped state is the output of a $t$-doped Clifford circuit on $\ket{0^n}$.
\end{definition}
Grewal et al.~\cite[Lemma~4.2]{grewal2023improved} showed that $t$-doped states have stabilizer dimension at least $n-2t$. Throughout this paper we will denote $\calS(n-t)$ to be the set of states with stabilizer dimension at least $n-t$.

\paragraph{Stabilizer rank.}
Introduced in~\cite{bravyi2016trading}, we say a quantum state $\ket{\psi}$ has \emph{stabilizer rank $k$}, if $\ket{\psi}$ can be expressed $\ket{\psi}=\sum_{i=1}^k c_i\ket{s_i}$ where $c_i\in \mathbb{C}$ and $\ket{s_i}$ are stabilizer states. More formally, we define stabilizer rank of a state $\ket{\psi}$ as
$$
\chi(\ket{\psi})=\min\{k: \ket{\psi}=\sum_{i=1}^k c_i \ket{s_i}, \ket{s_i} \text{ are stabilizer states}\}.
$$
Also, the \emph{stabilizer extent} of a quantum state is the minimal $\ell_1$ norm of the coefficients in the decomposition above. Formally, stabilizer extent of $\ket{\psi}$ is defined as 
$$
\xi(\ket{\psi})=\min\{\sum_i |c_i|: \ket{\psi}=\sum_i c_i \ket{s_i}, \ket{s_i} \text{ are stabilizer states}\}.
$$
We will often encounter operators  $\psi$ which (up to renormalization) can be written as a stabilizer rank $k$ state, i.e., $\psi=\sum_{i\in [k]} c_i \ket{s_i}$ which upto a renormalization $\alpha=\|\sum_{i\in [k]} c_i \ket{s_i}\|_2$ is a valid quantum~state.
In~\cite{mehraban2024improved,mehraban2024quadratic} it was shown that there exists a function $\delta:\mathbb{N}\rightarrow \mathbb{N}$ such that for every $\ket{\phi}$, we have $\chi(\ket{\phi})\leq \poly(n,2^{\delta(\xi(\ket{\phi})})$, i.e., stabilizer rank is polynomial in $n$ and exponential in some function that only dependence on the stabilizer extent $\xi(\ket{\phi})$. More recently, this result was improved in \cite{kalra2025stabilizer} which provided an expression for this function $\delta$.
\begin{theorem}[{\cite[Theorem~1]{kalra2025stabilizer}}]\label{thm:ub_stab_extent_stab_rank_states}
Let $\ket{\psi}$ be an $n$-qubit state with stabilizer rank $\kappa$. Then,
$$
\xi(\ket{\psi}) \leq \sqrt{e} \cdot (2 \kappa)^{(2\kappa + 1)/2}.
$$
\end{theorem}

When the arbitrary state $\ket{\psi}$ is promised to be produced by a circuit consisting of Clifford gates and $t$ many $T$ gates, the following result is known that gives a better upper bound on the stabilizer extent of $\ket{\psi}$.
\begin{lemma}[{\cite[Lemma~2.5]{grewal2022low}}]\label{lem:stab_extent_clifford_T_circs}
    Let $U$ be a circuit with Clifford gates and  $t$ many $T$ gates. Let $\ket{\psi} = U \ket{0^n}$. Then,
    $$
    \xi(\ket{\psi}) \leq \left(1 + \frac{1}{\sqrt{2}}\right)^t.
    $$
\end{lemma}

\paragraph{Mutually unbiased bases $(MUB)$.}
Often, we will attempt to cover a subgroup of Weyl operators by a union of Lagrangian subspaces (aka unsigned stabilizer subgroups), also called a \emph{stabilizer covering}. We  make use of the following fact regarding MUBs from~\cite{bandyopadhyay2002new,flammia2020efficient}.
\begin{fact}\label{fact:mub_unsigned_stab}
The group of $k$-qubit Weyl operators $\{W_x\}_{x \in \FF_2^{2k}}$ can be covered by $2^k + 1$ many $k$-qubit stabilizer groups $\{G_i\}_{i \in [2^k + 1]}$. Each of these are disjoint up to the identity element.
\end{fact}

\subsection{Gowers norm and inverse theorems}
For any $n$-qubit quantum state $\ket{\psi}=\sum_{x \in \{0,1\}^n} f(x)\ket{x}$ where $f=(f(x))_x$ is an $\ell_2$-normed vector, one can define its Gowers-$k$ norm~\cite{ad2024tolerant} as follows  
\begin{equation}
    \gowers{\ket{\psi}}{k} = 2^{n/2} \left[ \Exp_{x,h_1,h_2,\ldots,h_k \in \{0,1\}^n} \prod_{\omega \in \{0,1\}^k} C^{|\omega|} f(x + \omega \cdot h) \right]^{1/2^{k}},
\end{equation}
where $C^{|\omega|} f = f$ if $|\omega| := \sum_{j \in [k]} \omega_k$ is even and is $\overline{f}$ if $|\omega|$ is odd with $\overline{f}$ denoting the complex conjugate of $f$. We state facts about Gowers norm of quantum states, that will be useful~later.
\begin{fact}[{\cite[Lemma~3.3]{ad2024tolerant}}]
\label{fact:relation_expectation_paulis_qPsi_and_pPsi} 
For every $\ket{\psi}$ we have that 
\begin{equation*}
    {\gowers{\ket{\psi}}{3}^{8}}=\Exp_{x \sim p_\Psi}\left[|\la \psi | W_x | \psi \ra|^2 \right]  \geq  \Exp_{x \sim q_\Psi}\left[|\la \psi | W_x | \psi \ra|^2 \right] \geq \underbrace{\left(\Exp_{x \sim p_\Psi}\left[|\la \psi | W_x | \psi \ra|^2 \right] \right)^2}_{= \gowers{\ket{\psi}}{3}^{16}}. 
\end{equation*}
\end{fact}
We have the following inverse theorem for the Gowers-$3$ norm of quantum states that was established in~\cite{ad2024tolerant}.
\gowersstates*
As it is easier to estimate $\Exp_{x \sim q_\Psi}\left[|\la \psi | W_x | \psi \ra|^2 \right]$ which is a proxy for the Gowers-$3$ norm given the relation in Fact~\ref{fact:relation_expectation_paulis_qPsi_and_pPsi}, we also have the following inverse theorem.
\begin{restatable}{theorem}{inverseweyl}\label{thm:inverse_weyl_exp_states}
    Let $\gamma\in [0,1]$. If $\ket{\psi}$ is an $n$-qubit quantum state such that $\Exp_{x \sim q_\Psi}\left[|\la \psi | W_x | \psi \ra|^2 \right] \geq \gamma$, then there is an $n$-qubit stabilizer state $\ket{\phi}$ such that $|\la \psi | \phi \ra|^2 \geq \Omega(\gamma^{C'})$ for some constant $C'>1$.
\end{restatable}

\subsection{Useful  facts and subroutines}\label{sec:useful_subtroutines}

\begin{fact}
\label{fact:lowerboundexpectation}
Let $Y$ be a random variable with $|Y|\leq 1$. If $\Exp[Y]\geq \varepsilon$, then 
$$
\Pr[Y\geq \delta]\geq (\varepsilon-\delta)/(1-\delta)
$$
\end{fact}

\begin{fact}(Hoeffding bound)
\label{fact:hoeffding_sampling}
Let $Y$ be a random variable with $|Y|\leq 1$ and $\widehat{\mu}$ is the empirical average obtained from $T$ samples, then
$$
\Pr\left[|\Exp[Y] - \widehat{\mu}| > a \right] \leq \exp(-\Omega(a^2 T)).
$$
\end{fact}
We will need the following subroutines often used as primitives in various quantum algorithms.

\paragraph{Fidelity estimation.}
We can measure the fidelities of an unknown $n$-qubit quantum state $\Psi$ with a set of stabilizer states efficiently using classical shadows~\cite{huang2020predicting}.
\begin{lemma}[\cite{huang2020predicting}]
\label{lem:classicalshadow}
Given an $n$-qubit quantum state $\Psi$ and $M$ stabilizer states $\{ \ket{\phi_j} \}_{j \in [M]}$, there is an algorithm that estimates the fidelity $\la \phi_j | \Psi | \phi_j \ra $ to error at most $\varepsilon$ for all $j \in [M]$ with probability at least $1-\delta$, requiring $O\left(\frac{1}{\varepsilon^2} \log \frac{M}{\delta} \right)$ sample complexity and $O\left( \frac{M}{\varepsilon^2} n^2 \log \frac{M}{\delta} \right)$ time complexity. The algorithm uses only single-copy measurements.
\end{lemma}

It was shown in Chen et al.~\cite{chen2024stabilizer} that this can be extended to estimating fidelities with a set of states with high stabilizer dimension as well.
\begin{lemma}[{\cite[Lemma~4.17]{chen2024stabilizer}}]\label{lem:shadows_stab_dim}
Given $t\in \mathbb{N}$, an $n$-qubit quantum state $\Psi$, and $M$ Clifford unitaries $\{ U_j \}_{j \in [M]}$, there is an algorithm that estimates the fidelity $\Tr(\la 0^{n-t} |U_j^\dagger \Psi U_j | 0^{n-t} \ra)$ to error at most $\varepsilon$ for all $j \in [M]$ with probability at least $1-\delta$, requiring $O\left(\frac{2^{2t}}{\varepsilon^2} \log \frac{2^t M}{\delta} \right)$ sample complexity and $O\left( \frac{2^{3t} M}{\varepsilon^2} n^2 \log \frac{2^t M}{\delta} \right)$ time complexity. The algorithm uses only single-copy measurements.
\end{lemma}

\paragraph{State tomography.} We will also require the following tomography protocol~\cite{guctua2020fast}.
\begin{lemma}[Full tomography via single-copy measurements~\cite{guctua2020fast}]\label{lem:state_tomo}
Given copies of an $n$-qubit quantum state $\ket{\Psi}$, there is an algorithm that outputs a density matrix $\widehat{\Psi}$ such that $d_{\mathrm{tr}}(\Psi, \widehat{\Psi}) \leq \varepsilon$  with probability
at least $1-\delta$. The algorithm performs $O(2^{4n} n \log(1/\delta)/\varepsilon^2)$ single-copy measurements on $\Psi$ and takes $O(2^{4n} n^2 \log(1/\delta)/\varepsilon^2)$ time.
\end{lemma}

\paragraph{Subroutines for subspaces of Paulis.}
We will also require the following subroutines to generate a Clifford circuit given an isotropic subspace of $\FF_2^{2n}$ or a pair of anti-commuting Paulis.
\begin{lemma}[{\cite[Lemma~3.2]{grewal2023efficient}}]\label{lem:clifford_isotropic_subspace}
Given a set of $m$ vectors whose span is a $d$-dimensional isotropic subspace $A \subset \FF_2^{2n}$, there exists an efficient algorithm that outputs a Clifford circuit $C$ such that $CAC^\dagger = 0^{2n-d} \otimes \FF_2^d$. The algorithm runs in $O(mn\cdot\min(m,n))$ time and the gate complexity of $C$ is $O(nd)$.
\end{lemma}

\begin{lemma}[{\cite[Lemma~1]{bravyi2021clifford}}]\label{lem:clifford_anticomm_paulis}
There exists an algorithm that takes as input anti-commuting $n$-qubit Pauli operators $P$ and $P'$
and outputs an $n$-qubit Clifford circuit $U$ such that
$$
UPU^\dagger = X_1 \quad \text{and} \quad UP'U^\dagger = Z_1.
$$
The circuit $U$ has a CNOT cost of $\leq 3n/2 + O(1)$. The algorithm has runtime $O(n)$.
\end{lemma}

\paragraph{Stabilizer bootstrapping.} In a recent work Chen et al.~\cite{chen2024stabilizer} considered the task of agnostic learning stabilizer states,  and proved the following following theorem.
\begin{theorem}
\label{thm:sitanbootstrapping}
    Let $\mathcal{C}$ be  the class of stabilizer states. Fix any $\varepsilon\leq \tau \in (0,1)$. 
There is an algorithm that, given access to copies of $\rho$ with $\max_{|\phi'\rangle \in \mathcal{C}} |\langle \phi' | \rho | \phi' \rangle| \geq \tau$ for , outputs a $|\phi\rangle \in \C$ such that $|\langle \phi | \rho | \phi \rangle| \geq \tau - \varepsilon$ with high probability.
The algorithm performs single-copy and two-copy measurements on at most $n\cdot \poly(1/\varepsilon,(1/\tau)^{\log 1/\tau})$ copies of $\rho$ and runs in time $\poly(n,1/\varepsilon, (1/\tau)^{\log 1/\tau})$.
\end{theorem}

\newpage

\part{Self-Correction of Stabilizer States}
The main contribution of this work is a $\Selfcorrection$ protocol for stabilizer states, which  algorithmizes the main result in~\cite{ad2024tolerant}. In this part of the paper, we prove two versions of our self-correction results, when the output state is proper (i.e., it is a stabilizer state) and when the output state is improper (i.e., it is an arbitrary state). 
In order to describe our $\Selfcorrection$ algorithm, we separate our presentation into three sections. In Section~\ref{sec:subroutines}, we first present the algorithmic components that are used to algorithmize the inverse Gowers-$3$ theorem of quantum states from~\cite{ad2024tolerant}. Describing these subroutines by themselves does not immediately give a $\Selfcorrection$ protocol and we describe how these subroutines are combined together to give the final protocol in Section~\ref{sec:selfcorrectionprotocol}. Finally in Section~\ref{sec:improperselfcorrection} we present an improper algorithm for $\Selfcorrection$. 

\section{Algorithmic components for Self correction}
\label{sec:subroutines}
Throughout this section we will work  with $n$-qubit quantum states $\ket{\psi}$ with   $\Exp_{x \sim q_\Psi}[|\la \psi | W_x | \psi \ra|^2] \geq \gamma$ (which is a good \emph{proxy} for high Gowers-$3$ norm, see Fact~\ref{fact:relation_expectation_paulis_qPsi_and_pPsi}).  Denote $2^np_\Psi(x)=|\la \psi | W_x | \psi \ra|^2$. 
For a state~$\ket{\psi}$ and $\gamma>0$, define the set $S_\gamma\subseteq \FF_2^{2n}$ as 
\begin{align}
\label{eq:defnofsgamma}
S_\gamma = \{x \in \FF_2^{2n} :  2^np_\Psi(x) \geq \gamma\}.
\end{align}
First recall the main result of~\cite{ad2024tolerant}.
\gowersstates*
Below, all our self correction results are stated in terms of output distributions of Bell difference sampling and finally when we prove our theorem statement,  we will relate this to Gowers norm and prove our main Theorem~\ref{thm:self_correction}.
We now describe the proof of the theorem above, and how attempting to algorithmize the involved sequence of arguments motivate the algorithmic subroutines that we discuss below.  
\begin{enumerate}[$(i)$]
    \item \emph{Approximate subgroup}: In~\cite{ad2024tolerant} it was first observed that if $\Exp_{x\sim q_\Psi}[|\langle \psi|W_x|\psi\rangle|^2] \geq \gamma$, then there exists a large approximate subgroup $S \subseteq \mathbb{F}_2^{2n}$ such that $S \subseteq S_{\gamma/4}$ with size $|S|\in [\gamma/2 , 2/\gamma ]\cdot 2^n$ satisfying 
    $$
    \Pr_{x,y \in S}[x+y \in S]\geq \poly(\gamma).
    $$
    Since $S$ is exponentially sized, we will not be able to hold this in memory let alone hope to construct it. We will instead show how Bell difference sampling of $\ket{\psi}$ can be utilized to sample elements from $S$ with high probability. This is described in Section~\ref{sec:bellsampling}.
    \item \emph{Small doubling set}: Applying the Balog-Szemeredi-Gowers ($\BSG$) theorem~\cite{balog1994statistical,gowers2001new} to $S$ shows the existence of a large subset $S' \subseteq S$ that has a small doubling constant. In particular, \cite{ad2024tolerant} showed that 
    $$
    |S' + S'| \leq \poly(1/\gamma) |S|, \enspace |S'| \geq \poly(\gamma)\cdot  |S|.
    $$
    Again since $S'$ is exponentially large, we will not hold this nor construct it. We will describe a membership test for $S'$ and combined with the sampling subroutine from $S$ in $(i)$ will allow us to obtain samples from $S'$. The membership test is described in Section~\ref{subsec:BSG_test}.
    \item \emph{Subgroup with large mass}: Applying the polynomial Freiman-Ruzsa theorem~\cite{gowers2023conjecture} to~$S'$, one can observe that $S'$ is covered by few translates of a subgroup $V \subseteq S'$ and with some analysis, \cite{ad2024tolerant} showed that 
    $$
    \Exp_{x \in V}\left[2^n p_\Psi(x) \right] \geq \poly(\gamma).
    $$
    Assuming sample and membership access to $S'$ as described in $(iii)$, we will show that a basis of $V$ can be determined assuming the algorithmic $\PFR$ conjecture (Conjecture~\ref{conj:algopfrconjecture}).  We describe this directly as part of the $\Selfcorrection$ protocol.
    \item \emph{Efficient stabilizer covering}: The subgroup $V$ is shown to have a stabilizer covering of $\poly(1/\gamma)$, which allowed \cite{ad2024tolerant} to conclude that there is a stabilizer state $\ket{\phi}$ such that $|\la \psi | \phi \ra|^2 \geq \poly(\gamma)$.  To determine this stabilizer state, we need to determine a stabilizer covering of $V$. To this end, we use the Symplectic Gram Schmidt procedure that helps obtain a stabilizer covering of $V$ and describe this procedure in Section~\ref{sec:symplectic_gram_schmidt}.
\end{enumerate}
The entire goal of the subroutines will be to avoid using exponential memory (and time) in \emph{storing} the sets described above and instead have indirect access to these sets via oracles. 

\subsection{Bell difference sampling}
\label{sec:bellsampling}
The starting point of proving Theorem~\ref{thm:inversegowersstates}  was,  if $\Exp_{x \sim q_\Psi}[|\la \psi | W_x | \psi \ra|^2]  \geq \gamma$  then there exists a set $S \subseteq S_\gamma$ of size at least $\poly(\gamma) \cdot 2^n$ that is an \emph{approximate} group, i.e.,
$$
\Pr_{x,y\sim S}[x+y\in S]\geq \poly(\gamma),
$$
where the probability is uniformly random in $S$. 
More formally,~\cite{ad2024tolerant} showed the following. 

\begin{theorem}[{\cite[Theorem~4.6]{ad2024tolerant}}]
\label{thm:large_set_gamma_qPsi_expectations}
Let $\gamma >0$. If $\Exp_{x \sim q_\Psi}[|\la \psi | W_x | \psi \ra|^2] \geq \gamma$, then there exists $S \subseteq \mathbb{F}_2^{2n}$  satisfying $(i)$ $|S|\in [\gamma^2/80 , 4/\gamma ]\cdot 2^n$, $(ii)$ $S \subseteq S_{\gamma/4}$, $(iii)$ $\Pr_{x, y \in S}[ x+y \in S]\geq \Omega(\gamma^5)$, and $(iv)$ additionally $0^{2n}$ lies in $S$.
\end{theorem}

As part of algorithmizing Theorem~\ref{thm:inversegowersstates}, we first devise a protocol, that given copies of $\ket{\psi}$,  efficiently samples elements from the approximate subgroup $S$ in Theorem~\ref{thm:large_set_gamma_qPsi_expectations}. To this end, we have the following claim regarding sampling from the set $S_\gamma$ in Eq.~\eqref{eq:defnofsgamma}.

\begin{claim}
\label{claim:high_exp_value_bds}
If $\ket{\psi}$ with $\Exp_{x \sim q_\Psi}[|\la \psi | W_x | \psi \ra|^2] \geq \gamma$, then
\begin{enumerate}
    \item Bell sampling on $\ket{\psi} \otimes \ket{\psi^\star}$ outputs an $x\in S_{\gamma/4}$ with prob.~$\geq~{3\gamma}/{4}$.
    \item Bell difference sampling on $\ket{\psi}^{\otimes 4}$ outputs an $x\in S_{\gamma/4}$ with prob.~$\geq~{3\gamma}/{4}$.
\end{enumerate}
\end{claim}
\begin{proof}
To see the first bullet, Fact~\ref{fact:relation_expectation_paulis_qPsi_and_pPsi}  implies that 
$$
\Exp_{x \sim p_\Psi}\left[ 2^np_\Psi(x) \right]= \Exp_{x \sim p_\Psi}\left[ \langle \psi|W_x|\psi\rangle|^2\right] \geq \Exp_{x \sim q_\Psi}\left[|\la \psi | W_x | \psi \ra|^2 \right]\geq \gamma.
$$ 
Applying Fact~\ref{fact:lowerboundexpectation} for $\delta:=\gamma/4$ and $\varepsilon:=\gamma$ yields $\Pr_{x \sim p_\Psi}\left[2^np_\Psi(x) \geq \gamma/4 \right] \geq 3\gamma/4$. Observe that Bell sampling on $\ket{\psi}\otimes \ket{\psi^\star}$ allows us to sample $x$ according to the characteristic distribution $p_\Psi$. So the first part of the claim follows. To see the second bullet, again apply Fact~\ref{fact:lowerboundexpectation} for $\delta:=\gamma/4$ and $\varepsilon:=\gamma$ and this yields $\Pr_{x \sim q_\Psi}\left[ 2^np_\Psi(x) \geq \gamma/4 \right] \geq 3\gamma/4$. Observe that Bell difference sampling on four copies $\ket{\psi}$ allows us to sample $x$ according to the Weyl distribution~$q_\Psi$. So the second part of the claim~follows.
\end{proof}
We remark that in practice, copies of the unknown state's conjugate $\ket{\psi^\star}$ may not be available, hence one might not have access to the Bell sampling distribution, and only access to the Bell difference sampling. The claim above shows that one can carry out Bell difference sampling on copies of $\ket{\psi}$ to still produce Paulis $x \in \FF_2^{2n}$ with high expectation~values. We now state our main lemma which allows to sample from the set $S$ in Theorem~\ref{thm:large_set_gamma_qPsi_expectations}.

\begin{lemma}\label{lem:sample_bds}
Let $\upsilon>0$. Suppose $\ket{\psi}$ with $\Exp_{x \sim q_\Psi}[|\la \psi | W_x | \psi \ra|^2] \geq \gamma$.  There is a distribution over subsets of $S\subseteq S_\gamma$ such that: $|S|\geq 2^n\cdot \gamma^2/80$ and every element $x\in S$ is  (independently) sampleable using $\widetilde{O}(1/\gamma^2\cdot \log(1/\upsilon))$ copies of $\ket{\psi}$ and $\widetilde{O}(n/\gamma^2\cdot \log(1/\upsilon))$  time 
and furthermore with probability at least $1-\upsilon$, we have that
\begin{align}
\label{eq:noisycorollaryofAD24}
\Pr_{S}\left[[|S|/2^n \geq \gamma^2/80] \cap \Big[\Pr_{x,y\sim S}[x+y\in S]\geq \gamma^5/20\Big]\right]\geq \gamma^7/80.
\end{align}
\end{lemma}
\begin{proof}
In~\cite{ad2024tolerant} they showed the following: consider the set $S_{\gamma/4}$, and define a randomized \emph{choice set} $X\subseteq S_{\gamma/4}$ by including  every $x\in S_{\gamma/4}$ into $X$ with probability $2^np_\Psi(x)$. In~\cite{ad2024tolerant} they show that $X$ is dense with high probability, i.e., $|X|\geq 2^n\cdot \poly(\gamma)$. Then, they showed that this randomized procedure produces an $X$ that satisfies\footnote{We remark that in the proof of~\cite[Theorem~4.5]{ad2024tolerant}, they showed the following: if $L(X)=\Pr_{x,y\sim X}[x+y\in X]$, then $\Exp_{X}[L(X)]\geq \poly(\gamma)$ which they used to conclude there exists a good $X$. The existence of a good $X'$ was sufficient there but for $\Selfcorrection$ here, we will use the stronger statement that ``most" $X$'s are good, i.e., and use that $\Exp_{X}[L(X)]\geq \poly(\gamma)$ as we write below.}
\footnote{In \cite{ad2024tolerant}, it is shown that $\Pr_{X}[|X|/2^n \geq \gamma^2/80] \geq 1 - \exp(-0.4 \gamma^2) \geq \gamma^2/4$ for $\gamma \in [0,1]$ and $\Exp_{X}[L(X) | \,\, |X|/2^n \geq \gamma^2/80]\geq \gamma^5/10$ which implies $\Pr_{X}[L(X) \geq \gamma^5/20 | \,\, |X|/2^n \geq \gamma^2/80]\geq \gamma^5/20$ using Fact~\ref{fact:lowerboundexpectation}.}
\begin{align}
\label{eq:impliedbyAD24}
\Pr_{X}\left[[|X|/2^n \geq \gamma^2/80] \cap \Big[\Pr_{x,y\sim X}[x+y\in X]\geq \gamma^5/20\Big]\right]\geq \gamma^7/80.
\end{align}
Now, we first observe that, with probability $\geq 1-\delta$, we can find an element in $X'$  efficiently by consuming $O(1/\gamma^2\cdot \log(1/\delta))$ copies of $\ket{\psi}$. This is witnessed by the algorithm below
\begin{myalgorithm}
\begin{algorithm}[H]
    \caption{\sample($\gamma$, $\delta$)} \label{algo:sample}
    \setlength{\baselineskip}{1.5em} 
    \DontPrintSemicolon 
    \KwInput{$O(1/\gamma^2\cdot (\log 1/\delta))$ copies of $\ket{\psi}$.}
     \KwOutput{A list of Paulis $\{W_x\}$}
    Set  $M \leftarrow O((\log 1/\delta)/\gamma^{2})$ and initialize empty list $\calL \leftarrow \emptyset$\\
    \While {$m \leq M$}{
        Carry out Bell difference sampling on $\ket{\psi}^{\otimes 4}$ to produce a Pauli string $x \in \FF_2^{2n}$ \\
        Measure $\ket{\psi}^{\otimes 2}$ with respect to $W_x^{\otimes 2}$ to obtain eigenvalue outcome $b$ \\
        If $b=1$, then append $x$ to $\calL$ i.e., $\calL \leftarrow \calL \cup \{x\}$.
    }
    \Return $\calL$
\end{algorithm}
\end{myalgorithm}
We now analyze this algorithm. First Fact~\ref{claim:high_exp_value_bds} implies that step $(3)$ of the algorithm outputs an $x\in S_{\gamma/4}$ with probability $\geq 3\gamma/4$. Step $(4,5)$ essentially simulates the following: with probability $2^n p_\Psi(x)$, we include the $x\in \calL$ that was sampled in step $(3)$, else do not place it in $\calL$. Since the probability of steps $(3-5)$ succeeding is $\geq \Omega(\gamma^2)$, we set $M=O(\log (1/\delta)/\gamma^2)$ so that with probability $\geq 1-\delta$ (over the randomness of sampling in steps $(3-5)$), there exists a so-called ``good" $x\in \calL$ which lies in $X'$.\footnote{We remark that in $\calL$, there are going to be several other ``non-good" $x$s, which we choose to keep.}  
Now,  Eq.~\eqref{eq:impliedbyAD24} implies~the lemma statement.
\end{proof}
We remark in the lemma statement above, writing $S$ down might take exponential time (so below we will never choose to write it down), but with probability $\gamma^7/80$, there exists a \emph{good} $S$ that satisfies two properties $(i)$ each $x\in S$ can be sampled in polynomial time and $(ii)$ this good $S$ is approximately a subgroup, i.e.,  
$$
\Pr_{x,y\sim S}[x+y\in S]\geq \gamma^5/20.
$$
As part of algorithmizing Theorem~\ref{thm:inverse_weyl_exp_states}, we have so far discussed the subroutine $\sample$ that allows us to sample points from an approximate group with high probability. For the rest of our discussion, let us condition on the event that $\sample$ gives us elements from $S$, which happens with probability $\geq \gamma^7/80$ (Lemma~\ref{lem:sample_bds}). The corresponding distribution over elements from $\sample$ conditioned on landing in $S$, which we denote by $\Dpsi$, is given by\footnote{To see this, observe that in $\sample$, we first sample $v \sim q_\Psi$, then retain with probability $2^n p_\Psi(v)$, else discard~it.}
\begin{equation}
\label{eq:induced_dist_SAMPLE}
    \Dpsi(x) = \frac{q_\Psi(x) 2^n p_\Psi(x)}{\sum_{y \in S} q_\Psi(y)\cdot  2^n p_\Psi(y)}.
\end{equation}
We have the following useful fact regarding the values that $\Dpsi(x)$ take.
\begin{fact}
\label{fact:D_ub}
For every $x \in S$,  we have
$$ 
q_\Psi(x)\cdot \frac{\gamma}{4}\leq \Dpsi(x) \leq \frac{2^{10}\cdot 10^2}{\gamma^{10} \cdot |S|} \leq 2^{-n} \cdot \frac{8.2\cdot 10^6}{\gamma^{12}}.    
$$
\end{fact}
\begin{proof}
The lower bound is easy to see, observe that the denominator of $D_\Psi(x)$ from the expression of Eq.~\eqref{eq:induced_dist_SAMPLE} is at most $\sum_{y\in \FF^n} q_\Psi(x)=1$ and the numerator $q_\Psi(x) 2^np_\Psi(x) \geq q_\Psi(x) \cdot \gamma/4$ since for every $x\in S$, we have $2^np_\Psi(x)\geq \gamma/4$.

To prove the upper bound, we firstly note that the numerator of the expression of $\Dpsi(x)$ from Eq.~\eqref{eq:induced_dist_SAMPLE} is bounded as 
\begin{equation}\label{eq:ub_Dpsi_numerator}
q_\Psi(x) 2^n p_\Psi(x) \leq q_\Psi(x) \leq 2^{-n},    
\end{equation}
where we used $ p_\Psi(x),q_\Psi(x) \leq 2^{-n}$ by Lemma~\ref{lem:ub_qPsi}. We bound the denominator of $\Dpsi$ in Eq.~\eqref{eq:induced_dist_SAMPLE}~by
\begin{align}
\sum_{y \in S} q_\Psi(y) 2^n p_\Psi(y) &= 2^n \sum_{y \in S} p_\Psi(y) \sum_{a \in \FF_2^{2n}} p_\Psi(a) p_\Psi(a+y) \nonumber\\
&\geq 2^n \sum_{a,y \in S} p_\Psi(a) p_\Psi(y) p_\Psi(a + y) \nonumber \\
&\geq 2^n \sum_{a,y \in S} p_\Psi(a) p_\Psi(y) p_\Psi(a + y) [a+y \in S] \nonumber \\
&\geq 2^n \frac{\gamma^5}{20} |S|^2 2^{-3n} \frac{\gamma^3}{64} \\
&\geq \frac{\gamma^{10} \cdot 2^{-n} \cdot |S|}{2^{10}\cdot 10^2}
\geq \frac{\gamma^{12}}{8.2\cdot 10^6}
\label{eq:ub_Dpsi_denom}
\end{align}
where in the third line, we have used that $\Pr_{a,y \in S}[a+y \in S] \geq \gamma^5/20 \implies \sum_{a,y \in S}[a+y \in S] \geq (\gamma^5/20)|S|^2$ and that $2^n p_\Psi(x) \geq \gamma/4$ for all $x \in S$. In the fourth line, we used the lower bound on the size of $S$ from Lemma~\ref{lem:sample_bds} of $|S|\geq \gamma^2\cdot 2^n/80$. Combining Eq.~\eqref{eq:ub_Dpsi_denom},~\eqref{eq:ub_Dpsi_numerator} gives us the desired result.
\end{proof}

\subsection{BSG Test}\label{subsec:BSG_test}
In~\cite{ad2024tolerant}, after determining the approximate subgroup $S$, they used the well-known Balog-Szemeredi-Gowers ($\BSG$) theorem~\cite{balog1994statistical,gowers2001new} on $S$ to show the existence of a set $S' \subseteq S$ which has small doubling.
\begin{theorem}[$\BSG$ Theorem]\label{thm:bsg}
    Let $G$ be an Abelian group and $S \subseteq G$. If 
    \begin{equation}
    \label{eq:requiredpromiseofS}
        \Pr_{s,s' \in S}[s + s' \in S] \geq \varepsilon,
    \end{equation}
    then there exists $S' \subseteq S$ of size $|S'| \geq (\varepsilon/3) \cdot |S|$ such that $|S' + S'| \leq (6/\varepsilon)^8 \cdot |S|$.
\end{theorem}
As part of continuing to algorithmize~\cite{ad2024tolerant}, our natural next step is to then attempt to algorithmize the  BSG theorem above. However, just like $S$, the set $S'$ from the $\BSG$ Theorem is exponentially sized and thus could not have been constructed efficiently. We instead follow the strategy in \cite{tulsiani2014quadratic} by constructing a membership oracle for $S'$. 

\paragraph{Comparison to prior work.} Before we describe how to obtain a membership oracle to $S'$, we discuss first as to why our algorithmic $\BSG$ theorem is not simply a quantization of the sampling procedure of~\cite{tulsiani2014quadratic}. 
In their work, instead of defining a choice set as we have here, they consider a \emph{choice function} $\phi:\FF_2^n\rightarrow \FF_2$ chosen probabilistically and show (like how we do above) that with $\poly(\gamma)$ probability one could obtain a good $\phi$. But now, conditioned on a good choice function $\phi$, they are able to sample an element \emph{uniformly at random} from their approximate subgroup $S_\phi=\{(x,\phi(x)):x\in \FF_2^n\}$ by simply picking a uniformly random $x\in \FF_2^n$ and outputting $(x,\phi(x))$. In our setting, however, we have a \emph{choice set} $X'$ which we do not know how to \emph{uniformly sample} from and in fact the distribution of points $v\in X'$ is given by $\Dpsi$ in Eq.~\eqref{eq:induced_dist_SAMPLE}.  
Consequently, in this work we give a new  analysis to account for the fact that the points $v\in \FF^{2n}$ are sampled from $\Dpsi$ (instead of uniform from $S_\phi$). To this end, we crucially  use several properties of the graphs we are dealing with here (and their structural properties) to ensure all the arguments in algorithmizing $\BSG$ goes through even when the points are sampled from the Bell difference sampling distribution (which could be of independent interest).

\paragraph{Back to algorithmizing $\BSG$.} 
We first observe that the $\sample$ subroutine allow us to sample elements from $S'$ with high probability. To that end, it will be helpful to work with graphs associated with the approximate subgroup $S$ from Lemma~\ref{lem:sample_bds} and the set of elements obtained after calling $\sample$ multiple times. For a set $S$ (which will eventually be chosen at random), define a graph $\G(S,\E)$ on the vertex set $S$ and edge set $\E$ defined as
\begin{equation}\label{eq:graph_S}
    \E := \left\{ (x, y) : x+y \in S \text{ and } 2^np_\Psi(x), 2^np_\Psi(y), 2^np_\Psi(x+y) \geq \gamma/4 \right\},
\end{equation}
where we place an edge $(x,y)$ if the expectation values of the Weyl operators corresponding to $x,y,x+y$ are high, and $x+y \in S$. However, we will never work with this graph in practice as it is described over the set $S$, which is exponentially large. Instead, we will sample elements from $S$ with high probability using $\sample$ of Lemma~\ref{lem:sample_bds}. Let us call this resulting set $\calV$. We then define the graph $\G(\V,\E_\zeta)$ on the vertex set $\calV$ and edge set $\calE_\zeta$ for $\zeta> 0$, defined as
\begin{equation}
\label{eq:graphdefinition}
    \calE_{\zeta}:= \left\{ (x, y) : x+y \in S \text{ and } 2^np_\Psi(x), 2^np_\Psi(y), 2^np_\Psi(x+y) \geq \zeta \right\},
\end{equation}
where we place an edge $(x,y)$ if the expectation values of the Weyl operators corresponding to $x,y,x+y$ are high, and $x+y$ is in $S$, even if not in $\calV$. This respects the condition of placing an edge between two nodes in $\G(S,\E)$ which we do not have access to. Note that then for $\zeta = \gamma/4$, $\G(\V,\E_\zeta)$ is the subgraph of $\G(S,\E)$ over the nodes of $\V$. 

\subsubsection{Edge Test}\label{sec:edgetest}
Given vertices $x,y$, the goal of $\edgetest$ is to decide membership of $(x,y)$ in the edge set of $\G(\V,\E_\zeta)$, for which we use Algorithm~\ref{algo:edge_test} below.
\begin{myalgorithm} \setstretch{1.35}
\begin{algorithm}[H]
    \label{algo:edge_test}
    \setlength{\baselineskip}{1.5em} 
    \DontPrintSemicolon 
    \caption{\edgetest($x$, $y$, $\zeta$, $\zeta'$, $\delta$)} 
    \KwInput{Vertices $x, y \in \FF_2^{2n}$, vertex set $\calV$, error parameters $\zeta_1,\zeta_2$, failure probability $\delta$}
    \KwOutput{Presence of edge $(x,y)$}
    Set $T \leftarrow O(1/{\zeta'}^{2} \log(1/\delta))$ \\
    Obtain $\zeta'$-approximate estimates $\alpha_x$ of $2^np_\Psi(x)$, $\alpha_y$ of $2^np_\Psi(y)$, $\alpha_{x+y}$ of $2^np_\Psi(x+y)$ by measuring $\ket{\psi^{\otimes 2}}$ with respect to $W_x^{\otimes 2}$, $W_y^{\otimes 2}$, $W_{x+y}^{\otimes 2}$ respectively, using $T$ copies of $\ket{\psi}$. \\
    \uIf{$\alpha_x,\alpha_y,\alpha_{x+y} \geq \zeta$}{
        Measure $\ket{\psi}^{\otimes 2}$ with respect to $W_{x+y}^{\otimes 2}$ to obtain eigenvalue outcome $b$ \\
        Set $\mathrm{FLAG} \leftarrow b$.
    }
    \Else{Set $\mathrm{FLAG} \leftarrow 0$.}
    \Return $\mathrm{FLAG}$
\end{algorithm}
\end{myalgorithm}
Note that the algorithm for $\edgetest$ could have checked in $x+y$ is in the vertex set $\calV$ that has been sampled so far before attempting to check if it would have been retained with probability $2^n p_\Psi(x+y)$ in the graph $\G(S,\E)$. However, for brevity, we omit this extra check and just repeat even if $x+y \in \calV$. Finally, due to the \emph{approximate} nature of evaluating $2^n p_\Psi(x)$ for any $x \in \FF_2^{2n}$, we work with graphs $\calG^1(\calV,\calE_{\zeta + \zeta'})$ and $\calG^2(\calV,\calE_{\zeta- \zeta'})$ whose edge sets satisfy $\calE_{\zeta + \zeta'} \subseteq \calE_{\zeta - \zeta'}$.
\begin{claim}
    Given $\zeta,\zeta',\delta > 0$, the output of $\edgetest(x, y, \zeta, \zeta', \delta)$ with $M=O(1/\zeta_2^{ 2} \log(1/\delta))$ copies of $\ket{\psi}$, satisfies the following guarantee with probability at least $1-\delta$
    \begin{enumerate}[(i)]
        \item $\edgetest(x, y, \zeta, \zeta', \delta) = 1 \implies (x,y) \in \calE_{\zeta_1- \zeta_2}$.
        \item $\edgetest(x, y, \zeta, \zeta', \delta) = 0 \implies (x,y) \in \calE_{\zeta_1+ \zeta_2}$.
    \end{enumerate}
\end{claim}
\begin{proof}
The claim follows from the definitions of the edge sets $\calE_{\zeta - \zeta'}$, $\calE_{\zeta + \zeta'}$ and using Fact~\ref{fact:hoeffding_sampling}.    
\end{proof}

\subsubsection{Setup of $\BSG$ Test}
\label{subsec:description_bsg_test}
For a vertex $u\in \calV$ and $\zeta > 0$, we define $N_\zeta(u)$ as the set of neighbors of $u$ in the graph $\G(S,\E_\zeta)$:
\begin{equation}
\label{eq:bsg_set_N}
N_\zeta(u) := \left\{v \in S : (u,v) \in \calE_\zeta \right\}.    
\end{equation}
Note that $|N_\zeta(u)|$ is monotonically decreasing with increasing $\zeta$.  With this definition of $N_\zeta(\cdot)$, we first show that if we sample a $u$ from the $D_\Psi$ distribution, the neighborhood of $u$ has a large size (i.e., $|N_\zeta(u)|$ is large) and in particular, the \emph{weight} under $D_\Psi$, in this neighborhood is large. We make this formal in the lemma below. We prove this in Appendix~\ref{appsec:proof_bsg}.\footnote{We remark that in this section, we explicitly specify various constant factors since the parameters of the $\BSG$ test needs specification of these parameter settings to prove its correctness.}
\begin{restatable}{lemma}{sizeNu}
\label{lem:Dpsilowerbound}
For $\zeta \in (0,\gamma/4]$, we have the following
$$
\Exp_{u\sim D_\Psi}\Big[\sum_{v\in N_{\gamma/4}(u)}D_\Psi(v)\Big]\geq  \gamma^{64}/(2^{39} \times 10^{15}), \text{ and } \Exp_{u\sim D_\Psi}[|N_{\zeta}(u)|]\geq  \gamma^{74}/(2^{49} \times 10^{17}),
$$
\end{restatable}
The proof of this lemma relies on the combinatorial structure of $S$, in particular we use that there exists a large subset $A'$ of $S$ that has many \emph{length-$3$ paths} (i.e., many $a,x,y,b$ with $a,b \in A'$ and $x,y \in S$ such that there are edges between every neighboring pair of elements).  Next, within $N_\zeta(u)$, we define a subset of vertices $Q_\zeta(u)$ as the set $\{v_1,\ldots,v_k\}\subseteq N_\zeta(u)$ which \emph{do not} have many neighbors in common (where ``many" here is weighted by the underlying distribution $\Dpsi$). For example, we discard a $v_i\in N(u)$ if there are ``many" $w\in N(u)$ such that $|N(v_i)\cap N(w)|$ is ``small" (when defining all measures with respect to $\Dpsi$). More formally, for thresholds $\rho_1,\rho_2 \in (0,1)$, we~define
\begin{equation}\label{eq:bsg_set_S}
Q_\zeta(u) := \left\{v \in N_\zeta(u) : \Pr_{v_1\sim \Dpsi}\left[v_1 \in N_\zeta(u) \text{ and } \Pr_{v_2\sim \Dpsi} \left[v_2 \in N_\zeta(v) \cap N_\zeta(v_1) \right] \leq \rho_1 \right] > \rho_2 \right\}.
\end{equation}
Finally, define $T_\zeta(u) := N_\zeta(u) \backslash Q_\zeta(u)$ which contains all the vertices in $N_\zeta(u)$ that have many neighbors in common. Formally, 
\begin{align}
\label{eq:bsg_set_T}
T_\zeta(u) &:= N_\zeta(u) \setminus Q_\zeta(u) \nonumber \\
&:= \left\{v \in N_\zeta(u) : \Pr_{v_1\sim \Dpsi}\left[v_1 \in N_\zeta(u) \text{ and } \Pr_{v_2\sim \Dpsi} \left[v_2 \in N_\zeta(v) \cap N_\zeta(v_1) \right] \leq \rho_1 \right] \leq \rho_2 \right\}
\end{align}
By working with these graphs, our goal is to  leverage the characterization of a small doubling set shown by Sudakov et al.~\cite{sudakov2005question}. Particularly, they showed if the graph $\G(S,\E)$ (Eq.~\eqref{eq:graph_S}) has high density, then $T(u)$ for a random vertex $u \in S$ is a good choice for a set with small doubling, with high probability.  
However, like we mentioned earlier, unfortunately in our setting, we are not able to obtain \emph{uniform} samples from $S$ (which is possible in~\cite{tulsiani2014quadratic}) and so instead we also adapt the proof of~\cite{sudakov2005question}. 
Additionally, to avoid holding $T(u)$ in memory, we will not  be building the entire set $T(u)$ and instead construct a \emph{membership oracle} for the set $T(u)$ using Algorithm~\ref{algo:bsg_test} called $\bsgtest$ which we present now before giving its analysis.

\subsubsection{Algorithm of BSG test}
\label{sec:algo_BSG}
In Algorithm~\ref{algo:bsg_test}, we present the procedure used to decide membership in $T_\zeta(u)$. In particular, this is done by checking the condition of Eq.~\eqref{eq:bsg_set_T}. To account for the approximate nature of the $\edgetest$ and estimation of the thresholds corresponding to $T(u)$, we will in fact run the $\bsgtest$ and Algorithm~\ref{algo:bsg_test} with a set of parameters as we now describe.

\paragraph{Parameters of $\BSG$ test.} We now describe the parameters  used in Algorithm~\ref{algo:bsg_test}. The intuition for these parameters will be clear in the proof. We set $\rho = \gamma^5/20$, which is a lower bound on the density of the graph $\G(S,\E_{\gamma/4})$. Accordingly, we then set thresholds corresponding to $T(u)$ (Eq.~\eqref{eq:bsg_set_T}) as $\rho_1 = \gamma^{350}/(10240 C_1^3 C_2^5)$, and $\rho_2 = 9\gamma^{202}/(2560 C_1 C_2^3)$ for constants $C_1=2^{10}\cdot 10^2, C_2 = 2^{39} \times 10^{15}$. Given $\delta$, we take $r$ and $s$ to be $\poly(1/\rho, \log(1/\delta))$ so that with probability at least $1-\delta$, the error in the estimates of $Y_{k \ell}$ and $Z_{k \ell}$ computed as part of Algorithm~\ref{algo:bsg_test} is at most $\rho_1/100$. We also ensure that the error in all estimates used by $\edgetest$ is at most $\rho_1/100$.

To choose $\zeta_1,\zeta_2,\zeta_3$, we follow an approach inspired by \cite[Lemma 4.10]{tulsiani2014quadratic}. Let $\rho_3 = \gamma^{349}/(2560 C_1^3 C_2^5)$. Consider the interval $[\gamma/180, \gamma/18]$ and divide it into $1/\rho_3$ equal and consecutive sub-intervals of size $(\gamma \cdot \rho_3)/20$ each. To select $\zeta,\mu$, we randomly choose one of these $1/\rho_3$ many sub-intervals and choose positive parameters $\zeta,\mu$ so that $\zeta-\mu$ and $\zeta + \mu$ are the endpoints of this interval. The chosen sub-interval is then of the form $[\zeta-\mu, \zeta+\mu]$ with $\zeta$ chosen to be the center of this sub-interval and $\mu$, half the width of this sub-interval. We then set $\zeta_1 = \zeta_3 = \zeta + \mu/2$ and $\zeta_2 = \zeta - \mu/2$.
To see how the choice of error parameters influence the definition of $T(u)$ (Eq.~\eqref{eq:bsg_set_T}) and also deal with the approximate nature of the test, we revise our notation in defining the~set 
\begin{align}
\label{eq:approx_sets_in_bsg}
&T(u, \zeta_1, \zeta_2, \zeta_3, \rho_1, \rho_2) \\
&= \left\{ v \in N_{\zeta_1}(u) : \Pr_{v_1 \sim \Dpsi}\left[v_1 \in N_{\zeta_2}(u) \text{ and } \Pr_{v_2 \sim \Dpsi} \left[v_2 \in N_{\zeta_3}(v) \cap N_{\zeta_3}(v_1) \right] \leq \rho_1 \right] \leq \rho_2 \right\}.
\end{align}

\begin{myalgorithm} 
\setstretch{1.35}
\begin{algorithm}[H]
    \label{algo:bsg_test}
    \caption{\bsgtest($u$, $v$, $\zeta_1$, $\zeta_2$, $\zeta_3$, $\rho_1$, $\rho_2$, $\delta$)} 
    \setlength{\baselineskip}{1.5em} 
    \DontPrintSemicolon 
    \KwInput{Vertices $u$ and $v$, error parameters $\zeta_1,\zeta_2, \zeta_3$, thresholds $\rho_1,\rho_2$, failure prob.~$\delta$}
    \KwOutput{Flag $F$ indicating if $v \in T(u)$~(defined in Eq.~\eqref{eq:bsg_set_T}) or not}
    
    Set $r, s \leftarrow \poly(1/\gamma \log(1/\delta))$, $\zeta' = \rho_1/100$, $\delta'=\delta \poly(\gamma)$. \\
    Obtain a set of $r$ samples $\{z^{(k)}\}_{k \in [r]} \leftarrow \sample(r, \zeta_1, \zeta_2, \delta)$ \\
    For each $k \in [r]$, obtain a set of $s$ samples $\{w^{(k,\ell)}\}_{\ell \in [s]} \leftarrow \sample(s, \zeta_1, \zeta_2, \delta)$ \\
    If $\edgetest(u,v,\zeta_1, \zeta', \delta')=0$, then set output $F \leftarrow 0$ \label{algo_step:edge_test_u_v} \\
    \For{$k \in [r]$}{
        $X_k \leftarrow \edgetest(u, z^{(k)}, \zeta_2, \zeta', \delta')$ \\
        \For{$\ell \in [s]$}{
            $Y_{k\ell} \leftarrow \edgetest(v, w^{(k, \ell)}, \zeta_3, \zeta', \delta')$ \\
            $Z_{k\ell} \leftarrow \edgetest(z^{(k)}, w^{(k, \ell)}, \zeta_3, \zeta', \delta')$ \\
        }
        \eIf{$\frac{1}{s} \sum_{\ell=1}^s Y_{k \ell} \cdot Z_{k \ell} \leq \rho_1$}
        {$B_k \leftarrow 1$}
        {$B_k \leftarrow 0$}
    }
    Set $F \leftarrow 1$ if $\frac{1}{r} \sum_{k=1}^r X_k \cdot B_k \leq \rho_2$ and $F \leftarrow 0$ otherwise. \\
    \Return $F$
\end{algorithm}    
\end{myalgorithm}

\paragraph{Guarantee of {$\BSG$} Test.}
The guarantee of the $\BSG$ test is described in the following theorem. 
\begin{restatable}{theorem}{bsgtestlemma}
\label{lemma:analysis_bsg}
Let $\delta > 0$ and parameters $\rho_1,\rho_2,r,s$ be chosen as in Section~\ref{sec:algo_BSG}. For every $u \in S$ and choice of $\zeta_1,\zeta_2,\zeta_3$ as described above, there exist two sets $A^{(1)}(u) \subseteq A^{(2)}(u)$ defined as follows
\begin{align*}
    A^{(1)}(u) := T(u,\zeta + \mu,\zeta-\mu,\zeta+\mu,\rho_1,\rho_2),\quad A^{(2)}(u) := T(u,\zeta,\zeta,\zeta,10\rho_1/11,10\rho_2/9),
\end{align*}
where $\mu = \zeta_1 - \zeta_2$, such that the output of $\bsgtest$ satisfies the following with probability $\geq 1-\delta$
\begin{enumerate}[(i)]
    \item $\bsgtest(u, v, \zeta_1, \zeta_2, \zeta_3, \rho_1, \rho_2, \delta)=1 \implies v \in A^{(2)}(u)$.
    \item $\bsgtest(u, v, \zeta_1, \zeta_2, \zeta_3, \rho_1, \rho_2, \delta)=0 \implies v \notin A^{(1)}(u)$.
\end{enumerate}
Furthermore, with probability at least $\Omega(\gamma^{487})$ over $u \sim \Dpsi$ and $\zeta_1,\zeta_2,\zeta_3$, we~have
$$
|A^{(1)}(u)| \geq \Omega(\gamma^{138}) \cdot |S| \text{ and } |A^{(2)}(u) + A^{(2)}(u)| \leq O(\gamma^{-932}) \cdot |S|. 
$$
In particular, this implies $|A^{(1)}(u)|/|A^{(2)}(u)|\geq \Omega(\gamma^{1070})$.
\end{restatable}
On a high level, the proof of Theorem~\ref{lemma:analysis_bsg} follows similar to that in \cite{tulsiani2014quadratic} except the following: in~\cite{tulsiani2014quadratic}, their definitions of the sets $N_\zeta(\cdot),T_\zeta(\cdot),Q_\zeta(\cdot)$ (that we defined in the start of this section) were defined with respect to the uniform distribution over $S$ whereas we need to account for $u$ sampled from the $\Dpsi$ distribution. To that end, we modify the correctness analysis of their main theorem proof and account for the modified tests and parameters used here. 

We emphasize that the core challenge in redoing their analysis is, for~\cite{tulsiani2014quadratic} they know that for every $u$, the value of the distribution on the point $u$ was $1/|S|$, whereas for us, we only have an \emph{upper bound} for every $D_\Psi(u)$ but not a lower bound for every element $u \in S$. Instead, we use several properties of these sets and show that the cumulative weight on these sets (under $D_\Psi$) has a large-enough lower bound, with which we prove the theorem. Additionally, when using the $\bsgtest$ as part of our eventual $\Selfcorrection$ protocol, we will require that Bell difference sampling concentrates on the set $A^{(2)}(u)$. We show that Bell difference sampling in fact concentrates on $A^{(1)}(u) \subseteq A^{(2)}(u)$ for a good $u$ (which we sample with $\poly(\gamma)$ probability as in the theorem above).
To this end, our main contribution is proving  the following lemma.\footnote{We remark that there are several polynomial factors in $\bsgtest$, that we believe can be optimized further. Additionally, it may be possible to improve the polynomial-factors using an algorithmic version of the $\BSG$ theorem in~\cite{schoen2015new} as attempted recently in \cite{neumann2025adaptive} for the classical setting.}
\begin{restatable}{lemma}{lemmacoreBSG}
\label{lemma:lb_size_A1}
Consider the context of Theorem~\ref{lemma:analysis_bsg}. Let $\gamma'=\gamma^{64}/(2^{41}\cdot 10^{15})$.  Define $H_{\gamma'}(u)$  as
\begin{equation}\label{eq:set_Hu}
    H_{\gamma'}(u) := \{ v \in N_{\zeta + \mu}(u) : D_\Psi(v) \geq \gamma' \cdot |N_{\zeta + \mu}|^{-1} \}.
\end{equation}
Then, with probability 
$\Omega(\gamma^{487})$ over $u \sim \Dpsi$ and parameters $\zeta_1,\zeta_2,\zeta_3$, we have
\begin{enumerate}[$(i)$]
    \item $|H_{\gamma'}(u)| \geq \Omega(\gamma^{74}) \cdot |S|$,
    \item $|A^{(1)}(u)| \geq |A^{(1)}(u) \cap H_{\gamma'}(u)| \geq \Omega(\gamma^{138})\cdot|S|$,
    \item $\sum_{v\in A^{(1)}(u)}D_\Psi(v) \geq \Omega(\gamma^{202})$.
\end{enumerate}
\end{restatable}
\begin{proof}We provide the proof of $(i),(iii)$ here and defer the proof of $(ii)$ to shown in Appendix~\ref{appsec:proof_bsg} (since it is similar to~\cite{tulsiani2014quadratic}).  We first show $(i)$. Using
 Fact~\ref{fact:lowerboundexpectation} and Lemma~\ref{lem:Dpsilowerbound}, we have that
\begin{align}\label{eq:Donvislargewhensampled}
\Pr_{u\sim D_\Psi}\Big[\sum_{v\in N_{\zeta + \mu}(u)}D_\Psi(v)\geq \gamma^{64}/(2C_2)\Big]\geq \gamma^{64}/(2 C_2),
\end{align}
where $C_2 = 2^{39} \times 10^{15}$. Let $\gamma' = \gamma^{64}/(4 C_2)$. For a good $u$ for which the event of Eq.~\eqref{eq:Donvislargewhensampled} is true, we then observe that
\begin{align*}
2 \gamma' \leq \sum_{v\in N_{\zeta + \mu}(u)}D_\Psi(v)&=\sum_{v\in H_{\gamma'}(u)}D_\Psi(v) + \sum_{v\in N_{\zeta + \mu}(u)\backslash H_{\gamma'}(u)}D_\Psi(v)\\
&\leq |H_{\gamma'}(u)|\cdot \frac{C_1}{\gamma^{10} \cdot |S|} + \gamma' \cdot |N_{\zeta + \mu}|^{-1} \cdot |N_{\zeta + \mu}(u) \backslash H_{\gamma'}(u)|\\
&\leq |H_{\gamma'}(u)|\cdot \frac{C_1}{\gamma^{10} \cdot |S|} + \gamma',
\end{align*}
where we used the upper bound of $D_\Psi$ from Fact~\ref{fact:D_ub} (with $C_1 = 2^{10} \cdot 10^2$) and definition of $H(u)$ (Eq.~\eqref{eq:set_Hu}) in the second line. This implies that for a good $u$
\begin{equation}\label{eq:HuislargeifDpsilarge}
|H_{\gamma'}(u)|\geq \gamma^{74}/(4 C_1 C_2) \cdot |S|.    
\end{equation}
We can then lower bound
\begin{align}\label{eq:exp_Hu_lb}
\Exp_{u\sim D_\Psi}[|H_{\gamma'}(u)|]&\geq \Exp_{u\sim D_\Psi}\Big[|H_{\gamma'}(u)|\Big|\sum_{v\in N(u)}D_\Psi(v)\geq 2\gamma'\Big]\cdot \Pr_{u\sim D_\Psi}\Big[\sum_{v\in N(u)}D_\Psi(v)\geq 2\gamma'\Big]\\
&\geq  \Exp_{u\sim D_\Psi}\Big[|H_{\gamma'}(u)|\Big|\sum_{v\in N(u)}D_\Psi(v)\geq 2\gamma'\Big]\cdot (2\gamma')\\
&\geq \gamma^{138}/(4C_1 C_2^2) \cdot |S|,
\end{align}
where we used Eq.~\eqref{eq:Donvislargewhensampled} in the second inequality and third inequality used Eq.~\eqref{eq:HuislargeifDpsilarge}.

We now prove item $(iii)$ assuming $(ii)$. We lower bound the weight under the distribution $D_\Psi$ over $A^{(1)}(u)$ by considering its overlap with $H_{\gamma'}(u)$ as follows:
\begin{align}
\Exp_{u\sim D_\Psi}\Big[\sum_{v\in A^{(1)}(u)}D_\Psi(v)\Big]
&\geq \Exp_{u\sim D_\Psi}\Big[\sum_{v\in A^{(1)}(u)\cap H_{\gamma'}(u)} D_\Psi(v)\Big] \\
&\geq \gamma' \cdot |S|^{-1} \cdot \Exp_{u\sim D_\Psi}\Big[|A^{(1)}(u)\cap H_{\gamma'}(u)|\Big]\\
&\geq \gamma' \cdot |S|^{-1} \cdot \frac{\gamma^{138}}{8 C_1 C_2^2} \cdot |S| \\
&\geq \frac{\gamma^{202}}{32 C_1 C_2^3}.
\end{align}
In particular, for a good $u\sim D_\Psi$ and parameters $\zeta_1,\zeta_2,\zeta_3$, which we sample with probability $\geq \Omega(\gamma^{487})$, we have that 
\begin{align}
\sum_{v\in A^{(1)}(u)}D_\Psi(v)\geq \gamma^{202}/(64 C_1 C_2^3),
\end{align}
 concluding the proof of the lemma.
\end{proof}
Using Lemma~\ref{lemma:lb_size_A1}, we then show that the set $A^{(2)}(u)$ as defined in Theorem~\ref{lemma:analysis_bsg} is also pretty large with high probability over the choice of $u$ (from $D_\Psi$) since $A^{(1)}(u) \subseteq A^{(2)}(u)$. We defer the proof of this statement to Appendix~\ref{appsec:proof_bsg}.
\begin{restatable}{claim}{smalldoublingAtwo}
\label{claim:doubling_A2}
Consider the context of Theorem~\ref{lemma:analysis_bsg} and let $|A^{(1)}(u)| \geq \Omega(\gamma^{138}) \cdot |S|$.~Then,  
$$
|A^{(2)}(u)+A^{(2)}(u)| \leq O(1/\gamma^{932}) \cdot |A^{(2)}(u)|
$$
\end{restatable}
With these two claims, we are now ready to complete the proof of Theorem~\ref{lemma:analysis_bsg}.
\begin{proof}[\textbf{Proof of Theorem~\ref{lemma:analysis_bsg}}]
We present the algorithm for deciding membership in $A^{(1)}(u)$ and $A^{(2)}(u)$ in Algorithm~\ref{algo:bsg_test}. Using the definitions of $A^{(1)}(u)$ and $A^{(2)}(u)$ and our earlier observation  that $A^{(1)}(u) \subseteq A^{(2)}(u)$. Moreover, by Lemma~\ref{lemma:lb_size_A1} and Claim~\ref{claim:doubling_A2}, we have that with probability $\Omega(\gamma^{487})$ over the choice of $u \sim D_\Psi$ and parameters $\zeta_1,\zeta_2,\zeta_3$, we obtain
$$
|A^{(1)}(u)| \geq \Omega(\gamma^{138}) |S|, \qquad |A^{(2)}(u)+A^{(2)}(u)| \leq O(1/\gamma^{932}) \cdot |A^{(2)}(u)|.
$$
It then just remains to argue that for the choice of thresholds $\rho_1, \rho_2$ the lemma statement holds. 
For a given failure probability $\delta > 0$ and by choosing $r,s = \poly(1/\gamma \log(1/\delta))$, we ensure that for a given $v \in N_{\zeta}(u)$ and $v_1 \in N_{\zeta}(u)$, the estimate of $\Pr_{v_2 \sim \Dpsi} \left[v_2 \in N_{\zeta+\mu}(v) \cap N_{\zeta+\mu}(v_1) \right]$ is at most $\gamma^{350}/(10^6 C_1^3 C_2^5)$ or $\rho_1/8$. Similarly, the estimate of $\Pr_{v_1 \sim \Dpsi}\left[v_1 \in N_{\zeta - \mu}(u) \text{ and }  [(v,v_1) \text{ is bad}] \right]$ is at most  $\gamma^{350}/(10^6 C_1^3 C_2^5)$ or $\rho_2/10$.
We also choose the error parameter corresponding to the $\edgetest$ to be $\gamma^{350}/(10^6 C_1^3 C_2^5)$ to ensure that elements are appropriately placed in $N_{\zeta - \mu}(u)$ and $N_{\zeta + \mu}(u)$ where recall that $\mu = (\gamma \cdot \rho_3)/40 = \gamma^{350}/(51200 C_1^3 C_2^5)$.
\end{proof}

\subsection{Symplectic Gram-Schmidt procedure}\label{sec:symplectic_gram_schmidt}
So far, the $\bsgtest$ test produced a membership oracle for a subset which is dense and has small doubling constant. By the polynomial Freiman Ruzsa theorem (Theorem~\ref{thm:marton_conjecture}), this implies that this subset can be covered by polynomially many cosets of the so-called ``Freiman subgroup". Moreover,~\cite{ad2024tolerant} showed that this subgroup has a small stabilizer covering and a stabilizer state corresponding to one of these stabilizer subgroups in the covering has high fidelity with $\ket{\psi}$. In order to find a stabilizer state whose stabilizer subgroup has high overlap with this Freiman subgroup, we now introduce the symplectic Gram-Schmidt (SGS) procedure (Algorithm~\ref{alg:sym_gram_schmidt})~\cite{wilde2009logical,silva2001lectures,fattal2004entanglement} that will enable us in determining the stabilizer covering of the Freiman subgroup. We will discuss this again later when we give our $\Selfcorrection$ protocol by combining the subroutines discussed here, but the above gives the context for discussing the SGS procedure.

We follow the presentation of this procedure as given by Wilde~\cite{wilde2009logical}. The SGS procedure takes as input,  a generating set of Paulis of a subgroup $V$ and outputs a new generating set that is in a canonical form, i.e., a form where each generator anti-commutes with at most one other generator. In particular, Algorithm~\ref{alg:sym_gram_schmidt} will output 
\begin{enumerate}
    \item $C_V$ which will be a basis for $V\cap V^\perp$ (which is an isotropic subspace by definition)
    \item $A_V$ which will be a basis for $V\backslash \langle C_V\rangle$.
\end{enumerate}
Put together, $C_V \cup A_V$ forms a new generating set for $V$. 
\begin{fact}[Symplectic Gram-Schmidt procedure]
Given a generating set of the subgroup $V$ of dimension $m \leq 2n$, Algorithm~\ref{alg:sym_gram_schmidt} outputs a basis $C_V$ of $V\cap V^\perp$ and a basis $A_V$ of the quotient $V \setminus \la C_V \ra$ in time $O(nm^2)$ time.
\end{fact}
\begin{myalgorithm} \setstretch{1.35}
\begin{algorithm}[H]
    \label{alg:sym_gram_schmidt}
    \caption{\symgramschmidt($g_1,g_2,\dots,g_m$)} 
    \setlength{\baselineskip}{1.5em} 
    \DontPrintSemicolon 
    \KwInput{List of generators of a subgroup $V$= $\{g_1,g_2,\dots,g_m\}$}
    \KwOutput{Basis $C_V$ of the center of group $V$ i.e., $\la C_V \ra = V \cap V^\perp$ and basis $A_V$ of the quotient i.e., $\la A_V \ra = V \setminus \la C_V \ra$.}
    
    Initialize empty lists $C_V \leftarrow \emptyset$ and $A_V \leftarrow \emptyset$ \\
    Initialize list $G \leftarrow \{g_1,g_2,\dots,g_{m}\}$ \\
    \While{$G \neq \emptyset$}{
        Remove the first element $g$ from $G$ \\
        Find the first element of $G$ that anticommutes with $g$. If it exists, call it $h$ and remove it from $G$. \\
        \uIf{$h$ exists}{
            Update $A_V \leftarrow A_V \cup \{g,h\}$ \\
            \For{$x \in G$}{
                If $[x,h]=0$, set $a \leftarrow 0$ else set $a \leftarrow 1$. \\
                If $[x,g]=0$, set $b \leftarrow 0$ else set $b \leftarrow 1$. \\
                Replace $x \leftarrow x g^a h^b$. \\
            }
        }
        \Else{
            If $g \neq I$, append $g$ to $C_V$ i.e., $C_V \leftarrow C_V \cup \{g\}$. \\
        }
    }
    \Return $C_V$ and $A_V$
\end{algorithm} 
\end{myalgorithm}
Given a subgroup $V \subseteq \FF_2^{2n}$ of Paulis, we know that there exists a Clifford unitary $U$ that maps it to a canonical form as stated in the following fact  where we denote $\calP^k$ as the set of $k$-qubit Paulis and we denote $\calP^{m}_{\calZ} = \{I,Z\}^{\otimes m}$.
\begin{fact}\label{fact:clifford_action_on_group}(\cite{fattal2004entanglement}) Suppose $V$ is a subgroup of the $n$-qubit Pauli group. Then, there exists $m+k \leq n$ and an $n$-qubit Clifford $U$~such~that
\begin{equation}
    U V U^\dagger = \la Z_1, X_1, \ldots, Z_k, X_k, Z_{k+1}, Z_{k+2}, \ldots, Z_{k+m} \ra = \calP^k \times \calP_\calZ^m.
\end{equation}
\end{fact}

We now claim that the Clifford unitary in the fact above can  be determined efficiently.
\begin{claim}\label{claim:clifford_synthesis_subgroup}
    Suppose $V$ is a subgroup of the $n$-qubit Pauli group. The Clifford unitary $U$ that takes $V$ to the canonical form of Fact~\ref{fact:clifford_action_on_group} can be determined in $O(n^3)$ time.
\end{claim}
\begin{proof}
We first run Algorithm~\ref{alg:sym_gram_schmidt} on $V$ to produce a basis $C_V = \{s_1,\ldots,s_m\}$ (for some $m$) of  $V \cap V^\perp$ and a basis $A_V=\{(g_i,h_i)\}_{i \in [k]}$ (for some $k$) of the remaining elements outside $\langle C_V\rangle$, such that the pair $g_i,h_i$ anti-commute with each other but commute with all the other generators in $C_V$ or $A_V$. This takes time $O(n^3)$. Using Lemma~\ref{lem:clifford_isotropic_subspace} on the basis $C_V$ (which spans an isotropic subspace of dimension $m$), we can determine a Clifford circuit $U_1$ of gate complexity $O(n^2)$ in time $O(n^3)$~such~that
$$
U_1 \la C_V \ra U_1^\dagger = \la Z_{k+1}, Z_{k+2}, \ldots, Z_{k+m} \ra.
$$
So far, we have then transformed the group $V$ under the action of the Clifford circuit $U_1$ as
$$
U_1 V U_1^\dagger = \la U_1 g_1 U_1^\dagger, U_1 h_1 U_1^\dagger, \ldots, U_1 g_k U_1^\dagger, U_1 h_k U_1^\dagger, Z_{k+1}, \ldots, Z_{k+m} \ra.
$$
Denoting $g_i' = U_1 g_i U_1^\dagger$ and $h_i'=U_1 h_i U_1^\dagger$ for all $i \in [k]$. The time complexity of determining these new generators is $O(n^3)$, consuming $O(n^2)$ time for each generator (i.e., computing $UgU^\dagger$ takes trivially $O(n^2)$ time). We note that the commutation relations of $g_i$ (or $h_i$) with all the other generators in $C_V$ or $A_V$ remain unchanged under the action of $U_1$ i.e., $[a', b'] = [a, b]$ for every $a,b \in C_V \cup A_V$ and where $a'=U_1 a U_1^\dagger, b' = U_1 b U_1^\dagger$. In other words, $g_i'$ anti-commutes with $h_i'$ and commutes with everything else. Moreover, we can multiply each $g_i'$ (or $h_i'$) with combinations of $\{Z_{k+1},\ldots,Z_{k+m}\}$ such that $g_i'$ (or $h_1'$) involve identity on qubits $k+1,\ldots, k+m$, and without changing their commutation relations as $\{Z_{k+1},\ldots,Z_{k+m}\}$ commute with all of them.

Now, starting with the first anti-commuting pair of Paulis $(g_1',h_1')$, we use Lemma~\ref{lem:clifford_anticomm_paulis} to determine another Clifford circuit $W_1$ of gate complexity $O(n)$ in time $O(n)$ such that $W_1 g_1' W_1^\dagger = X_1$ and $W_1 h_1' W_1^\dagger = Z_1$. This circuit $W_1$ will not act on qubits $k+1,\ldots,k+m$ due to our earlier action of ensuring the Paulis $g_i',h_i'$ act trivially as the identity on these qubits. We have so far then transformed the group $V$ under the action of $W_1 U_1$ as
$$
W_1 U_1 V U_1^\dagger W_1^\dagger = \la X_1, Z_1, W_1 g_2' W_1^\dagger, W_1 h_2' W_1^\dagger, \ldots, W_1 g_k' W_1^\dagger, W_1 h_k' W_1^\dagger, Z_{k+1}, \ldots, Z_{k+m} \ra.
$$
Considering the new set of Paulis $g_i'' = W_1 g_i' W_1^\dagger$ and $h_i'' = W_1 h_i' W_1^\dagger$ for all $i \in \{2,\ldots,k\}$ which again have the same commutation relations as before, we can repeat the process above of ensuring that these new Paulis only act trivially on qubits $1,k+1,\ldots,k+m$ and then use Lemma~\ref{lem:clifford_anticomm_paulis} to determine the next Clifford circuit $W_2$ that has the desired action of transforming $g_2'',h_2''$ to $X_2,Z_2$. We repeat this until all the pairs of anti-commuting Paulis are exhausted, obtaining a new Clifford circuit $W_i$ for $i \in [k]$ in each step. This takes $k \leq n$ rounds. The desired Clifford unitary is then $U=W_k W_{k-1} \ldots W_1 U_1$ which transforms $V$ to the canonical form of Fact~\ref{fact:clifford_action_on_group} as
$$
U V U^\dagger = \la Z_1, X_1, \ldots, Z_k, X_k, Z_{k+1}, Z_{k+2}, \ldots, Z_{k+m} \ra.
$$
Overall, we consumed $O(n^3)$ time in determining $U_1$, $O(n)$ time in determining each $W_i$ for $i \in [k]$ where $k \leq n$ anti-commuting pairs, $O(n^3)$ time in determining the intermediate transformed generators after obtaining $U_1$ and $O(n^2)$ time in determining the intermediate transformed generators after obtaining each $W_i$ for each $i \in [k]$. The total time complexity is then $O(n^3)$.
\end{proof}

\section{Algorithm for Self-Correction}
\label{sec:selfcorrectionprotocol}
In this section, we now prove Theorem~\ref{thm:self_correction} using the subroutines that we have discussed earlier (Section~\ref{sec:subroutines}). The learning algorithm  on a high-level follows the following three- step approach
\begin{enumerate}
    \item First, given copies of $\ket{\psi}$ we use $\sample$ and the $\bsgtest$ to sample elements from $A^{(2)}$ which will also lie in $A^{(1)}$ with high probability. We prove this in Section~\ref{sec:determinesubgroup}. 
    \item We next use the algorithmic $\PFR$ conjecture~\ref{conj:algopfrconjecture} to give a polynomial time procedure that uses samples from the first step to find a ``small" subgroup $V$ whose weight is still high (i.e., $\Exp_{x\in V}[2^np_\Psi(x)]\geq \gamma$) and furthermore $|V|\sim 2^n\cdot \poly(\gamma)$. We prove this in Section~\ref{sec:succintgroup}.
    \item Given access to subgroup $V$, copies of $\ket{\psi}$, using the $\symgramschmidt$ and other~subroutines, we find a stabilizer $\ket{\phi}$ such that $|\langle \psi|\phi\rangle|^2\geq \poly(\gamma)$. We prove this in Section~\ref{sec:findstabilizer}.
\end{enumerate}
\noindent In Section~\ref{sec:puttingtogethereverything} we put all of these steps  together and prove our main theorem. 
\begin{theorem}
\label{thm:restatementofselfcorr}
Let $\gamma>0$, $C > 1$ be a universal constant and $\eta = \Omega(\gamma^C)$. Let~$\ket{\psi}$ be an unknown $n$-qubit quantum state such that $\Exp_{x \sim q_\Psi}[|\la \psi | W_x | \psi \ra|^2] \geq \gamma$. Assuming the algorithmic $\PFR$ conjecture, there is a protocol that with probability $1-\delta$, outputs a $\ket{\phi}\in \Sh$ such that $|\langle\phi |\psi\rangle|^2 \geq  \eta$ using $\poly(n,1/\gamma,\log(1/\delta))$ time and copies of~$\ket{\psi}$.
\end{theorem}

\subsection{Sampling points from a small doubling set}
\label{sec:determinesubgroup}
In the previous section, we showed how one can sample from the approximate subgroup $S$ efficiently using $\sample$ (Lemma~\ref{lem:sample_bds}) and introduced the $\bsgtest$ (Theorem~\ref{lemma:analysis_bsg}) that gives us a membership test for a small doubling set lying inside $S$. In particular, we showed that there is a good choice of $u \in S$, which can be obtained efficiently using Bell difference sampling and a set of parameters for which the sets $A^{(1)}(u),A^{(2)}(u)$ (as defined in Theorem~\ref{lemma:analysis_bsg}) are dense and have small doubling constants respectively. For notational convenience, conditioned on such a \emph{good} $u$, we will simply refer to $A^{(1)}(u)$ as $A^{(1)}$ and $A^{(2)}(u)$ as $A^{(2)}$. In this section, we give a procedure that, given copies of $\ket{\psi}$ allows to efficiently sample from the small doubling set $A^{(1)}$.

\begin{lemma}
\label{lem:sampling_small_doubling_set}
Let $\gamma,\delta \in (0,1), t\geq 1$. Suppose $\ket{\psi}$ is an unknown $n$-qubit state with $\Exp_{x \sim q_\Psi}[|\la \psi | W_x | \psi \ra|^2] \geq \gamma$. Then, there is a quantum algorithm that with probability $1-\delta$, produces at least $t$ elements from $A^{(1)}$ using $O(t\cdot \poly(1/\gamma)\log(1/\delta))$ copies of $\ket{\psi}$ and runs in $O(nt\poly(1/\gamma)\log(1/\delta))$~time.
\end{lemma}
\begin{proof}
The procedure that we invoke in order to prove the lemma is fairly straightforward. We run $\sample$ $M$ many times (which we fix later) and then run $\bsgtest$ on these samples for parameters as in Theorem~\ref{lemma:analysis_bsg}.
We summarize this in the algorithm below.
\begin{myalgorithm} \setstretch{1.3}
\begin{algorithm}[H]
    \label{algo:sample_small_doubling_set}
    \caption{\sampleAone($t$,$\gamma,\delta$)} 
    \setlength{\baselineskip}{1.5em} 
    \DontPrintSemicolon 
    \KwInput{Parameter $t \in \mathbb{N}$, copies of $\ket{\psi}$ with $\Exp_{x \sim q_\Psi}[|\la \psi | W_x | \psi \ra|^2] \geq \gamma$, failure prob.~$\delta$}
    \KwOutput{Collection of sets $\calA = \{A_i\}_{i \in [N_s]}$ with $N_s = O(1/\gamma^{350})$ s.t. $\exists A' \in \calA$ with at least $t$ samples from $A^{(1)}$} \vspace{2mm}
    Set $\rho=\gamma^5/20$ and $M = 400t/\poly(\rho)$ \\
    Call $\sample(\gamma, \delta)$ $M$ many times to produce set $\calV$ \\
    Set collection of parameters $P = \{(\zeta^{(i)}_1,\zeta^{(i)}_2,\zeta^{(i)}_3)\}_{i}$  as in Theorem~\ref{lemma:analysis_bsg} \\
    Set ${A}'(u,\vec{\eta}) = \emptyset$ for all $u \in \widetilde{S}$ and $\vec{\zeta} = (\zeta_1,\zeta_2,\zeta_3) \in P$. \\
    Let $\calA = \varnothing$. \\
    \For{$u \in \calV$}{
        \For{$(\zeta_1,\zeta_2,\zeta_3) \in P$}{
            \For{$v \in \calV \setminus \{u\}$}{
                Run $\bsgtest(u, v, \zeta_1,\zeta_2,\zeta_3)$ and let $F_{u,v}$ be the output \\
                If $F_{u,v}=1$, then ${A}'(u,\overline{\eta}) \leftarrow {A}'(u,\overline{\zeta}) \cup \{v\}$ \\
            }
        }
        \lIf{$|A'(u,\overline{\zeta})| \geq t$}{$\calA \leftarrow \calA \cup A'(u,\overline{\eta})$}
    }
    \Return $\calA$
\end{algorithm} 
\end{myalgorithm}
From Lemma~\ref{lem:sample_bds}, we know that calling $\sample$ allows us to produce samples in an approximate group $S$ with density at least $\gamma^5/20$, with probability at least $\Omega(\gamma^7)$. We will set $\rho = \gamma^5/20$. Let us condition first on having a good $S$ (whose probability we will incorporate  at the end). The corresponding conditional distribution $D_\Psi$ is given by Eq.~\eqref{eq:induced_dist_SAMPLE}. From Theorem~\ref{lemma:analysis_bsg}, we are guaranteed with probability at least $\Omega(\gamma^{487})$ over the choice of $u \sim D_{\Psi}$ and parameters $\overline{\zeta} = (\zeta_1,\zeta_2,\zeta_3)$ that we sampled a $u$ for which $A^{(1)}(u)$ is comparatively large and $A^{(2)}(u)$ has small doubling. Let us assume that we would sample a good vertex $u$ and parameters $\zeta_1,\zeta_2,\zeta_3$. We thus use the following procedure and as described in Algorithm~\ref{algo:sample_small_doubling_set}.

Let $M \in \mathbb{N}$ be a parameter to be fixed later. We firstly sample a set of at least $M$ elements, which we denote by $\calV = \{v_i\}_{i \in [M]}$, by calling $\sample$ $O(M/\gamma^2 \log(1/\delta))$ many times. As mentioned above, we have conditioned on sampling from $S$ (which occurs with probability $\Omega(\gamma^7)$), so the set $\calV$ contains elements from the approximate subgroup $S$ with density $\rho$. For each $u \in \calV$, we now run $\bsgtest(u,v,\vec{\zeta})$ by iterating over all $v \in \calV \setminus u$ and all sets of parameters $\vec{\zeta} \in P$ with $\delta' = \delta/(5M)$. We retain only those $v$(s) on which the $\bsgtest$ outputs $1$ into a set $A'(u,\vec{\zeta})$. As mentioned above, we have conditioned on sampling a good $u$ and parameters, so we are guaranteed that there exists a set $A'(u)$ (where we have suppressed the dependence on the choice of parameters $\vec{\zeta}$) containing many elements from $A^{(1)}(u)$.

We now determine $M$ such that with probability at least $1-\delta/5$, the $\bsgtest$ will output $1$ on at least $t$ many elements in $\calV$, for a good $u$ and parameters. First observe by Lemma~\ref{lemma:lb_size_A1} that for a good $u$, we have 
$$
\Pr_{x\sim \Dpsi}[x\in A^{(1)}(u)]\geq \Omega(\gamma^{202}),
$$
Hence, the probability of hitting $A^{(1)}(u)$ is at least $\Omega(\gamma^{202})$ and the expected size of $A'(u)$ can then be lower bounded as $\Exp_{u \sim D_\Psi}[|A'(u)|] \geq \gamma^{202} \cdot M$. Using the Hoeffding bound in Fact~\ref{fact:hoeffding_sampling}, we have\footnote{To see this, consider the indicator random variable $X_i = [v_i \in A^{(1)}(u)]$ for all $v_i \in A'(u)$.}
$$
\Pr[ |A'(u)| \leq t] \leq 2 \exp\left(-2 (\gamma^{202} M-t)^2/M \right),
$$
which can be upper bounded by $\delta/5$ for the choice of $M =O\left(t/\gamma^{202}\cdot  \log (1/\delta) \right))$. 
The value of $M$ was obtained after conditioning on sampling from $S$ which occurs with probability $\geq \Omega(\gamma^7)$ (Lemma~\ref{lem:sample_bds}) and good choice of $u$ along with parameters which occurs with probability $\Omega(\gamma^{487})$ (Theorem~\ref{lemma:analysis_bsg}). Hence, setting $M = \poly(1/\gamma)\cdot  t\cdot  \log(1/\delta)$ proves the claim statement.
\end{proof}

We have so far shown that Lemma~\ref{lem:sampling_small_doubling_set} gives us a procedure to sample points from the small doubling set $A^{(1)}$ efficiently. We now argue that with enough sampled points, we can obtain a linear subspace that has high intersection with $A^{(1)}$. This is formally described below.

\paragraph{Subspace with high intersection with $A^{(1)}$.} We show the following claim, adapted from \cite[Claim~4.14]{tulsiani2014quadratic}, which states that the span of the $t$ sampled points has high overlap with $A^{(1)}$ for a high enough value of $t$ (assuming all the points lied in $A^{(1)}$). 
\begin{claim}[adapted {\cite[Claim~4.14]{tulsiani2014quadratic}}]\label{claim:span_z_high_intersection_A1}
Let $\delta > 0$, $\rho = \poly(\gamma)$. $\{z_1,\ldots,z_t\}$ be $t$ sampled elements from $A^{(1)}$ (according to the distribution $D$ in Eq.~\eqref{eq:induced_dist_SAMPLE}).  For every $t \geq 4n^2 + \log(10/\delta)$, with probability $\geq 1-\delta/5$, we have that
\begin{enumerate}[(i)]
    \item $|\la z_1,\ldots,z_t \ra \cap A^{(1)}| \geq \gamma^{62}/2 \cdot |A^{(1)}|$.
    \item $\mathrm{dim}(\la z_1,\ldots,z_t \ra) \geq n - \log(12/\rho)$
\end{enumerate}
\end{claim}
\begin{proof}
For the first part, we consider $\langle z_1,\ldots,z_t\rangle$ and bound the probability that its intersection with $A^{(1)}$ is small. To this end, we need to bound
\begin{align}
\label{eq:claim4.14tw}
&\Pr_{z_1,\ldots,z_t\sim \Dpsi}\Big[|\la z_1,\ldots,z_t \ra \cap A^{(1)}| < \gamma^{62}/C' \cdot |A^{(1)}|\Big]\\
&=\sum_{T:|T\cap A^{(1)}|\leq \gamma^{62}/C' |A^{(1)}|}\Pr\Big[z_1,\ldots,z_t\in T| \{z_i\}_{i \in [t]} \in A^{(1)}\Big]\cdot \\
&\quad \qquad \qquad \qquad \Pr\Big[\langle z_1,\ldots,z_t\rangle=T|\{z_i\}_{i \in [t]} \in A^{(1)}, \{z_i\}_{i \in [t]} \in T\Big],
\end{align}
where the sum is over all subspaces $T$ satisfying the condition $|T\cap A^{(1)}|\leq (\gamma^{62}/C')\cdot  |A^{(1)}|$. Now, we bound the first probability term from above as follows: for each $i$,  
we have that
\begin{align}
    \Pr_{z_i\sim \Dpsi}[z_i\in T|z_i\in A^{(1)}]&=\frac{\Pr_\Dpsi[z_i\in T\cap A^{(1)}]}{\Pr_\Dpsi[z_i\in A^{(1)}]}\\
    &\leq \frac{\sum_{z_i\in T\cap A^{(1)}}q_\Psi(z_i)\cdot 2^np_\Psi(z_i)}{\sum_{z \in A^{(1)}} q_\Psi(z) \cdot 2^np_\Psi(z)}\\
    &\leq \frac{|T\cap A^{(1)}|}{|A^{(1)}|}\cdot \frac{2^{-n}}{C \gamma^{202}} \\
    &\leq \frac{|T\cap A^{(1)}|}{|A^{(1)}|}\cdot \frac{2^{-n}}{C' \gamma^{62} /2^n},
\end{align}
for some universal constants $C,C'>1$, 
where in the inequality we used that $p_\Psi(z),q_\Psi(z)\leq 2^{-n}$ for all $z$ in the numerator (by Lemma~\ref{lem:ub_qPsi}) and Lemma~\ref{lemma:lb_size_A1} for the denominator. Now, since we assumed that  $|T\cap A^{(1)}|/|A^{(1)}|\leq \gamma^{62}/2$, we have that $\Pr_{z_i\sim \Dpsi}[z_i\in T|z_i\in A^{(1)}]\leq 1/(2C')$ and so can upper bound the expression in Eq.~\eqref{eq:claim4.14tw} as
$
2^{-t}\cdot O(2^{4n^2}),
$
since there are $O(2^{4n^2})$ subspaces in $\FF_2^{2n}$. So  for every $t\geq 4n^2+\log 10/\delta$, we get our desired bound in item $(i)$. The same analysis as above and the same argument as in~\cite{tulsiani2014quadratic} gives us the second item as well.
\end{proof}

\subsection{Finding a subgroup with high probability mass}
\label{sec:succintgroup}
In the previous section, we showed that one can obtain samples from the small-doubling set $A^{(1)}$ using the $\bsgtest$ and using copies of $\ket{\psi}$ with high probability. We now aim to utilize these samples from $A^{(1)}$ to determine a subgroup $V$ that has large probability mass as was shown in \cite{ad2024tolerant} and has indicated at the beginning of this section (Section~\ref{sec:selfcorrectionprotocol}). Our first observation is that $A^{(2)}$ (and thereby $A^{(1)}$) are small doubling sets and the combinatorial $\PFR$ theorem states that these sets can be covered by a small number of cosets of a Freiman subgroup. For convenience, we restate the combinatorial $\PFR$ theorem below. 

\combpfrtheorem*

As mentioned before in Section~\ref{sec:intro}, we currently do not have an algorithm for determining the subgroup $H$ given some form of access to $A$. We thus assume the algorithmic $\PFR$ conjecture that we restate below for convenience.
\algopfrconjecture* 
We will use the above conjecture to show the following theorem.
\begin{theorem}
\label{thm:analysis_algo_V}
Let  $C > 1$ be a universal constant. Let $\ket{\psi}$ be an $n$-qubit state satisfying $\Exp_{x \sim q_\Psi}[|\la \psi | W_x | \psi \ra|^2] \geq \gamma$. Assume the algorithmic $\PFR$ Conjecture~\ref{conj:algopfrconjecture} is true, there is a procedure that with probability $\geq 1-\delta$, outputs a subgroup $V$ with $n - O(\log(1/\gamma)) \leq \dim(V) \leq n + O(\log(1/\gamma))$ such that
$$
\Exp_{x \in V} [2^n p_\Psi(x)] \geq \poly(\gamma).
$$    
This procedure consumes
$$
\widetilde{O}( \poly(n,1/\gamma, \log(1/\delta)))
$$ copies of $\ket{\psi}$ and 
$$
\widetilde{O}(\poly(n,1/\gamma, \log(1/\delta)))$$
time complexity.
\end{theorem}

Before proving this theorem and giving the corresponding algorithm, we give the following fact from \cite{ad2024tolerant} that comments on the probability mass of a subgroup in relation to the probability mass in one of its cosets.
\begin{fact}[{\cite[Claim~4.13]{ad2024tolerant}}]\label{fact:remove_coset}
    Let $W \subseteq \FF_2^{2n}$ be a subgroup. For any $z'$, we have\footnote{Clearly the case where $z'\in W$ is an equality, but we write it out this way for succinctness.}
   \[
   \Exp_{x \in W}[2^n p_\Psi(x)]\geq
   \begin{cases} 
      \Exp_{x \in W}[2^n p_\Psi(z'+x)]
 & \text{ if }z'\in W \\
       \Exp_{x \in W}[2^n p_\Psi(z' + x)] & \text{ if }z'\notin W.
   \end{cases}
\] 
\end{fact}

\begin{claim}\label{claim:tstar_high_overlap_coset_H}
Consider context of Theorem~\ref{thm:analysis_algo_V}. Assume we obtain an approximate group $S$ satisfying the guarantee of Lemma~\ref{lem:sample_bds} and have a good choice of $u$ and parameters as defined in Theorem~\ref{lemma:analysis_bsg}. Let $H$ be the subgroup obtained by applying Theorem~\ref{thm:marton_conjecture} to $A^{(2)}$. Then, given a set $A'$ containing at least $t^\star$ samples from $A^{(1)}$, obtained via Lemma~\ref{lem:sampling_small_doubling_set}, there exists a coset $a^\star + H$ such that $Q = A' \cap (a^\star + H)$ satisfies
$$
|A^{(1)} \cap Q| \geq 4n^2 + \log(10/\delta), \text{ and } |A^{(1)} \cap \mathrm{span}(Q)| \geq (\gamma^{62}/2) \cdot |A^{(1)}|,
$$
as long as
$$
t^\star \geq \big( 4n^2 + \log(10/\delta)\big)\cdot \poly(1/\gamma).
$$
\end{claim}
\begin{proof}
By applying Theorem~\ref{thm:marton_conjecture} to $A^{(2)}$, which is a small doubling set with $K = O(1/\gamma^{932})$ (by Theorem~\ref{lemma:analysis_bsg}), we obtain that there exists a subgroup $H$ satisfying
\begin{equation}\label{eq:pfr_on_A2}
|H|\leq |A^{(2)}| \text{ and } A^{(1)} \subseteq A^{(2)} \subseteq \bigcup \limits_{i \in [2 K^9]} (a_i + H),    
\end{equation}
where we have denoted the different cosets of $H$ as $\{a_i+H\}$ with $\{a_i\}$ being the translates. These cosets of $H$ must also cover $A^{(1)} \cap A'$.\footnote{They also cover $A'$ itself since $A' \subseteq A^{(2)}$ but we do not need to use this fact here.} Thus, there exists a coset $a^\star$, such that 
$$
|A^{(1)} \cap A' \cap (a^\star +H)|\geq (2K^9)^{-1}\cdot |A^{(1)} \cap A'|.
$$
By choosing $t^\star \geq \big( 4n^2 + \log(10/\delta)\big)\cdot (2K^9)$, we then have $|A^{(1)} \cap A'| \geq \big( 4n^2 + \log(10/\delta)\big)\cdot(2K^9)$ and hence ensure that $|A^{(1)} \cap A' \cap (a^\star +H)| \geq 4n^2 + \log(10/\delta)$. Note that this involves $100t^\star \poly(1/\gamma)$ calls to the sampling procedure inside Algorithm~\ref{algo:sample_small_doubling_set} of Lemma~\ref{lem:sampling_small_doubling_set}. Defining $Q = \mathrm{span}(A' \cap (a^\star + H))$, Claim~\ref{claim:span_z_high_intersection_A1} then implies that
$$
|A^{(1)} \cap Q| \geq (\gamma^{62}/2) \cdot |A^{(1)}|,
$$
which gives us the desired result.
\end{proof}

We are now ready to prove the main theorem in this section.
\begin{proof}[Proof of Theorem~\ref{thm:analysis_algo_V}]
We start off by sampling
\begin{equation}\label{eq:value_tstar}
t^\star=\big( 4n^2 + \log(10/\delta)\big)\cdot (2K^9),    
\end{equation}
many points from $A^{(1)}$, which we will denote by $A' = \{z_1,\ldots,z_{t^\star}\}$ and with $K=O(1/\gamma^{932})$ (see Theorem~\ref{lemma:analysis_bsg}. We will argue in the proof later why $t^\star$ is chosen this way. We then run Algorithm~\ref{algo:obtain_pfr_subgroup_V} that takes as input $A'$ and will output the basis of a subspace $V$ that is promised to match the guarantees of Theorem~\ref{thm:analysis_algo_V}.\footnote{We make a remark about Step $(2)$ of the algorithm below. Our algorithmic $\PFR$ conjecture~\ref{conj:algopfrconjecture} requires \emph{uniform} samples from $A^{(2)}$ to output a subspace $H$. Here, we draw samples from $A^{(2)}$  according to $D$ given in Eq.~\eqref{eq:induced_dist_SAMPLE}. Although this is not the uniform distribution, the calculations that we mimicked in the proofs in this section showed that sampling from $D_\Psi$ for dense sets is ``almost" same as close to the uniform distribution upto a $\poly(\gamma)$ loss in parameters. So if we solved Conjecture~\ref{conj:algopfrconjecture} as stated, we could use it as a subroutine for step $(2)$ below.}

\begin{myalgorithm} \setstretch{1.35}
\begin{algorithm}[H]
    \label{algo:obtain_pfr_subgroup_V}
    \caption{\pfrsubgroup($A',\gamma,\delta$)} 
    \setlength{\baselineskip}{1.5em} 
    \DontPrintSemicolon 
    \KwInput{Set of elements $A'$, parameter $\gamma$, sample and  membership oracle access to $A^{(2)}$, failure prob.~$\delta$}
    \KwOutput{Subgroup $V$ defined by its basis}
    \Promise{If input $A'$ contains at least $t^\star = \big( 4n^2 + \log(10/\delta)\big)\cdot (2K^9)$ samples from $A^{(1)}$ then output $V =\mathrm{span}(\{w_i\}_{i \in [t]})$ with $t \geq 4n^2 + \log(10/\delta)$ s.t. $\Exp_{x \in V}[2^n p_\Psi(x)] \geq \Omega(\poly(\gamma))$, $|V| \leq |A^{(2)}|$, and $|V| \geq \poly(\gamma) |A^{(1)}|$} \vspace{2mm}
    Set $t = 4n^2 + \log(10/\delta)$ and $K = O(1/\gamma^{932})$. \\
    Output membership oracle for subgroup $H$, denoted by $\textsf{MEM}_H$, using algorithmic $\PFR$ (Conjecture~\ref{conj:algopfrconjecture}) given sample and  and membership oracle to $A^{(2)}$. \\
    Set $L = \{z_i + z_j : z_i, z_j \in A'\}$ and $L' = \varnothing$. \\
    \For{$w \in L$}{
        If $\textsc{MEM}_H(w)=1$, then $L' \leftarrow L' \cup \{w\}$. \\
    }
    \uIf{$|L'| \geq t$}{\Return $V = \mathrm{span}(L')$}\Else{\Return $\perp$}
\end{algorithm} 
\end{myalgorithm}
In Algorithm~\ref{algo:obtain_pfr_subgroup_V}, we define the set 
$$
L =  \{z_i + z_j : z_i, z_j \in A'\},
$$
and then apply the membership oracle for the subgroup $H$ (obtained via Conjecture~\ref{conj:algopfrconjecture}) to each element in $L$ to create the set $L'$. In other words, the set $L'$ is defined as
$$
L' = \{z_i + z_j : z_i, z_j \in A',\, z_i + z_j \in H\}.
$$
We will now argue that $(i)$ $L'$ is pretty large for the chosen value of $t^\star$ and $(ii)$ $V = \mathrm{span}(L')$ has high probability mass in addition to satisfying the stated size guarantees.

$(i)$ To comment on the size of $L'$, let us first analyze the elements in $L'$. From Theorem~\ref{thm:marton_conjecture}, we know that $A^{(2)}$ is covered by $2K^9$ cosets of $H$ and $A' \subseteq A^{(2)}$ (since $A'$ contains points from $A^{(1)}\subseteq A^{(2)}$). Denoting the cosets of $H$ as $\{a_i + H\}$, we have that for any $z_i, z_j$ belonging to same coset i.e., $z_i,z_j \in a + H$, the sum $z_i + z_j \in L'$. Thus, to show that $L'$ is large, it is enough to show that $A'$ has a high overlap with a coset of $H$. Using Claim~\ref{claim:tstar_high_overlap_coset_H} and the value of $t^\star$ in Eq.~\eqref{eq:value_tstar}, we ensure that there is a coset $a^\star + H$ such that the set $Q^\star = A' \cap (a^\star +H)$ satisfies
\begin{equation}\label{eq:overlap_A1_Qstar}
|A^{(1)} \cap Q^\star|\geq 4n^2 + \log(10/\delta), \text{ and } |A^{(1)} \cap \mathrm{span}(Q^\star)| \geq (\gamma^{62}/2) \cdot |A^{(1)}|
\end{equation}
Denoting the elements in $Q^\star = \{w_i\}$, we also define $\overline{Q} = w_1 + Q^\star = \{w_1 + w_i : w_i \in Q^\star\}$ where $w_1$ is a fixed point in $Q^\star$. Using our earlier observation regarding elements $z_i,z_j$ in the same coset of $H$ would $z_i + z_j \in H$, we then have that $\overline{Q} \subseteq L'$. We can thus lower bound the size of $L'$ as
\begin{equation}\label{eq:size_Lprime}
|L'| \geq |\overline{Q}| \geq |Q^\star| \geq |A^{(1)} \cap Q^\star|\geq 4n^2 + \log(10/\delta),
\end{equation}
where we have used the definition of $Q^\star$ and Eq.~\eqref{eq:overlap_A1_Qstar} in the last inequality. This proves claim $(i)$.

$(ii)$ We note that $V = \mathrm{span}(L')$ satisfies $|V| \leq |H| \leq |A^{(2)}|$ as $L'$ by construction only contains elements in $H$ and the upper bound on $|H|$ follows from Theorem~\ref{thm:marton_conjecture}. Thus,
\begin{equation}\label{eq:ub_dim_V}
    \dim(V) \leq n + O(\log(1/\gamma)).
\end{equation}
To comment on the lower bound of $|V|$, we note that $\overline{Q} \subseteq L'$ and thus 
\begin{align}
    \label{eq:QbarinsideV}
    \mathrm{span}(\overline{Q}) \subseteq \mathrm{span}(L')= V
\end{align}
To lower bound $|V|$, it is then enough to comment on a lower bound on $|\mathrm{span}(\overline{Q})|$. As $\overline{Q} = w_1 + Q^\star$ for a fixed point $w_1 \in Q^\star$, we also have $Q^\star = w_1 + \overline{Q}$ which implies
\begin{equation}\label{eq:size_span_Qbar}
    \mathrm{span}(Q^\star) = \mathrm{span}(w_1 + \overline{Q}) = \mathrm{span}(\overline{Q}) \cup (w_1 + \mathrm{span}(\overline{Q})) \implies |\mathrm{span}(\overline{Q})| \geq |\mathrm{span}(Q^\star)|/2,
\end{equation}
where in the last implication we used $|w_1 + \mathrm{span}(\overline{Q})| = |\mathrm{span}(\overline{Q})|$. Using Eq.~\eqref{eq:size_span_Qbar} and Eq.~\eqref{eq:overlap_A1_Qstar}, we thus have
\begin{equation}\label{eq:size_V}
|V| \geq |\mathrm{span}(\overline{Q})| \geq |\mathrm{span}(Q^\star)|/2 \geq |A^{(1)} \cap \mathrm{span}(Q^\star)|/2 \geq (\gamma^{62}/4) \cdot |A^{(1)}|.
\end{equation} 
Since  $|A^{(1)}|\geq 2^n\cdot \poly(\gamma)$, we have
\begin{equation}\label{eq:lb_dim_V}
    \dim(V) \geq n - O(\log(1/\gamma)).
\end{equation}
It remains to show that the probability mass in $V$ is high. Towards that, we first note
\begin{equation}\label{eq:prob_mass_span_Qstar}
    \Exp_{x \in \mathrm{span}(Q^\star)}[2^n p_\Psi(x)] \geq \frac{1}{|\mathrm{span}(Q^\star)|} \sum_{x \in A^{(1)} \cap \mathrm{span}(Q^\star)} 2^n p_\Psi(x) \geq \frac{1}{|A^{(2)}|} \cdot \frac{\gamma^{62} |A^{(1)}|}{2} \cdot \frac{\gamma}{8} \geq \poly(\gamma),
\end{equation}
where we have used that $|\mathrm{span}(Q^\star)| \leq |a^\star + H| \leq |H| \leq |A^{(2)}|$ in the first inequality. In the final inequality, we used that that all the elements in $A^{(1)}$ have expectation $\geq \gamma/8$, Eq.~\eqref{eq:overlap_A1_Qstar} to comment on the size of $|A^{(1)} \cap \mathrm{span}(Q^\star)|$ and $|A^{(1)}|/|A^{(2)}| \geq \poly(\gamma)$ by Theorem~\ref{lemma:analysis_bsg}.
Using our earlier observation that $\mathrm{span}(Q^\star) = \mathrm{span}(\overline{Q}) \cup (w_1 + \mathrm{span}(\overline{Q}))$, we also have
\begin{align}
\poly(\gamma) \leq \Exp_{x \in \mathrm{span}(Q^\star)}[2^n p_\Psi(x)] 
&= \frac{1}{|\mathrm{span}(Q^\star)|}\left(\sum_{x \in \mathrm{span}(\overline{Q})}[2^n p_\Psi(x)] + \sum_{x \in w_1 + \mathrm{span}(\overline{Q})}[2^n p_\Psi(x)]\right) \\
&= \frac{|\mathrm{span}(\overline{Q})|}{|\mathrm{span}(Q^\star)|}\left(\Exp_{x \in \mathrm{span}(\overline{Q})}[2^n p_\Psi(x)] + \Exp_{x \in w_1 + \mathrm{span}(\overline{Q})}[2^n p_\Psi(x)]\right) \\
&\leq \frac{2|\mathrm{span}(\overline{Q})|}{|\mathrm{span}(Q^\star)|} \cdot \Exp_{x \in \mathrm{span}(\overline{Q})}[2^n p_\Psi(x)],
\label{eq:prob_mass_Qbar}
\end{align}
where we used $\Exp_{x \in \mathrm{span}(\overline{Q})}[2^n p_\Psi(x)] \geq \Exp_{x \in w_1 + \mathrm{span}(\overline{Q})}[2^n p_\Psi(x)]$ by Fact~\ref{fact:remove_coset} in the last inequality. This implies that
\begin{equation}
\Exp_{x \in V}[2^n p_\Psi(x)] \geq \frac{1}{|V|} \sum_{x \in \mathrm{span}(\overline{Q})} 2^n p_\Psi(x) \geq \frac{|\mathrm{span}(Q^\star)|}{2 |V|} \cdot \poly(\gamma) \geq \frac{\gamma^{62}}{4} \cdot \frac{|A^{(1)}|}{|A^{(2)}|} \cdot \poly(\gamma) \geq \poly(\gamma),
\end{equation}
where the first inequality used Eq.~\eqref{eq:QbarinsideV}, Eq,~\eqref{eq:prob_mass_Qbar} in the second inequality and $|V| \leq |H| \leq |A^{(2)}|$ along with $|\mathrm{span}(Q^\star)| \geq \gamma^{62}/2 |A^{(1)}|$ ( from Eq.~\eqref{eq:overlap_A1_Qstar}) in the third inequality and $|A^{(1)}|/|A^{(2)}| \geq \poly(\gamma)$ by Theorem~\ref{lemma:analysis_bsg} in the final inequality. This proves the desired properties of $V$.

The runtime of the algorithm is entirely dominated by $O(n^2 \poly(1/\gamma))$ calls to the $\bsgtest$ to ensure collection of $A'$ containing at least $t^\star$ elements from $A^{(1)}$, which in turn calls the $\sample$ subroutine $\poly(1/\gamma)$ many times, each taking time $O(n \poly(1/\gamma))$ time and $O(1/\gamma)$ samples (Lemma~\ref{lem:sample_bds}). The overall time complexity is thus $\tilde{O}(n^4 \poly(1/\gamma))$ with a corresponding sample complexity of $\tilde{O}(n^3 \poly(1/\gamma))$. Since the complexity of algorithmic $\PFR$ conjecture is $\poly(n,1/\gamma)$, the sample and time complexity in the theorem statement follows.
\end{proof}

\subsection{Finding the stabilizer state}
\label{sec:findstabilizer}
We have so far determined a subgroup $V$  that has high probability mass, i.e., $\Exp_{x \in V}[2^n p_\Psi(x)] \geq \poly(\gamma)$. We now show that given such a subgroup $V$, we can determine a stabilizer state that has high fidelity with the state $\ket{\psi}$. Formally, we have the following result.
\begin{theorem}\label{thm:analysis_algo_stab_state}
Let $\gamma,\delta \in (0,1]$, $\ket{\psi}$ be an $n$-qubit state with $\Exp_{x\sim q_\Psi}\left[ |\la \psi | W_x | \psi \ra|^2 \right]\geq \gamma$ and $V~\subseteq~\FF_2^{2n}$ be a subgroup such that $\Exp_{x \in V}[2^n p_\Psi(x)] \geq \Omega(\gamma^C)$ (for a universal constant $C > 1$). Algorithm~\ref{algo:find_good_stab} outputs a stabilizer state $\ket{\phi}$  such that $|\la \phi| \psi\ra|^2 \geq \Omega(\gamma^C)$ with probability $\geq 1-\delta$, consuming 
$$
\widetilde{O}(\poly(1/\gamma) \log 1/\delta)
$$ copies of $\ket{\psi}$ and 
$$
\widetilde{O}\Big(n^3 + n^2 \cdot \poly(1/\gamma)  \log 1/\delta \Big)
$$ total number of gates.
\end{theorem}
To prove the above theorem, we require the following theorem from \cite{ad2024tolerant} that states a subgroup $V$ having high probability mass must necessarily have a low \emph{stabilizer covering} (whose definition will be clear in the theorem statement below).
\begin{theorem}[{\cite[Theorem~4.14]{ad2024tolerant}}]\label{thm:stabilizer_covering_group}
Let $\ket{\psi}$ be an $n$-qubit state such that there exists a subgroup $V \subseteq \FF_2^{2n}$ with $\Exp_{x \in V}[2^n p_\Psi(x)] \geq \Omega(\gamma^C)$ (for a universal constant $C > 1$). Then, 
\begin{enumerate}[(i)]
    \item there exists a Clifford unitary $U$ and $m+k\leq n$ with $k\leq \log(\gamma^{-C})$ such that $$UVU^\dagger = \la Z_1,X_1,\ldots,Z_k,X_k,Z_{k+1},\ldots,Z_{k+m}\ra$$
    \item $V$ can be covered by a union of $2^{k}+1 = O(\gamma^{-C})$ many stabilizer subgroups $G_j\subseteq \01^{2n}$. Particularly, $\{G_j\}$ correspond to mutually unbiased bases of Fact~\ref{fact:mub_unsigned_stab}.
\end{enumerate}
\end{theorem}
\begin{remark}
The stabilizer subgroups in the above theorem could have been alternately constructed using Paulis but this leads to a worse polynomial dependence (see~\cite{ad2024tolerant} for details). \footnote{The assumption regarding existence of a subgroup $V$ with high probability mass is implied by $\Exp \limits_{x \sim q_\Psi}\left[|\la \psi | W_x | \psi \ra|^2 \right] \geq \gamma$ or $\gowers{\ket{\psi}}{3}^8 \geq \gamma$ for different universal constants $C$. See~\cite{ad2024tolerant} for details.}
\end{remark}

Moreover, we will use following characterization of the stabilizer state $\ket{\phi}$ that is promised to have high fidelity with the unknown state $\ket{\psi}$ and which states that the unsigned stabilizer group of $\ket{\phi}$ corresponds to one of the subgroups in the $\poly(1/\gamma)$-sized stabilizer covering of $W$ (Theorem~\ref{thm:analysis_algo_stab_state})\footnote{This can also be observed as a consequence of applying Lemma~\ref{lem:state_isotropic_space_high_mass} to the stabilizer subgroup in the stabilizer covering of $V$ (Theorem~\ref{thm:analysis_algo_stab_state}) that maximizes the probability mass.}. This was implicit in the proof of the inverse Gowers-$3$ norm of quantum states~\cite[Proof of Theorem~4.2]{ad2024tolerant} which we state here explicitly.
\begin{theorem}[{\cite[Proof of Theorem~4.2]{ad2024tolerant}}]\label{thm:characterization_good_stab}
Let $\ket{\psi}$ be an $n$-qubit state with $\Exp_{x\sim q_\Psi}\left[|\la\psi | W_x | \psi \ra|^2 \right]\geq \gamma$ and  $V \subset \FF_2^{2n}$ be the subgroup in Theorem~\ref{thm:stabilizer_covering_group}. Then, there exists an $n$-qubit stabilizer state $\ket{\phi}$ such that $G \subset \weyl{\ket{\phi}}$ where $G$ is an unsigned stabilizer subgroup in the covering of $V$ from Theorem~\ref{thm:stabilizer_covering_group} and $|\la \phi | \psi \ra|^2 \geq \Omega(\gamma^C)$ (for a universal constant $C>1$). Additionally, considering the Clifford unitary $U$ from Theorem~\ref{thm:stabilizer_covering_group}, the stabilizer state $U^\dagger \ket{\phi}$ can be expressed as
$$
U^\dagger \ket{\phi} = \ket{\varphi_z} \otimes \ket{z},
$$
where $\ket{z}$ is an $n-k$ basis state, $\ket{\varphi_z}$ is a $k$-qubit stabilizer state corresponding to~$U G U^\dagger$.
\end{theorem}

We now give the proof of Theorem~\ref{thm:analysis_algo_stab_state} which gives a guarantee on the correctness of Algorithm~\ref{algo:find_good_stab} using subroutines that algorithmize Theorem~\ref{thm:stabilizer_covering_group} and Theorem~\ref{thm:characterization_good_stab} as we will see shortly.

\begin{proof}[Proof of Theorem~\ref{thm:analysis_algo_stab_state}]Before giving the algorithm that proves the theorem, we first give an outline of the algorithm. Given the subgroup $V$ in terms of its basis, we first use Claim~\ref{claim:clifford_synthesis_subgroup}, which in turn uses $\symgramschmidt$, to determine the Clifford unitary $U$~(Fact~\ref{fact:clifford_action_on_group}) that transforms $V$ to the form of $\calP^k \otimes \calP^{m}_{\calZ}$,~i.e.,
$$
U V U^\dagger = \la Z_1, X_1, \ldots, Z_k, X_k, Z_{k+1}, Z_{k+2}, \ldots, Z_{k+m} \ra = \calP^k \times \calP_\calZ^m,
$$
where the values of $k$ and $m \leq n-k$ are determined as part of Claim~\ref{claim:clifford_synthesis_subgroup}. This consumes $O(n^3)$ time and $U$ has a gate complexity of $O(n^2)$.

Furthermore, Theorem~\ref{thm:stabilizer_covering_group} guarantees that $k \leq \log(\gamma^{-C})$. The goal now is to construct a \emph{good} stabilizer state $\ket{\phi}$ corresponding to the unsigned stabilizer subgroups in $UVU^\dagger$ which has high fidelity with $U \ket{\psi}$ and then one can output $U^\dagger \ket{\phi}$ (which would then have high fidelity with $\ket{\psi}$) by noting that $U$ was a Clifford circuit. 

To find this good stabilizer, one (slightly wrong) approach could have been: 
\begin{enumerate}[$(i)$]
    \item Construct an unsigned stabilizer covering of $W=UVU^\dagger$, denoted by $\calS$, by considering the $2^k+1$ many unsigned $k$-qubit stabilizer subgroups $\{A_i\}$ covering $\calP^k$ using Fact~\ref{fact:mub_unsigned_stab} and then setting $\calS = \cup_{i \in [M]} (A_i \times \calP_Z^m)$. Let $M=2^k +1 \leq O(\gamma^{-C})$.
    \item Since $m+k\leq n$, extend each stabilizer subgroup $A_i \times \calP_Z^m$ for all $i \in [M]$ to an $n$-qubit stabilizer group $\widetilde{A}_i$ arbitrarily,
    \item Consider all the stabilizer states with an unsigned stabilizer group of $\widetilde{A}_i$ for all $i \in [M]$,
    \item Determine the fidelity of all the stabilizer states so obtained with $U \ket{\psi}$ output the stabilizer state with maximal fidelity as $U^\dagger \ket{\phi}$. 
\end{enumerate}
The issue is that in step $(iii)$ we would obtain $O(2^n)$ signed stabilizer groups by arbitrarily assigning signs of $\{\pm 1\}$ to the generators in the unsigned stabilizer group $\widetilde{A}_i$ and would thus have $O(2^n)$ stabilizer states corresponding to each $\widetilde{A}_i$ for each $i \in [M]$. To overcome this issue, we first note that we can extend the stabilizer groups in the stabilizer covering  by completing $W = UVU^\dagger$ to the full $n$-qubit Pauli group obtain $W'$:
$$
W' = \la Z_1, X_1, \ldots, Z_k, X_k, Z_{k+1}, Z_{k+2}, \ldots, Z_{k+m}, \ldots, Z_n \ra = \calP^k \times \calP_\calZ^{n-k}.
$$
We will now call the unsigned stabilizer covering of $W'$ as $\widetilde{\calS} = \cup_{i \in [M]} \widetilde{A}_i$ (where $M = O(\gamma^{-C})$). By definition observe that $I^{\otimes k} \otimes \{I,Z\}^{\otimes (n-k)} \subseteq \widetilde{A}_i$ for all $i \in [M]$, we have that any stabilizer state $\ket{\phi}$ with an unsigned stabilizer group of $\widetilde{A}_i$ can be expressed as a product stabilizer state as follows
\begin{equation}\label{eq:prod_stab_state}
    \ket{\phi} = \ket{\varphi_x} \otimes \ket{x},
\end{equation}
where $\ket{\varphi}$ is a $k$-qubit stabilizer state and $\ket{x}$ is an $(n-k)$-qubit computational basis state.  The natural question is, how to find this  $x$? For that, we observe that from Theorem~\ref{thm:characterization_good_stab}, we know that there exists an $n$-qubit stabilizer state of the from $\ket{\varphi_{x^*}} \otimes \ket{{x^*}}$  such that $|\la \psi| U^\dagger (\ket{\varphi_{x^*}} \otimes \ket{x^*})|^2 \geq \poly(\gamma)$. In other words,
\begin{align}
\label{eq:stabandz}
\max_{x\in \{0,1\}^{n-k}} |\la \psi|U^\dagger (\ket{\varphi_x} \otimes \ket{x})|^2\geq \poly(\gamma).
\end{align}
The above also implies that
\begin{equation}\label{eq:opt_product_stab_state}
\ket{\varphi_{x^*}} \otimes \ket{{x^*}} = \argmax_{x\in \{0,1\}^{n-k}} |\la \psi|U^\dagger (\ket{\varphi_x} \otimes \ket{x})|^2.
\end{equation}
To determine this stabilizer state $\ket{\varphi_{x^*}} \otimes \ket{{x^*}}$ which has high fidelity with $U \ket{\psi}$, it is then enough to consider all those states of the from in Eq.~\eqref{eq:prod_stab_state} and achieves Eq.~\eqref{eq:stabandz}. Now we find this $x^*$ by  measuring and sampling from the $U\ket{\psi}$ state.   With this we are ready to state the final~algorithm.
\begin{myalgorithm}
\setstretch{1.35}
\begin{algorithm}[H]
    \label{algo:find_good_stab}
    \setlength{\baselineskip}{1.5em} 
    \DontPrintSemicolon 
    \caption{\findstabilizer($V,\gamma,\delta$)} 
    \KwInput{Access to copies of $\ket{\psi}$ with $\Exp_{x\sim q_\Psi}\left[ |\la \psi | W_x | \psi \ra|^2 \right]\geq \gamma$, subgroup $V$ such that $\Exp_{x \in V} [2^n p_\Psi(x)] \geq \poly(\gamma)$ (as in Theorem~\ref{thm:analysis_algo_stab_state}), failure probability $\delta$.} 
    \KwOutput{Output stabilizer state $\ket{\phi}$ such that $|\la \phi | \psi \ra|^2 \geq \Omega(\gamma^C)$}
    \vspace{2mm}
    Run $\symgramschmidt$, that on input the basis for $V$, outputs its centralizer $C_V$ (basis for $V\cap V^\perp$ and anti-commutant $A_V$ (basis for $V \backslash \la C_V \ra$). \\ 
    Obtain Clifford circuit $U$ and integers $m,k\geq 1$ such that $UVU^\dagger = \calP^k \times \calP_{\calZ}^m$ using Claim~\ref{claim:clifford_synthesis_subgroup} on inputs of $C_V$ and $A_V$ 
   \label{algostep:find_clifford}\\
    Set $W \leftarrow UVU^\dagger$ and extend $W$ to $\widetilde{W} = \calP^k \times \calP_{\calZ}^{n-k}$ \\
    Let $Q$ be the set of unsigned stabilizer groups corresponding to MUBs on $k$-qubits (Fact~\ref{fact:mub_unsigned_stab}) (note $|Q| = 2^k + 1$) \label{algostep:mub}\\
    Let $\Phi=\{\ket{\varphi_1},\ldots,\ket{\varphi_K}\}$ with $N_s=2^k\cdot (2^k+1)$ be the set of the stabilizer states with unsigned stabilizer groups in $Q$ along with all possible sign assignments of $\{\pm 1\}$. \label{algostep:signed_stab_groups} \\
    Let $N_r=\poly(1/\gamma)$ and initialize empty list $\calL \leftarrow \emptyset$.\\
    \For{$i=1:N_s$}{ \label{algostep:sample_good_choices_z}
        \For{$j=1:N_r$}{
            Measure the first $k$ qubits of $U \ket{\psi}$ in the $\{\ketbra{\varphi_i}{\varphi_i},\id-\ketbra{\varphi_i}{\varphi_i}$\} basis. Let measurement outcome be $m$. \label{algostep:measure_kqubits}\\
            \If{$m=0$}{
                Measure the remaining $(n-k)$ qubits in the computational basis. Let measurement outcome be $z$. \\
                Update $\calL \leftarrow \calL \cup \{\ket{\varphi_i} \otimes \ket{z} \}$ \label{algostep:mark_collisions}\\
            }
        }
    }
    Compute fidelity of all stabilizers in $\calL$ with $U \ket{\psi}$ using classical shadows in Lemma~\ref{lem:classicalshadow}. \label{algostep:measure_last_qubits} \\
    Obtain the stabilizer state $\ket{\varphi} \otimes \ket{x}$ in $\calL$ that has maximized stabilizer fidelity. \label{algostep:shadows} \\
    \Return stabilizer state $U^\dagger(\ket{\varphi} \otimes \ket{x})$ \\
\end{algorithm}
\end{myalgorithm}
Steps $1-2$ are as we defined before the algorithm. We now discuss  Lines~\ref{algostep:mub}-\ref{algostep:shadows}, Algorithm~\ref{algo:find_good_stab}.
Recall that the goal is to obtain a stabilizer product state $\ket{\varphi_z} \otimes \ket{z}$ that achieves the guarantee of Eq.~\ref{eq:stabandz}, we construct a set $\Phi$ of all possible choices of $k$-qubit stabilizer states $\ket{\varphi_x}$ and then sample \emph{good} choices of $\ket{z}$ for $z \in \{0,1\}^{n-k}$. To construct $\Phi$, we proceed as follows. We consider the set of $2^k+1$ unsigned stabilizer groups corresponding to the $k$-qubit mutually unbiased bases (MUB) of Fact~\ref{fact:mub_unsigned_stab}, which we denote by $Q$. $\Phi$ is then set to contain stabilizer states with signed stabilizer groups corresponding to each unsigned stabilizer group in $Q$ with all possible sign assignments $\{\pm 1\}$. The resulting $\Phi$ will be of size $N_s =2^k (2^k + 1)$. This is executed in Lines~\ref{algostep:mub}-\ref{algostep:signed_stab_groups} of Algorithm~\ref{algo:find_good_stab}.

We now discuss how to obtain good $z\in \01^{n-k}$ This is carried out in Lines~\ref{algostep:sample_good_choices_z}-\ref{algostep:mark_collisions} of Algorithm~\ref{algo:find_good_stab}. We first observe that for a given $i\in [N_s]$, the probability that the probability of obtaining the measurement outcome $m=0$ upon measuring the first $k$ qubits of $U \ket{\psi}$ in the basis $\{\ketbra{\varphi_i}{\varphi_i},I-\ketbra{\varphi_i}{\varphi_i}\}$ in Line~\ref{algostep:measure_kqubits} is given by
\begin{align}
    \Pr[m=0] = \la \psi| U^\dagger (|\varphi_i\ra \la \varphi_i|\otimes \id^{n-k} ) U |\psi\ra &=  \sum_{x\in \FF_2^{n-k}}  \la \psi|U^\dagger ( |\varphi_i\rangle \la \varphi_i|\otimes \ketbra{x}{x} ) U |\psi\ra \\
    & \geq \max_{x\in \FF_2^{n-k}}|\la \psi| U^\dagger (\ket{\varphi_i} \otimes \ket{x} ) |^2
\end{align}
which for some $i \in [N_s]$ will be $\geq \poly(\gamma)$ due to Eq.~\eqref{eq:stabandz} and by construction, $\ket{\varphi_{y}} \in \Phi$. Conditioned on obtaining $m=0$, we then measure the last $(n-k)$ qubits of $U\ket{\psi}$ in the computational basis (Line~\ref{algostep:measure_last_qubits}). Note that for some $i^\star \in [N_s]$, $\ket{\varphi_{i^\star}}$ coincides with $\ket{\varphi_y}$ of Eq.~\eqref{eq:opt_product_stab_state}. For this $i^\star$ and conditioned on $m=0$, we will sample a $z$ (satisfying Eq.~\eqref{eq:opt_product_stab_state}) with probability 
$$
\langle  \psi|U^\dagger \big(\ketbra{\varphi_i}{\varphi_i}\otimes \ketbra{y}{y}\big) U|\psi\rangle = |\la \psi|U^\dagger (\ket{\varphi_y} \otimes \ket{y})|^2 \geq \poly(\gamma).
$$
We thus have that the following: Line~\ref{algostep:measure_kqubits} succeeds with probability $\geq \poly(\gamma)$ for some $i^\star \in [K]$, and would allow us to move to Line~\ref{algostep:measure_last_qubits} which would produce the desired $z\in \FF_2^{n-k}$ with probability $\geq \poly(\gamma)$. Setting $N_r=\poly(1/\gamma)$ ensures both these steps succeed and that $\ket{\varphi_z} \otimes \ket{z} \in \calL$. We now determine the desired state $\ket{\varphi_z} \otimes \ket{z}$ by estimating fidelity of all the states in the list $\calL$ with $U\ket{\psi}$ up to error $O(\gamma^C)$ (with $C > 1$ being the constant of the theorem), using Lemma~\ref{lem:classicalshadow} and outputting the one with maximal fidelity. As $\calL \leq O(4^k N_r) \leq \poly(1/\gamma)$, this consumes $\widetilde{O}\left(\poly(1/\gamma) \log 1/\delta \right)$ sample complexity and $\widetilde{O}\left( n^2 \cdot \poly(1/\gamma) \log(1/\delta)\right)$ time complexity. The stabilizer state $\ket{\phi}$ that matches the guarantee of the theorem is then $U(\ket{\varphi_y} \otimes \ket{y})$. We compute $U(\ket{\varphi_y} \otimes \ket{y})$ classically using \cite{aaranson2004improved} which takes as input the $O(n)$ generators of the stabilizer group of $\ket{\varphi_y} \otimes \ket{y}$ (which we have learned as part of $\findstabilizer$) and classical description of $U$ (which we also have from step $2$ of $\findstabilizer$). Noting that $U$ has $O(n^2)$ gate complexity, this consumes $O(n^2)$ time~\cite{aaranson2004improved}. The overall time complexity is
$$
\widetilde{O}\left(n^3 + n^2 \cdot \poly(1/\gamma)\cdot  \log 1/\delta\right)
$$
with main contributions due to cost of obtaining $U$ and estimation of stabilizer fidelities.
\end{proof}

\subsection{Putting everything together}
\label{sec:puttingtogethereverything}
We now prove the main Theorem~\ref{thm:restatementofselfcorr} of this section, that will then imply our first result. 
\begin{proof}[Proof of Theorem~\ref{thm:restatementofselfcorr}]
Let $t^\star = (4n^2 + \log(10/\delta)) \cdot (2K^9)$ with $K=O(1/\gamma^{932})$.
\begin{enumerate}
    \item Given copies of $\ket{\psi}$ we call the $\sample$ subroutine $O(t^\star/\gamma^2 \log(1/\delta))$ many times to produce a set $\calV$ containing elements from the approximate subgroup $S$ with density at least $\rho = \gamma^5/20$ and size $|S| \geq \Omega(\gamma^2) 2^n$, with probability at least $\Omega(\gamma^7)$. Let $D_\Psi$ be the conditional distribution corresponding to Bell difference sampling conditioned on landing in $S$ (Eq.~\eqref{eq:induced_dist_SAMPLE}).
    \item We then use $\sampleAone$ which calls the $\bsgtest$ on vertices in $\calV$ using the parameters of Theorem~\ref{lemma:analysis_bsg}. By Theorem~\ref{lemma:analysis_bsg}, we are guaranteed that with probability at least $\Omega(\gamma^{487})$ over $u \sim D_\Psi$ and parameter choices (Section~\ref{sec:algo_BSG}) that exists subsets $A^{(1)}(u) \subseteq A^{(2)}(u)$ (as defined in Theorem~\ref{lemma:analysis_bsg}) such that
    $$
    |A^{(1)}(u)| \geq \Omega(\gamma^{138}) \cdot |S|, \quad |A^{(2)}(u) + A^{(2)}(u)| \leq O(1/\gamma^{932}) \cdot |S|.
    $$
    Using $\sampleAone$ of Lemma~\ref{lem:sampling_small_doubling_set}, we sample a set $A'$ containing elements from the small doubling set $A^{(2)}(u)$ corresponding to the good $u$ and with at least $t^\star$ many elements from $A^{(1)}(u)$.
    \item We then determines a basis for a subgroup $V$ for which $\Exp_{x \in V}[2^n p_\Psi(x)] \geq \poly(\gamma)$ and $|V| \leq |A^{(2)}(u)|$ by using $\pfrsubgroup$ (which requires the algorithmic $\PFR$ conjecture~\ref{conj:algopfrconjecture}). It takes as input the set $A'$ containing at least $t^\star$ many samples from $A^{(1)}(u)$.
    \item After having determined a basis for the subgroup $V$, we use $\findstabilizer$ to find a stabilizer state $\ket{\phi}$ corresponding to $V$ such that $|\langle \psi|\phi\rangle|^2\geq \poly(\gamma)$.
\end{enumerate}
Thus, with probability $O(\gamma^{494})$, we find a stabilizer state $\ket{\phi}$ with the desired fidelity with $\ket{\psi}$. 

In $\sampleAone$, the $\bsgtest$ calls the $\sample$ subroutine $\poly(1/\gamma)$ many times, each taking time $O(n \poly(1/\gamma))$ time and $O(1/\gamma)$ samples (Lemma~\ref{lem:sample_bds}). The time complexity of this step is thus $\tilde{O}( \poly(n,1/\gamma))$ with a corresponding sample complexity of $\tilde{O}(n^3 \poly(1/\gamma))$. Once we have determined $V$, we then utilize Theorem~\ref{thm:analysis_algo_stab_state} to determine the stabilizer state that has $\geq \poly(\gamma)$ fidelity with $\ket{\psi}$. This consumes an additional $\widetilde{O}(\poly(1/\gamma) \log 1/\delta)$ copies of $\ket{\psi}$ and $\widetilde{O}(n^3 + n^2 \cdot \poly(1/\gamma)  \log 1/\delta)$ time.  Since the complexity of the algorithmic $\PFR$ conjecture is $\poly(n,1/\gamma)$, the overall self-correction procedure has complexity that scales polynomial in $n,1/\gamma,\log 1/\delta$.
\end{proof}

Finally it is immediate to see the analogous main theorem in terms of Gowers norm of $\ket{\psi}$. 
\begin{theorem}
Let $\varepsilon>0$ and $\eta=\varepsilon^{C}$ for a universal constant $C>1$. Let~$\ket{\psi}$ be an $n$-qubit quantum state such that $\gowers{\ket{\psi}}{3}^8 \geq \varepsilon$.  There is a protocol that with probability $1-\delta$, outputs a $\ket{\phi}\in \Sh$ such that $|\langle\phi |\psi\rangle|^2 \geq  \eta$ using $\poly(n,1/\varepsilon,\log(1/\delta))$ time and copies of~$\ket{\psi}$.
\end{theorem}

The proof above follows immediately by the application of Fact~\ref{fact:relation_expectation_paulis_qPsi_and_pPsi}. The proof of Theorem~\ref{thm:self_correction} follows from the observation of Fact~\ref{fact:lower_bound_stabilizer_fidelity} that $\calF_{\calS}(\ket{\psi}) \geq \uptau \implies \Exp_{x \sim q_\Psi}[2^n p_\Psi(x)] \geq \uptau^6$ so $\gamma=\uptau^6$ is what we instantiate with.

\section{Improper self-correction of stabilizer states}
\label{sec:improperselfcorrection}
In this section, we consider the task of \emph{improper} self-correction, i.e., can one do better if we allow the learning algorithm to output an arbitrary quantum state (for the task of self-correction) instead of a stabilizer state? We show that one can obtain a polynomially better dependency in the complexity for this weaker task. From a technical perspective, in order to prove our improper self-correction result, along the way we give a new local Gowers-$3$ inverse theorem showing that if $\Exp_{x \sim q_\Psi}[2^n p_\Psi(x)] \geq \gamma$ then $\calF_{\calS(n-O(\log(1/\gamma))}(\ket{\psi}) \geq \poly(\gamma)$. This then allows us to give a \emph{tolerant testing} algorithm for the class of low-stabilizer dimension states, which is new as far as we know. In particular, we show that $\poly(2^{2t}/\varepsilon)$ many copies of an unknown quantum state $\ket{\psi}$ are sufficient to tolerantly test if $\ket{\psi}$ is $\varepsilon$-close to a state with stabilizer dimension $n-t$ or $\varepsilon^C$-far from all such states (for some $C>1$). We then algorithmize the local Gowers-$3$ inverse theorem to give an improper self-correction protocol, which is the main result of this section and is formally stated below.
\begin{theorem}
\label{thm:improper_SC}
Let $\tau>0$ and $\ket{\psi}$ be such that $\stabfidelity{\ket{\psi}} \geq  \tau$. There is a protocol that with probability $ 1-\delta$, outputs a $\ket{\phi}$ whose stabilizer dimension is at least $n-O(\log(1/\tau))$ such that $|\langle\phi |\psi\rangle|^2 \geq  \tau^{C_2}$ using $\poly(n,1/\tau,\log(1/\delta))$ copies of $\ket{\psi}$ and time.
\end{theorem}
In order to prove this theorem, we break down the proof into three steps, in Section~\ref{sec:step1} we give an existential proof, that if the stabilizer fidelity is high, then there \emph{exists} a high stabilizer dimension state that is close enough, in Section~\ref{sec:step2} we give a tolerant testing algorithm for states with high stabilizer dimension and in Section~\ref{sec:learning_algo_improper_SC} we give our final algorithm putting together these~steps. Recall the notation that  $\mathcal{S}(n-t)$ is the class of states with stabilizer dimension $n-t$ (which we use often in this section).
\subsection{Local inverse Gowers-3 theorem of quantum states}
\label{sec:step1}
In this section, we prove Step $(1)$, i.e., $\Exp_{x \sim q_\Psi}[2^n p_\Psi(x)] \geq \gamma$ implies there is a stabilizer dimension $n-\log(1/\gamma)$ state close to $\ket{\psi}$. We state it formally below.
\begin{restatable}{theorem}{inv_thm_stab_dim}\label{thm:inverse_gowers3_stab_dim}
Let $\gamma\in [0,1]$, $C'>1$ be a constant. If $\ket{\psi}$ is an $n$-qubit quantum state such that $\Exp_{x \sim q_\Psi}[2^n p_\Psi(x)] \geq \gamma$, then there is an $n$-qubit state $\ket{\phi}$ of stabilizer dimension at least $n - \poly(1/\gamma)$ such that $|\la \psi | \phi \ra|^2 \geq \gamma^{C'}$.\footnote{We will show this for $C'<C$ where $C$ is as defined in Theorem~\ref{thm:inverse_weyl_exp_states} and in particular for $C'=728$.}
\end{restatable}

To prove Theorem~\ref{thm:inverse_gowers3_stab_dim}, we will require the following result that shows the fidelity of any state with $\calS(n-k)$ is bounded from below by the probability mass over the space obtained by the product of the $k$-qubit Pauli group and a Lagrangian space of dimension $(n-k)$. This can be viewed as an extension of \cite[Lemma~4.6]{grewal2023efficient} who showed that the stabilizer fidelity of any state is bounded by the probability mass over any Lagrangian space below (see Theorem~\ref{fact:lower_bound_stabilizer_fidelity}).

\begin{lemma}\label{lem:lower_bound_stab_dim_fidelity}
Suppose $\ket{\psi}$ is an $n$-qubit state and $\alpha \in (0,1]$. If $W$ is a subgroup of Paulis of the form $W \subseteq \calP^k \times \calP_Z^{n-k}$  such that 
$$
\sum_{x \in W} p_\Psi(x) \geq \alpha
$$
then there exists $\ket{\phi} = \ket{\varphi} \otimes \ket{z} \in \calS(n-k)$ where $\ket{\varphi}$ is a $k$-qubit state and $\ket{z}$ is an $(n-k)$-qubit computational basis state such that
$$
|\la \phi | \psi \ra|^2 \geq \sum_{x \in W} p_\Psi(x) \geq \alpha.
$$
\end{lemma}

We will require the following structural result to prove the above lemma.
\begin{prop}\label{prop:prob_mass_Pk_PZm}
Let $t\geq 1$. For any $n$-qubit pure quantum state $\ket{\psi}$, the following is true
$$
\sum_{W_y \in \calP^t, W_z \in \calP^{n-t}_Z} |\la \psi | W_y \otimes W_z | \psi \ra|^2 = 2^n \sum_{a \in \FF_2^{n-t}} |\alpha_a|^4,
$$
where $\alpha_x \in \mathbb{C}$ are the amplitudes corresponding to expressing $\ket{\psi} = \sum_{x \in \FF_2^{n-t}} \alpha_x \ket{\phi_x}\ket{x}$ with $\ket{\phi_x}$ being a $t$-qubit state corresponding to the $(n-t)$-qubit computational basis state $\ket{x}$.
\end{prop}
\begin{proof}
We can write every $\ket{\psi} = \sum_{x \in \FF_2^{n-t}} \alpha_x \ket{\phi_x}\ket{x}$ where the amplitudes $\alpha_x \in \mathbb{C}$, $\ket{\phi_x}$ is a $t$-qubit state corresponding to the $(n-t)$-qubit computational basis state $\ket{x}$ and $\sum_x |\alpha_x|^2 = 1$. We then~have
\begin{align*}
\la \psi | W_y \otimes W_z | \psi \ra &= \sum_{a,b \in \FF_2^{n-t}} \alpha_a \alpha_b^\star \la \phi_b | W_y | \phi_a \ra \cdot \la b | W_z | a \ra \\
&= \sum_{a,b \in \FF_2^{n-t}} \alpha_a \alpha_b^\star \la \phi_b | W_y | \phi_a \ra (-1)^{z \cdot a} [a = b] \\
&= \sum_{a \in \FF_2^{n-t}} |\alpha_a|^2 (-1)^{z \cdot a} \la \phi_a | W_y | \phi_a \ra
\end{align*}

One can then show
\begin{align*}
\sum_{W_y \in \calP^t, W_z \in \calP^{n-t}_Z} |\la \psi | W_y \otimes W_z | \psi \ra|^2 &= \sum_{W_y \in \calP^t, W_z \in \calP^{n-t}_Z} \sum_{a,b \in \FF_2^{n-t}} |\alpha_a|^2 |\alpha_b|^2 \la \phi_a | W_y | \phi_a \ra \overline{\la \phi_b | W_y | \phi_b \ra} (-1)^{z \cdot (a+b)} \\
&=  \sum_{W_y \in \calP^t} \sum_{a,b \in \FF_2^{n-t}} |\alpha_a|^2 |\alpha_b|^2 \la \phi_a | W_y | \phi_a \ra \overline{\la \phi_b | W_y | \phi_b \ra} \sum_{W_z \in \calP^{n-t}_Z} (-1)^{z \cdot (a+b)} \\
&= \sum_{W_y \in \calP^t} \sum_{a,b \in \FF_2^{n-t}} |\alpha_a|^2 |\alpha_b|^2 \la \phi_a | W_y | \phi_a \ra \overline{\la \phi_b | W_y | \phi_b \ra} 2^{n-t} [a=b] \\
&= 2^{n-t} \sum_{W_y \in \calP^t} \sum_{a \in \FF_2^{n-t}} |\alpha_a|^4 \cdot |\la \phi_a | W_y | \phi_a \ra|^2 \\
&= 2^{n-t} \sum_{a \in \FF_2^{n-t}} |\alpha_a|^4 \cdot \sum_{W_y \in \calP^t} |\la \phi_a | W_y | \phi_a \ra|^2 \\
&= 2^{n-t} \sum_{a \in \FF_2^{n-t}} |\alpha_a|^4 \cdot 2^t \\
&= 2^n \sum_{a \in \FF_2^{n-t}} |\alpha_a|^4,
\end{align*}
where in the sixth line, we used that $\sum_{W_y \in \calP^t} |\la \phi_a | W_y | \phi_a \ra|^2 = 2^t \sum_{y \in \FF_2^{2t}} p_{\phi_a}(y) = 2^t$. This proves the desired result.
\end{proof}

We now give the proof of Lemma~\ref{lem:lower_bound_stab_dim_fidelity}.
\begin{proof}[Proof of Lemma~\ref{lem:lower_bound_stab_dim_fidelity}]
We can write every $\ket{\psi} = \sum_{x \in \FF_2^{n-k}} \alpha_x \ket{\phi_x}\ket{x}$ where the amplitudes $\alpha_x \in \mathbb{C}$, $\ket{\phi_x}$ is a $k$-qubit state corresponding to the $(n-k)$-qubit computational basis state $\ket{x}$ and $\sum_x |\alpha_x|^2 = 1$. Using Proposition~\ref{prop:prob_mass_Pk_PZm} and $W \subseteq \calP^k \times \calP_Z^{n-k}$, we have that
\begin{align*}
    \alpha \leq \sum_{x \in W} p_\Psi(x) \leq \frac{1}{2^n} \sum_{\substack{W_y \in \calP^k \\ W_z \in \calP^{n-k}_Z}} |\la \psi | W_y \otimes W_z | \psi \ra|^2 = \sum_{x \in \FF_2^{n-k}} |\alpha_x|^4 \leq \max_{x \in \FF_2^{n-k}} |\alpha_x|^2 \cdot \sum_{x \in \FF_2^{n-k}} |\alpha_x|^2 = \max_{x \in \FF_2^{n-k}} |\alpha_x|^2,
\end{align*}
where we used $\sum_x |\alpha_x|^2 = 1$ in the last inequality. Noting that 
$$
\max_{x \in \FF_2^{n-k}} | \la \psi | (\ket{\varphi_x} \otimes \ket{x} )|^2 = \max_{x \in \FF_2^{n-k}} |\alpha_x|^2 \geq \alpha,
$$
and choosing $\ket{\varphi} = \ket{\varphi_x}$ corresponding to the computational basis state $\ket{x}$ that maximizes the above, completes the proof.
\end{proof}

An immediate corollary of Lemma~\ref{lem:lower_bound_stab_dim_fidelity} is the following.
\begin{corollary}\label{lem:lower_bound_stab_dim_fidelity_gen}
Suppose $\ket{\psi}$ is an $n$-qubit state and $\alpha \in (0,1]$. If $W$ is a subgroup of Paulis such that there exists a Clifford unitary $U$ satisfying $UWU^\dagger \subseteq \calP^k \times \calP_Z^{n-k}$ and
$$
\sum_{x \in W} p_\Psi(x) \geq \alpha
$$
then there exists $\ket{\phi} = U^\dagger(\ket{\varphi} \otimes \ket{z}) \in \calS(n-k)$ where $\ket{\varphi}$ is a $k$-qubit state and $\ket{z}$ is an $(n-k)$-qubit computational basis state such that
$$
|\la \phi | \psi \ra|^2 \geq \sum_{x \in W} p_\Psi(x) \geq \alpha.
$$
\end{corollary}

We are now in a position to give a proof of Theorem~\ref{thm:inverse_gowers3_stab_dim}.
\begin{proof}[Proof of Theorem~\ref{thm:inverse_gowers3_stab_dim}]
In \cite[Proof of Theorem~4.2]{ad2024tolerant}, it was shown that if $\Exp_{x \sim q_\Psi}[2^n p_\Psi(x)] \geq \gamma$, then there exists a subgroup $V \subset \FF_2^{2n}$ of size $(4/\gamma) \cdot 2^n \geq |V| \geq \Omega(\gamma^{367})2^n$ such that
$$
\Exp_{x \in V}[2^n p_\Psi(x)] \geq \Omega(\gamma^{361}).
$$
Theorem~\ref{thm:stabilizer_covering_group} then implies that $V$ has a stabilizer covering of $O(\gamma^{-361})$ and there exists a Clifford circuit $U$ such that $UVU^\dagger = \calP^k \times \calP^m_Z$ where $m \leq n - k$ and $k = O(\log(\gamma^{-361}))$. Corollary~\ref{lem:lower_bound_stab_dim_fidelity_gen} then implies that there exists $\ket{\phi} \in \calS(n-k)$ such that
$$
|\la \phi| \psi \ra|^2 \geq \sum_{x \in V} p_\Psi(x) \geq \frac{|V|}{2^n} \cdot \Omega(\gamma^{361}) \geq \Omega(\gamma^{728}).
$$
This concludes the proof of our main theorem.
\end{proof}

\subsection{Tolerant testing high stabilizer-dimension states}
\label{sec:step2}
In this section, we show that one can tolerantly test the class of states with high stabilizer-dimension. 
Below, we first prove completeness in order to show that the tester proposed in~\cite{ad2024tolerant} can also be used to tolerantly-test states with high stabilizer dimension. 

\paragraph{Completeness analysis.} We now present the result on completeness, which extends the result from \cite{gross2021schur} applicable to stabilizer states, to states with high stabilizer dimension.
\begin{lemma}\label{lemma:completeness_stabdim}
Suppose $t < n$ and let $\ket{\psi}$ be an $n$-qubit pure state. If the fidelity of an $n$-qubit state $\ket{\psi}$ with $\calS(n-t)$ is at least $\varepsilon_1$, then\footnote{We remark that  for $t=0$, this reduces to Fact~\ref{fact:lower_bound_stabilizer_fidelity} which comments on a lower bound on $\Exp_{x \sim q_\Psi}\left[ |\la \psi | W_x | \psi \ra|^2 \right]$ when the stabilizer fidelity of $\ket{\psi}$ is known to be high.}  
$$
\Exp_{x \sim q_\Psi}\left[ |\la \psi | W_x | \psi \ra|^2 \right] \geq 2^{-2t} \cdot \varepsilon_1^6. 
$$
\end{lemma}
\begin{proof} 
From Fact~\ref{fact:relation_qPsi_and_pPsi}, we have that
\begin{align}
\Exp_{x \sim q_\Psi}\left[ |\la \psi | W_x | \psi \ra|^2 \right] = \frac{1}{2^n} \sum_{x \in \FF_2^{2n}} | \la \psi | W_x | \psi \ra|^6 &\geq \frac{1}{2^n} \sum_{W_y \in \calP^t, W_z \in \calP^{n-t}_Z} | \la \psi | W_y \otimes W_z | \psi \ra|^6 \nonumber \\
&\geq \frac{1}{2^{3n+2t}} \left[ \sum_{W_y \in \calP^t, W_z \in \calP^{n-t}_Z} | \la \psi | W_y \otimes W_z | \psi \ra|^2 \right]^3, \label{eq:lb_exp_weyl}
\end{align}
where the last inequality follows from using the relation between $p$-norms on real vectors $x \in \mathbb{R}^m$ via Holders inequality that states $\norm{x}_p \leq m^{1/p - 1/q} \norm{x}_q$, which we instantiate for $p=2$ and $q=6$. We note that the number of terms in the summation of the second inequality is of size $2^{n+t}$. We can write every $\ket{\psi} = \sum_{x \in \FF_2^{n-t}} \alpha_x \ket{\phi_x}\ket{x}$ where the amplitudes $\alpha_x \in \mathbb{C}$, $\ket{\phi_x}$ is a $t$-qubit state corresponding to the $(n-t)$-qubit computational basis state $\ket{x}$ and $\sum_x |\alpha_x|^2 = 1$. The given condition of the fidelity of $\ket{\psi}$ with $\calS(n-t)$ being at least $\varepsilon_1$ can be expressed as
\begin{equation}
\max_{x \in \FF_2^{n-t}} |\alpha_x|^2 = \max_{x \in \FF_2^{n-t}} | \la \psi | (\ket{\varphi_x} \otimes \ket{x} )|^2 \geq \varepsilon_1.
\end{equation}
Using Proposition~\ref{prop:prob_mass_Pk_PZm}, we then have
\begin{align*}
\sum_{W_y \in \calP^t, W_z \in \calP^{n-t}_Z} |\la \psi | W_y \otimes W_z | \psi \ra|^2 &= 2^n \sum_{a \in \FF_2^{n-t}} |\alpha_a|^4 
\geq  2^n \cdot (\max_{a \in \FF_2^{n-t}} |\alpha_a|^2)^2= 2^n \cdot \varepsilon_1^2,
\end{align*}
Substituting the above result into Eq.~\eqref{eq:lb_exp_weyl} gives us
$$
\Exp_{x \sim q_\Psi}\left[ |\la \psi | W_x | \psi \ra|^2 \right] \geq \frac{1}{2^{3n + 2t}} \left[ \sum_{W_y \in \calP^t, W_z \in \calP^{n-t}_Z} | \la \psi | W_y \otimes W_z | \psi \ra|^2 \right]^3 \geq \frac{1}{2^{3n + 2t}} \cdot 2^{3n} \varepsilon_1^6 = 2^{-2t} \varepsilon_1^6,
$$
which proves the desired result.
\end{proof}

\paragraph{Tolerant testing algorithm}
We now present the  problem of testing high stabilizer-dimension states in the tolerant framework and present our testing algorithm thereafter. 
\begin{problem}(Tolerant testing high stabilizer-dimension states)
Fix $0 \leq t \leq n$. Suppose an algorithm is given copies of an unknown $n$-qubit quantum state $\ket{\psi}$ promised
\begin{enumerate}[$(i)$]
    \item  $\ket{\psi}$ is $\varepsilon_1$-close to a a state in $\calS(n-t)$ in fidelity or
    \item  $\ket{\psi}$ is $\varepsilon_2$-far from all states in $\calS(n-t)$
\end{enumerate}
and the goal is to decide which is the case.
\end{problem}

Accompanied by the completeness result~(Lemma~\ref{lemma:completeness_stabdim}), we  show the following result.
\begin{theorem}
\label{thm:test_stab_dim_with_gowers}
Fix $\varepsilon_1 \in (0,1)$ and $t < n$. There is a constant $C > 1$ such that the following is true. There is an algorithm that given $\poly(2^{2t}/\varepsilon_1)$ copies of an $n$-qubit $\ket{\psi}$, can decide if $\max_{\ket{\phi}\in \calS(n-t)}|\langle\phi |\psi\rangle|^2 \geq  \varepsilon_1$ or $\max_{\ket{\phi}\in \calS(n-t)}|\langle\phi |\psi\rangle|^2 \leq  \varepsilon_2$ for every $\varepsilon_2 \leq \left(2^{-2t} \varepsilon_1^6\right)^C$ with gate complexity $n\cdot \poly(2^{2t}/\varepsilon_1)$. 
\end{theorem}
\begin{proof}
Let $\delta$ be a parameter that we fix at the end. The testing algorithm simply takes $O(1/\delta^2)$ copies of $\ket{\psi}$ and estimates $\Exp \limits_{x \sim q_\Psi}\left[|\la \psi | W_x | \psi \ra|^2 \right]$ up to additive error $\delta/2$. Observe that in the \emph{no} instance i.e.,  $\max_{\ket{\phi}\in \calS(n-t)}|\langle\phi |\psi\rangle|^2 \leq  \varepsilon_2$,  Theorem~\ref{thm:inverse_weyl_exp_states} implies that 
$$
\Exp \limits_{x \sim q_\Psi}\left[|\la \psi | W_x | \psi \ra|^2 \right] \leq O(\calF_\calS(\ket{\psi})^{1/C}) \leq O(\calF_{\calS(n-t)}(\ket{\psi})^{1/C}) = O(\varepsilon_2^{1/C}),
$$ 
where we used that $\calF_\calS(\ket{\psi}) \leq \calF_{\calS(n-t)}(\ket{\psi})$ for every $t\geq 0$ (i.e., stabilizer fidelity can only become grow for larger classes $\calS\subseteq \calS(n-t)$) and for some constant $C > 0$. In the \emph{yes} instance  i.e.,  $\max_{\ket{\phi}\in \calS(n-t)}|\langle\phi|\psi\rangle|^2~\geq~\varepsilon_1$, Lemma~\ref{lemma:completeness_stabdim} implies that
$$
\Exp \limits_{x \sim q_\Psi}\left[|\la \psi | W_x | \psi \ra|^2 \right] \geq  2^{-2t} \cdot \varepsilon_1^{6}.
$$
By letting 
$$
\delta=\frac{1}{10}\Big(2^{-2t} \cdot \varepsilon_1^{6} - O(\varepsilon_2^{1/C})\Big),
$$
the testing algorithm can distinguish between the \emph{yes} and \emph{no} instances of the problem using 
\begin{align*}
\poly(1/\delta)=\poly\Big(2^{-2t} \cdot\varepsilon_1^{6}- O(\varepsilon_2^{1/C})\Big)^{-1}=\poly(2^{2t}/\varepsilon_1)
\end{align*}
many copies of the unknown state~(Lemma~\ref{lem:est_weyl_exp}), where we used that $\varepsilon_2 \leq \left(2^{-2t} \varepsilon_1^6\right)^C$. The gate complexity of this protocol involves a factor-$n$ overhead in comparison to the sample complexity.
\end{proof}

\subsection{Learning algorithm}
\label{sec:learning_algo_improper_SC}
In this section, we now prove Theorem~\ref{thm:improper_SC} by algorithmizing the local inverse Gowers-$3$ theorem (Theorem~\ref{thm:inverse_gowers3_stab_dim}. Our algorithm will mostly use the subroutines discussed earlier in Section~\ref{sec:subroutines}. Particularly, our learning algorithm,  will only differ from the \Selfcorrection~algorithm in Theorem~Theorem~\ref{thm:restatementofselfcorr} in the step after having determined the subgroup $V$ which has high probability mass.
Overall, the learning algorithm has the following three-step approach
\begin{enumerate}
    \item First, given copies of $\ket{\psi}$ we use $\sample$ and the $\bsgtest$ to sample elements from $A^{(2)}$ which will also lie in $A^{(1)}$ with high probability. We proved this earlier in Section~\ref{sec:determinesubgroup}. 
    \item We next use the algorithmic $\PFR$ conjecture (Conjecture~\ref{conj:algopfrconjecture}) to give a polynomial-time procedure that uses samples from the first step to find a subgroup $V$ whose weight is high (i.e., $\Exp_{x\in V}[2^np_\Psi(x)]\geq \gamma$) and furthermore $|V|\sim 2^n\cdot \poly(\gamma)$. We proved this earlier in Section~\ref{sec:succintgroup}.
    \item Given the subgroup $V$ and copies of $\ket{\psi}$, using the $\symgramschmidt$ procedure and other subroutines, we find an $n$-qubit state $\ket{\phi} \in \calS(n-k)$ with $k \leq O(\log(1/\gamma))$ such that $|\langle \phi|\psi\rangle|^2\geq \poly(\gamma)$. This is where we deviate from the \Selfcorrection~algorithm and we will prove this part in the current section. Putting these steps together will prove Theorem~\ref{thm:improper_SC}.
\end{enumerate}

We now move on to proving the third item from above, which is stated formally below.
\begin{theorem}\label{thm:analysis_algo_stab_dim_state}
Let $\ket{\psi}$ be an $n$-qubit state with $\Exp_{x \sim q_\Psi}[2^n p_\Psi(x)] \geq \gamma$ and $W \subseteq \FF_2^{2n}$ be a subgroup such that $\Exp_{x \in W}[2^n p_\Psi(x)] \geq \Omega(\gamma^{C_1})$ (for a  constant $C_1 > 1$). Algorithm~\ref{algo:find_good_stab_dim} outputs a $\ket{\phi} \in \calS(n-k)$ with $k \leq O(\log(1/\gamma))$ such that $|\la \phi| \psi\ra|^2 \geq \Omega(\gamma^{C_1})$ with probability at least $1-\delta$, consuming 
$
\widetilde{O}\left(\poly(1/\gamma) \cdot \log (1/\delta)\right)
$ copies of $\ket{\psi}$ and 
$
\widetilde{O}\left(n^3 + n^2 \cdot \poly(1/\gamma) \cdot \log(1/\delta) \right)
$ total number of gates.
\end{theorem}

To prove this, we use the following lemma from \cite{chen2024stabilizer}, stated for pure quantum states.
\begin{lemma}[{\cite[Lemma~7.3]{chen2024stabilizer}}]\label{lem:ST_small_state}
Let $\uptau,\varepsilon,\delta > 0$, $t \in \mathbb{N}$, $\ket{\psi}$ be an $n$-qubit state.  Let $U$ be a Clifford  such that $U\ket{\psi}=\sum_x \alpha_x \ket{x}$ satisfies $\alpha:=\sum_{y\in \FF_2^t}|\alpha_{0^{n-t},y}|^2\geq \tau$. Given access to copies of $\ket{\psi}$, there exists an algorithm that outputs the density matrix of a state $\ket{\sigma_0}$ such that $|\langle\psi' | \sigma_0\rangle|^2  \geq 1-\varepsilon$ with probability at least $1-\delta$, where $\ket{\psi'}=\frac{1}{\sqrt{\alpha}}\sum_{y\in \FF_2^t}\alpha_{0^{n-t},y}\ket{y}$. The algorithm performs $2^{O(t)}\log(1/\delta)/(\varepsilon^2 \uptau)$ single-copy measurements on $\rho$ and takes $2^{O(t)}n^2 \log(1/\delta)/(\varepsilon^2\uptau)$ time.
\end{lemma}

\begin{lemma}[Folklore]\label{lem:fidelity_prod_state}
Let $t \in \mathbb{N}$, $z \in \{0,1\}^{n-t}$. Let $\ket{\phi} = \ket{\sigma} \otimes \ket{z}$ and $\ket{\psi}$ be two $n$-qubit states. Denote $\Psi' = \la z | \Psi| z \ra/ \Tr(\la z | \Psi | z \ra)$. We have
$$
|\la \psi | \phi \ra|^2 = \Tr(\la z |\Psi | z\ra) \cdot \la \sigma |\Psi'|\sigma\ra \leq \Tr(\la z | \Psi | z \ra).
$$
\end{lemma}

\begin{proof}[Proof of Theorem~\ref{thm:analysis_algo_stab_dim_state}]
Before giving the algorithm that proves the theorem, we first give an outline of the algorithm. We start off in a similar manner as we did in Theorem~\ref{thm:analysis_algo_stab_state}. Given the subgroup $V$ in terms of its basis, we first use Claim~\ref{claim:clifford_synthesis_subgroup}, which in turn uses $\symgramschmidt$, to determine the Clifford unitary $U$~(Fact~\ref{fact:clifford_action_on_group}) that transforms $V$ to the form of $\calP^k \otimes \calP^{m}_{\calZ}$,~i.e.,
$$
U V U^\dagger = \la Z_1, X_1, \ldots, Z_k, X_k, Z_{k+1}, Z_{k+2}, \ldots, Z_{k+m} \ra = \calP^k \times \calP_\calZ^m,
$$
where the values of $k$ and $m \leq n-k$ are determined as part of Claim~\ref{claim:clifford_synthesis_subgroup}. This consumes $O(n^3)$ time and $U$ has a gate complexity of $O(n^2)$. Furthermore, Theorem~\ref{thm:stabilizer_covering_group} guarantees that $k \leq \log(\gamma^{-C})$. We complete $W = UVU^\dagger$ to the full $n$-qubit Pauli group obtain $W'$:
$$
W' = \la Z_1, X_1, \ldots, Z_k, X_k, Z_{k+1}, Z_{k+2}, \ldots, Z_{k+m}, \ldots, Z_n \ra = \calP^k \times \calP_\calZ^{n-k}.
$$
From Corollary~\ref{lem:lower_bound_stab_dim_fidelity_gen} and proof of Theorem~\ref{thm:inverse_gowers3_stab_dim}, we know that there exists an $n$-qubit state $\in \calS(n-O(\log(1/\gamma))$ of the from $\ket{\sigma_{z^*}} \otimes \ket{{z^*}}$, where $\ket{z^*}$ is an $(n-k)$ computational basis state and $\ket{\sigma_{x^\star}}$ is a $k$-qubit state, such that $|\la \psi| U^\dagger (\ket{\sigma_{z^*}} \otimes \ket{z^*})|^2 \geq \poly(\gamma)$. In other words,
\begin{align}
\label{eq:genstate_and_z}
\max_{x\in \{0,1\}^{n-k}} |\la \psi|U^\dagger (\ket{\sigma_z} \otimes \ket{z})|^2\geq \poly(\gamma).
\end{align}
The above also implies that
\begin{equation}\label{eq:opt_stab_dim_state}
\ket{\sigma_{z^*}} \otimes \ket{{z^*}} = \argmax_{z\in \{0,1\}^{n-k}} |\la \psi|U^\dagger (\ket{\sigma_z} \otimes \ket{z})|^2.
\end{equation}
To determine this state $\ket{\sigma_{z^*}} \otimes \ket{{z^*}}$, we will use the following algorithm.
\begin{myalgorithm} \setstretch{1.35}
\begin{algorithm}[H]
    \setlength{\baselineskip}{1.4em} 
    \DontPrintSemicolon 
    \caption{\textsf{Find-High-Stab-Dim}($V,\gamma,\delta$)} \label{algo:find_good_stab_dim}
    \KwInput{Access to copies of $\ket{\psi}$ with $\Exp_{x\sim q_\Psi}\left[ |\la \psi | W_x | \psi \ra|^2 \right]\geq \gamma$, basis of subgroup $V$ such that $\sum_{x \in V} [2^n p_\Psi(x)] \geq \poly(\gamma)$ (as in Theorem~\ref{thm:analysis_algo_stab_state}), failure probability $\delta$.}
    \KwOutput{Output $(n-k)$ stabilizer-dimension state $\ket{\phi}$ such that $|\la \phi | \psi \ra|^2 \geq \Omega(\gamma^C)$}
    \vspace{2mm}
    Run $\symgramschmidt$, that on input the basis for $V$, outputs its centralizer $C_V$ (basis for $V\cap V^\perp$) and anti-commutant $A_V$ (basis for $V \backslash \la C_V \ra$). \\ 
    Obtain Clifford circuit $U$ and integers $m,k\geq 1$ such that $UVU^\dagger = \calP^k \times \calP_{\calZ}^m$ using Claim~\ref{claim:clifford_synthesis_subgroup} on inputs of $C_V$ and $A_V$ 
   \label{algostep_improper:find_clifford}\\
    Set $W \leftarrow UVU^\dagger$ and extend $W$ to $\widetilde{W} = \calP^k \times \calP_{\calZ}^{n-k}$ \\
    Let $N_r=\poly(1/\gamma)$ and initialize empty list $\calL \leftarrow \emptyset$.\\
    \For{$i=1:N_r$}{ \label{algostep_improper:sample_good_choices_z}
        Measure the last $(n-k)$ qubits of $U \ket{\psi}$ in the computational basis. Let the $(n-k)$-bit string be $z_i$. \\
        Set $U_i$ to be the Clifford unitary $X^{z_i} U$. \\
        Update $\calL \leftarrow \calL \cup \{U_i\}$
    }
    Estimate $\Tr(\la 0^{n-k} | U_i \Psi U_i^\dagger | 0^{n-k}\rangle)$ within error $\varepsilon=\poly(\gamma)$ for each $U_i \in \calL$ using the classical shadows protocol of Lemma~\ref{lem:shadows_stab_dim}. Let $U_{i^\star}$ be the one with the maximum estimate. \label{algostep_improper:fidelity_shadows} \\
    Use Lemma~\ref{lem:ST_small_state} with sample complexity $2^{O(k)}\log(1/\delta)/\poly(\gamma)$ to learn a $k$-qubit state $\ket{\sigma}$ with the Clifford unitary set to $U_{i^\star}$ such that $|\la \psi | U_{i^\star}^\dagger (\ket{\sigma} \otimes \ket{0^{n-k}})|^2 \geq \Omega(\gamma^C)$. \label{algostep_improper:ST} \\
    \Return $U_{i^\star}^\dagger(\ket{\sigma} \otimes \ket{0^{n-k}})$
\end{algorithm}
\end{myalgorithm}
Steps $1-3$ are as we defined before the algorithm. We now discuss  Lines~\ref{algostep_improper:sample_good_choices_z}-\ref{algostep_improper:ST}, Algorithm~\ref{algo:find_good_stab_dim}. Recall that the goal is to obtain a product state $\ket{\sigma_z} \otimes \ket{z}$ that achieves the guarantee of Eq.~\ref{eq:genstate_and_z}. We observe that by measuring the last $(n-k)$ qubits of $U \ket{\psi}$ in the computational basis, we obtain the $(n-k)$-bit binary string $m \in \{0,1\}^{n-k}$ with probability
$$
\Pr[m = z^*] = \max_{z \in \FF_2^{n-k}} |\la \psi | U^\dagger(\ket{\sigma_z} \otimes \ket{z})|^2 \geq \poly(\gamma),
$$
where we have used Eq.~\eqref{eq:opt_stab_dim_state}. So by measuring the last $(n-k)$ qubits of $U \ket{\psi}$ in the computational basis $N_r = O(\poly(1/\gamma)\log(1/\delta))$ many times, we ensure the collected strings $\{z_i\}_{i \in N_r}$ contains the optimal choice of $z^*$. To determine the optimal one, we use the following procedure. Let $\varepsilon \in (0,1)$ be an error parameter to be decided later. We consider the Clifford $X^{z_i}$ corresponding to each $z_i$ that would have mapped $\ket{z_i}$ to $\ket{0^{n-k}}$ and then utilize Lemma~\ref{lem:shadows_stab_dim} with error set to $\varepsilon/6$ as indicated in Line~\ref{algostep_improper:fidelity_shadows}, Algorithm~\ref{algo:find_good_stab_dim}, to estimate $\Tr(\la 0^{n-k}|U_i \Psi U_i^\dagger | 0^{n-k}\ra)$ with $U_i = X^{z_i} U$ for all $i \in [N_r]$. We then output the $U_i$ with the largest estimated value, which we set to be $U_{i^*}$. This satisfies 
$$
\Tr(\la 0^{n-k}|U_{i^*} \Psi U_{i^*}^\dagger | 0^{n-k}\ra) \geq \max_{z\in \{0,1\}^{n-k}} |\la \psi|U^\dagger (\ket{\sigma_z} \otimes \ket{z})|^2 - 2\varepsilon/3.
$$ 
We then apply Lemma~\ref{lem:ST_small_state} with error set to  $\varepsilon/3$ to learn the corresponding $k$-qubit state $\ket{\sigma}$  such that $\la \sigma | \Psi'|\sigma \ra \geq 1 - \varepsilon/3$, where $\Psi' = \la 0^{n-k}|U_{i^*} \Psi U_{i^*}^\dagger | 0^{n-k}\ra/\Tr(\la 0^{n-k}|U_{i^*} \Psi U_{i^*}^\dagger | 0^{n-k}\ra)$ Then, Lemma~\ref{lem:fidelity_prod_state} implies that
\begin{align*}
|\la \psi | U_{i^*}^\dagger (\ket{\sigma} \otimes \ket{z_{i^*}})|^2 
&= \Tr(\la 0^{n-k} | U_{i^*} \Psi U_{i^*}^\dagger|0^{n-k}\ra) \cdot |\la \psi'|\sigma\ra|^2 \\
&\geq (\max_{z\in \{0,1\}^{n-k}} |\la \psi|U^\dagger (\ket{\sigma_z} \otimes \ket{z})|^2 - 2\varepsilon/3)(1-\varepsilon/3) \\
&\geq \poly(\gamma) - \varepsilon,    
\end{align*}
where we have used that $\max_{z\in \{0,1\}^{n-k}} |\la \psi|U^\dagger (\ket{\sigma_z} \otimes \ket{z})|^2 \geq \poly(\gamma)$ is guaranteed by Theorem~\ref{thm:inverse_gowers3_stab_dim}. Setting $\varepsilon = \poly(\gamma)/100$ gives us the desired result.

Collection of $\{z_i\}$ required $O(\poly(1/\gamma)\log(1/\delta))$ copies of $\ket{\psi}$ and $O(n^2 \poly(1/\gamma)\log(1/\delta))$ time as $U$ has $O(n^2)$ gate complexity. Utilizing Lemma~\ref{lem:shadows_stab_dim} required sample complexity $O(1/\varepsilon^2 \cdot \poly(1/\gamma) \log(1/(\delta\gamma)))$ and $O(n^2/\varepsilon^2 \cdot \poly(1/\gamma) \log(1/(\delta\gamma)))$ time. Finally, applying Lemma~\ref{lem:ST_small_state} utilized a sample complexity of $O(1/\varepsilon^2 \poly(1/\gamma) \log(1/\delta))$ sample complexity and $O(n^2/\varepsilon^2 \poly(1/\gamma) \log(1/\delta))$ time as $k \leq O(\log(1/\gamma))$. Finally, we output $U_{i^\star}^\dagger(\ket{\sigma} \otimes \ket{0^{n-k}})$ by classical simulation using the fact that we have the classical description of $\ket{\sigma}$ in hand and a description of $U_{i^\star}$. This requires $O(n^2 \poly(1/\gamma))$ time~\cite{aaranson2004improved}.

Noting we set $\varepsilon=\poly(\gamma)$, the overall time complexity is
$$
\widetilde{O}\left(n^3 + n^2 \cdot \poly(1/\gamma) \cdot \log(1/\delta) \right),
$$
with main contributions due to cost of obtaining $U$, Lemma~\ref{lem:shadows_stab_dim} and Lemma~\ref{lem:ST_small_state}.
\end{proof}

The proof of Theorem~\ref{thm:improper_SC} then follows from the application of Theorem~\ref{thm:analysis_algo_stab_dim_state} in place of the last step of the proof of Theorem~\ref{thm:restatementofselfcorr}.

\newpage
\part{Learning structured decompositions}
In this section, we discuss how self-correction can be used as a subroutine to learn a \emph{structured} decomposition of an arbitrary $n$-qubit state $\ket{\psi}$ as follows:
\begin{equation}
\label{eq:structuredpart}
\ket{\psi} = \underbrace{\sum_{t=1}^k \beta_t \ket{\phi_t}}_{\text{rank-$k$ stabilizer decomposition}} + \alpha \underbrace{\ket{\phi_R}}_{\text{unstructured}},
\end{equation}
where $\ket{\phi_t}$s are $n$-qubit stabilizer states with $\beta \in \calB_\infty^{k},\alpha \in \calB_\infty$,\footnote{Recall that $\calB_\infty$ as the unit complex ball, i.e., $a\in \calB_\infty$ if $a \in \mathbb{C}$ and $|a| \in (0,1)$}  and $\ket{\phi_R}$ is \emph{unstructured}, by which we mean that its stabilizer fidelity is small. We prove this statement formally in Section~\ref{sec:iteratedselfcorrection} and in  Section~\ref{sec:applications} we give applications of this iterated self-correction algorithm.

\section{Iterative self-correction algorithm}
\label{sec:iteratedselfcorrection}
 Our main result in this section is to show that the structured part of $\ket{\psi}$ (in Eq.~\eqref{eq:structuredpart}), aka the stabilizer rank-$k$ decomposition $\ket{\widehat{\psi}} = \sum_{t=1}^k \beta_t \ket{\phi_t}$ (unnormalized) can be learned efficiently by using our $\Selfcorrection$ algorithm iteratively. This is summarized in the following theorem.

\begin{restatable}{theorem}{iterSCgen}
\label{thm:iterSC_gen}
    Let $\varepsilon,\upsilon\in (0,1)$, $\eta(\varepsilon)$ be a function of $\varepsilon$.\footnote{For technical reasons, we require that $\eta(\varepsilon
    )$ is a monotonically increasing function of $\varepsilon
    $. In the two theorems below, this will be the case.} Let $\ket{\psi}$ be an unknown $n$-qubit quantum state. Let $\calA$ be an algorithm that given copies of  $\ket{\psi}$ satisfying $\Exp_{x \sim q_\Psi}[2^n p_\Psi(x)] \geq \varepsilon$,  with probability at least $1-\delta$, outputs a stabilizer state $\ket{\phi}$ such that $|\la \phi| \psi \ra|^2 \geq \eta(\varepsilon)$. Let $\calS_A$ and  $\calT_\calA$ be the sample and time complexity of $\A$ respectively.
    
    Then, there is an algorithm $\calA'$ that with probability $\geq 1-\upsilon$, satisfies the following:  Given access to $U_\psi,\textsf{con}U_\psi$, outputs 
    $\beta \in \calB_\infty^k,\alpha \in \calB_\infty$  and stabilizer states $\{\ket{\phi_i}\}_{i\in [k]}$ such that one can write $\ket{\psi}$ as
    $$
      \ket{\psi}=\sum_{i\in [k]} \beta_i \ket{\phi_i}+\alpha\ket{\phi^\perp},
    $$
    where the residual state $\ket{\phi^\perp}$ satisfies $|\alpha|^2\cdot \stabfidelity{\ket{\phi^\perp}} < \varepsilon$ and $k \leq O(1/\eta(\varepsilon)^2)$. This algorithm $\calA'$ invokes $\calA$ $k$~times, hence the overall complexity is
    \begin{align*}
        & \text{ Query complexity: }\poly(n,1/\varepsilon,\log(1/\upsilon))\cdot \calS_\calA\\
        & \text{ Time complexity: } \poly(n,1/\varepsilon,\log(1/\upsilon))\cdot \calT_\calA.
    \end{align*}
\end{restatable}

In particular, if the let $\A$ be the $\Selfcorrection$ algorithm that we discussed in the previous part of this paper, then we obtain the following result. In Section~\ref{sec:boostrappingSC} we discuss the case when one applies the stabilizer booststrapping algorithm of~\cite{chen2024stabilizer} as our base algorithm.

\subsection{Useful subroutines}
We will invoke the following $\Selfcorrection$ theorem which we established earlier.

\selfcorrectionstatement*
Instead of estimating the Gowers-$3$ norm, we will work with its proxy $\Exp_{x \sim q_\Psi}[|\la \psi | W_x | \psi \ra|^2$~(Fact~\ref{fact:relation_expectation_paulis_qPsi_and_pPsi}) for which an inverse theorem exists~(Theorem~\ref{thm:inverse_weyl_exp_states}). A protocol to compute this quantity was established in prior work~\cite{gross2021schur,ad2024tolerant}. 
\begin{lemma}[{\cite[Lemma~3.8]{ad2024tolerant}}]
\label{lem:est_gowers3_norm_states}
Let $\ket{\psi}$ be an $n$-qubit state. We can estimate $\Exp_{x \sim q_\Psi}[|\la \psi | W_x | \psi \ra|^2]$ up to additive error~$\delta$ using $O(1/\delta^2)$ copies of $\ket{\phi}$ and $O(n/\delta^2)$ many one-qubit and two-qubit gates.
\end{lemma} 
We will use the Hadamard test to estimate overlaps of states.
\begin{lemma}[Hadamard test]
\label{lem:hadamardtest}
    Let $\ket{\psi},\ket{\psi'}$ be quantum states with state preparation unitaries $U_\psi$ and $U_{\psi'}$ respectively. The value of $\mathsf{Re}(\langle \psi|\psi'\rangle)\big)$ (or $\mathsf{Im}(\langle \psi|\psi'\rangle)\big)$) can then be estimated using the Hadamard circuit  up to error $\varepsilon$ using $O(1/\varepsilon^2)$ applications of controlled-$U_\psi,U_{\psi'}$.
\end{lemma}
We remark that this test is slightly different the usual SWAP test. In the test above, one first prepares $\ket{+}\ket{0^n}$, and controlled on $\ket{0}$ prepares $\ket{\psi}$ on the second register and controlled on $\ket{1}$ prepares $\ket{\psi'}$. So the algorithm prepares the state $\frac{1}{\sqrt{2}}\ket{0}\ket{\psi}+\frac{1}{\sqrt{2}}\ket{1}\ket{\psi'}$.  The algorithm then applies Hadamard on the first qubit and measures: the bias in measuring $0$ is exactly $\mathsf{Re}(\langle \psi|\psi'\rangle)\big)$. Repeating $O(1/\varepsilon^2)$ many times gives an $\varepsilon$-estimate of this quantity.

We will need a subroutine to estimate inner product between stabilizer states.
\begin{lemma}[{\cite[Algorithm~5.2]{garcia2014geometry}}]
\label{lem:inner_prod_stab_states}
    Given classical descriptions of two $n$-qubit stabilizer states $\ket{\phi}$ and $\ket{\psi}$, there exists an algorithm running in $O(n^3)$ time to compute $\la \phi | \psi \ra$.
\end{lemma}

\paragraph{Linear combination of unitaries.}
We will require a subroutine to prepare quantum states expressed as linear combination of stabilizer states and $\ket{\psi}$. To this end, we will turn to the method of linear combination of unitaries $\LCU$ introduced by Childs and Wiebe~\cite{childs2012lcu}, and has seen multiple improvements and applications over the years~\cite{berry2015sim,berry2015hamiltonian,childs2017linear}. For a set of unitaries $\{U_1,\ldots,U_k\}$, consider the linear combination of unitaries $V = \sum_{i=1}^k a_i U_i$. Assuming $a_i > 0$ to be non-negative (we can absorb the complex phases into the unitary $U_i$), consider the operators $U_{\PREP}$ and $U_{\SEL}$ defined as
\begin{align}
    U_{\SEL} := \sum_{i \in [k]} \ketbra{i}{i}  \otimes U_i, \quad
    U_{\PREP}\ket{0^m} := \frac{1}{\norm{a}_1} \sum_{i \in [k]} \sqrt{a_i} \ket{i}
    \label{eq:prep_sel_ops}
\end{align}
where $\norm{a}_1 = \sum_i |a_i|$ is the $1$-norm of the coefficients and $m = \ceil{\log_2 k}$. We then have the following lemma~\cite[Lemma~2.1]{kothari2014efficient} regarding the implementation of $W$. 
\begin{lemma}[$\LCU$ Lemma]
\label{lemma:lcu_basic} 
Let $V = \sum_{i=1}^k a_i U_i$ be a linear combination of unitary matrices $U_i$ with $a_i > 0$ for all $i$. Let $U_{\PREP}$, $U_{\SEL}$ be the unitary operators as defined in Eq.~\eqref{eq:prep_sel_ops} and $p:=1/\norm{a}_1^2$. Then, $W := (U_{\PREP}^\dagger \otimes I) U_{\SEL} (U_{\PREP} \otimes I)$ satisfies the following for all states $\ket{\phi}$
$$
W \ket{0^m} \ket{\phi} = \sqrt{p} \ket{0^m} V \ket{\phi} + \ket{\perp},
$$
where the unnormalized state $\ket{\perp}$ depends on $\ket{\phi}$ and satisfies $(\ketbra{0^m}{0^m} \otimes \id) \ket{\perp} = 0$. In particular, if the outcome $0^m$ is obtained in the ancilla qubits, the prepared state is $V \ket{\phi} / \norm{V \ket{\phi}}$ on the system register and this occurs with a probability of $p \cdot \norm{V \ket{\phi}}^2$.
\end{lemma}
The above \LCU~lemma requires queries to $U_{\PREP}$ and $U_{\SEL}$. To determine the cost of these queries, we will first use the following result from~\cite{mottonen2005prep,sun2023prep} 
to implement $U_{\PREP}$.
\begin{lemma}[State preparation]\label{lem:state_prep}
Let $m \in \mathbb{N}$. There exists a classical algorithm that outputs a quantum circuit $U$, running in time $O(m 2^m)$, that maps $\ket{0^m}$ to any arbitrary $m$-qubit quantum state $\ket{\phi}$. The circuit $U$ uses $O(2^m)$ \textsf{CNOT} gates and single-qubit rotations.
\end{lemma}

Secondly, we need an implementation of the select operator $U_{\SEL}$. We will utilize the following lemma from \cite{an2023lchs}, which has also appeared previously in~\cite{childs2017linear,low2018hamiltonian} in less general contexts.
\begin{lemma}[{\cite[Lemma~10, Supplementary]{an2023lchs}}]\label{lem:cost_select}
Let $\{U_i\}_{i \in [k]}$ be a set of $k$ unitaries. Then, the $U_{\SEL} := \sum_{i \in [k]} \ket{i} \bra{i} \otimes U_i$ can be constructed with $\ceil{\log_2(k)}$ queries to  $\textsf{con}U_i$(s).
\end{lemma}
The set of unitaries $\{U_j\}$ for us will include $U_\Psi$ and Clifford unitaries (which prepare particular stabilizer states). For the latter, we require the following algorithmic result that outputs a Clifford circuit preparing a specified stabilizer state.
\begin{lemma}[Clifford synthesis~\cite{dehaene2003clifford,patel2003efficient}]\label{lem:clifford_synthesis}
Given the classical description of an $n$-qubit stabilizer state $\ket{\phi}$, there is a quantum algorithm that outputs a Clifford circuit $C$ that prepares $\ket{\phi}$, using $O(n^2)$ many single-qubit and two-qubit Clifford gates.
\end{lemma}
To construct the controlled versions of $\{U_j\}$, we note that we have query access to $\textsf{con}U_\Psi$ and only need to comment on the controlled versions of the Clifford unitaries obtained via Lemma~\ref{lem:clifford_synthesis}. To obtain $\textsf{con}U_i$ of a Clifford circuit $U_i$, it is enough to put controls on each gate. Further, the resulting Toffoli gates ($\textsf{CCNOT}$) from controlling \textsf{CNOT}, can be decomposed into constantly many \textsf{CNOT} gates and single-qubit rotations. Similarly, controlled rotations can be decomposed into constantly many \textsf{CNOT} gates and single-qubit rotations. The resulting single-qubit and two-qubit gate complexity remains unchanged i.e., $O(n^2)$.

Putting everything together (by invoking Lemma~\ref{lem:state_prep} for $m=\textsf{polylog}(n)$), we obtain the following corollary of the $\LCU$ lemma that we stated above (specialized to our setting).  
\begin{corollary}
\label{cor:LCUfinal}
 Let $\calU=\{C_i\}_i\cup U_\Psi$, where $C_i$ are $n$-qubit Clifford circuits and   $U_\Psi$ is a state preparation unitary for $n$-qubit $\ket{\psi}$. Let $V=\sum_{i\in [k]} a_i U_i$ where $U_i\in \calU$ and $a_i>0$ for all $i$. Given access to $\textsf{con}U_i$s for all $U_i\in \calU$, there is a $\poly(n,k)$-time quantum algorithm that \emph{implements} a unitary $W$ that satisfies the following: with probability $\geq (\norm{V \ket{\phi}}/\|a\|_1)^2$, we have
 $
W \ket{\phi} =  V \ket{\phi}/\norm{V \ket{\phi}}.
$
\end{corollary}

\subsection{Error-free iterative self-correction}
In this section, we first present the iterative $\Selfcorrection$ algorithm assuming all the subroutines that we use (like Gowers norm estimation, $\Selfcorrection$, Hadamard test and the LCU lemma~\ref{lemma:lcu_basic}) are error-free and succeed  with probability $1$. This simplifies the presentation of the algorithm (and consequently the analysis), distilling the main idea in iterative $\Selfcorrection$. Our goal will be to prove Theorem~\ref{thm:iterSC_gen} in the error-free case.

\subsubsection{Algorithm}
We  present the error-free iterative $\Selfcorrection$ in Algorithm~\ref{alg:iterSC_error_free}. Given access to copies of $\ket{\psi}$ and state-preparation unitary $U_\psi$ (and controlled version $\textsf{con}U_\psi$), the  algorithm outputs a state $\ket{\widehat{\psi}}$ (not necessarily normalized) which can be expressed as a linear combination of stabilizer states
\begin{equation}\label{eq:final_estimate_iterSC}
    \ket{\widehat{\psi}} = \sum_{i=1}^k \beta_i \ket{\phi_i},
\end{equation}
where $\ket{\phi_i} \in \cal{S}$ are stabilizer states and $\beta \in \calB_\infty^k$ are their corresponding coefficients. Moreover, as indicated in Theorem~\ref{thm:self_correction},  the state proportional to $\ket{\psi} - \ket{\widehat{\psi}}$ is shown to have Gowers-$3$ norm $\leq \varepsilon^6$, which implies its stabilizer fidelity $\leq \varepsilon$. In Algorithm~\ref{alg:iterSC_error_free}, note that we output the classical description of $\ket{\widehat{\psi}}$ as a list of $\{\beta_i\}_{i \in [k]}$ and $\{\ket{\phi_i}\}_{i \in [k]}$ where each stabilizer state is described by its~$n$ generators.

The algorithm constructs $\ket{\widehat{\psi}}$, i.e., $\{\beta_i\}_{i \in [k]},\{\ket{\phi_i}\}_{i \in [k]}$ progressively across $k$ many iterations by learning $\ket{\phi_i}$ and $\beta_i$ one at a time in a sequential manner. In the $t$th iteration, the running estimate, denoted by $\ket{\widehat{\psi}^{(t)}}$, is thus
\begin{equation}\label{eq:running_estimate_iterSC}
    \ket{\widehat{\psi}^{(t)}} = \sum_{i=1}^t \beta_i \ket{\phi_i}.
\end{equation}
The corresponding residual state, denoted by $\ket{\psi_{t+1}}$, is then
\begin{equation}\label{eq:residual_state_iterSC}
    \ket{\psi_{t+1}} = \Big(\ket{\psi} - \ket{\widehat{\psi}^{(t)}}\Big)/\alpha_{t+1},
\end{equation}
where $\alpha_{t+1} \in \mathbb{R}$ is a normalization factor to ensure $\ket{\psi_{t+1}}$ is a valid quantum state. For the first iteration, we denote the residual state as $\ket{\psi_1} = \ket{\psi}$. We stop after $k$ iterations when either of the two following conditions are met
\begin{equation}\label{eq:stopping_cond_iterSC}
    \Exp_{x \sim q_{\Psi_k}}[|\la \psi_{k} | W_x | \psi_{k} \ra|^2] < \varepsilon^6 \quad \emph{or} \quad |\alpha_t|^2 < \varepsilon.
\end{equation}
The former condition of $\Exp_{x \sim q_{\Psi_k}}[|\la \psi_{k} | W_x | \psi_{k} \ra|^2] < \varepsilon^6$ implies that $\calF_\calS(\ket{\psi_k}) < \varepsilon$ by Fact~\ref{fact:lower_bound_stabilizer_fidelity}, hence is a valid stopping point in order to satisfy the requirements of our theorem statement. The latter condition $|\alpha_t|^2 < \varepsilon$ implies that the state $\ket{\widehat{\psi}^{(t)}}$ is already close to the unknown $\ket{\psi}$ in terms of $\ell_2$ distance, which again is a valid stopping point for the algorithm (since this is a stronger condition than self-correction even). So Algorithm~\ref{alg:iterSC_error_free}  checks these before applying $\Selfcorrection$ on $\ket{\psi_k}$ (see Lines~\ref{line_algo:stopping_cond_iterSC_error_free_tomography},\ref{line_algo:stopping_cond_iterSC_error_free}). As part of our analysis, we will show that the total number of iterations $k$ can be bounded. We now provide the details on how $\ket{\phi_i}$ and the corresponding coefficient $\beta_i$ is determined in each iteration of the algorithm by going through iterations $1,2$, before formally proving the correctness and complexity of the algorithm later. 
\begin{myalgorithm} 
\setstretch{1.35}
\begin{algorithm}[H]
    \caption{Error-free iterative self-correction}\label{alg:iterSC_error_free}
    \setlength{\baselineskip}{1.5em} 
    \DontPrintSemicolon 
    \KwInput{$\varepsilon \in (0,1)$, copies of  $\ket{\psi}$, access to $U_\psi$ (and $\textsf{con}U_\psi$) }
    \KwOutput{List of stabilizer states $L = \{\ket{\phi_i}\}_{i \in [k]}$, coefficients $B = \{\beta_i\}_{i \in [k]}$ for some $k \in \mathbb{N}$}
    
    Set $\eta = C_1 \varepsilon^{6C_2}$ (with constants $C_1, C_2$ as defined in Theorem~\ref{thm:self_correction}). \\
    Let $r_0=\alpha_1=1$. \label{line_algo:r0_def_iterSC_error_free} \\
    Set $L=\varnothing$, $B=\varnothing$, $t_{\max}=\left \lceil1 /\eta^2\right\rceil$.\\
    \For{$t=1$ \KwTo $t_{\max}$}{
        \lIf{$|\alpha_t|^2 < \varepsilon$}{Let $\ket{\phi^\perp} = \ket{\psi_t}$ and break from loop} \label{line_algo:stopping_cond_iterSC_error_free_tomography}
        Run $\LCU$ to prepare $\ket{\psi_{t}}=V_{t}\ket{0^n}$ where $V_{t} = (U - \sum_{j=1}^t \beta_j W_j)/\alpha_{t}$, for $t\geq 2$ and set $\ket{\psi_t} = \ket{\psi}$ for $t=1$. \\
        Estimate $\Exp_{x \sim q_{\Psi_t}}[|\la \psi_{t} | W_x | \psi_{t} \ra|^2]$ given copies of $\ket{\psi_{t}}$ using Lemma~\ref{lem:est_gowers3_norm_states} \\
        \lIf {$\Exp_{x \sim q_{\Psi_t}}[|\la \psi_{t} | W_x | \psi_{t} \ra|^2] < \varepsilon^6$}{Let $\ket{\phi^\perp} = \ket{\psi_t}$ and break from loop} \label{line_algo:stopping_cond_iterSC_error_free}
        Run $\Selfcorrection$ on copies of $\ket{\psi_t}$ to learn $\ket{\phi_t} \in \textsf{Stab}$ s.t. $|\la \phi_t | \psi_t \ra|^2 \geq \eta$. \\
        Obtain the Clifford unitary $W_t$ that prepares $\ket{\phi_t} = W_t \ket{0^n}$ using Lemma~\ref{lem:clifford_isotropic_subspace}.\\
        Estimate $\langle \phi_t|\psi \rangle$ via the Hadamard test in Lemma~\ref{lem:hadamardtest} using $\textsf{con}W_t$ and $\textsf{con}U_\psi$.\\
        Compute $\la \phi_t | \phi_j \ra$  for all $j \in [t-1]$ classically using Lemma~\ref{lem:inner_prod_stab_states}. \\
        Set $\beta_t = \la \phi_t| \psi \ra - \sum_{j=1}^{t-1} \beta_j \la \phi_t | \phi_j \ra$.\label{eq:defnofbetainalgo11} \\
        Set $c_t = \beta_t/\left(\prod_{j=0}^{t-1} r_t\right), r_t=\sqrt{1-|c_t|^2}$ and $\alpha_{t+1} = \prod_{j=0}^t |r_j|$. \\
    }
    \Return List of $k \leq 1/\eta^2$ stabilizer states $L=\{\ket{\phi_i}\}_{i \in [k]}$ along with their coefficients $B=\{\beta_\ell=c_\ell \prod_{t=1}^{\ell-1} r_t\}_{\ell\in [k]}$ held classically.
\end{algorithm}
\end{myalgorithm}

\paragraph{Iteration $t=1$.} When $t=1$, let $\ket{\psi_1} = \ket{\psi}$ and let $r_0 = 1$. Suppose $\Exp_{x \sim q_\Psi}[|\la \psi | W_x | \psi \ra|^2] \geq  \varepsilon^6$ then performing $\Selfcorrection$ on $\ket{\psi_1}$ gives us a stabilizer state $\ket{\phi_1}$ such that $|\la \phi_1 | \psi_1 \ra|^2 \geq \eta$ (from Theorem~\ref{thm:self_correction}) where $\eta = \varepsilon^{6C}$ (for an universal constant $C>1$).\footnote{We remark that to perform $\Selfcorrection$ for the first iteration, we do not require access to the unitary $U_\psi$.} We can find the Clifford unitary $W_1$ that prepares the stabilizer state $\ket{\phi_1}$ as $\ket{\phi_1} = W_1 \ket{0^n}$ using Lemma~\ref{lem:clifford_isotropic_subspace}. To determine the corresponding coefficient $\beta_1$ as in Eq.~\eqref{eq:running_estimate_iterSC}, we compute the inner product $\la \psi_1| \phi_1 \ra$ via the Hadamard test (Lemma~\ref{lem:hadamardtest}) using  $\textsf{con}U_\psi$ and $\textsf{con}W_t$. We will also denote $c_1 = \la \psi_1| \phi_1 \ra$ for convenience. The decomposition we can obtain so far is then
$$
\ket{\psi} = \beta_1 \ket{\phi_1} + r_1 \ket{\phi_1^\perp},
$$
where $\ket{\phi_1^\perp}$ is a state orthogonal to $\ket{\phi_1}$, $\beta_1 = c_1$, and $r_1 = \sqrt{1 - |\beta_1|^2}$ (we can absorb any phase into $\ket{\phi_1^\perp}$ without loss of generality). 

\paragraph{Iteration $t=2$.} We now consider iteration $t=2$. First observe that we only proceed to iteration $t=2$ if $|\alpha_2|^2 = |r_1|^2 \geq \varepsilon$ (which can be checked as $r_1$ was obtained in iteration $t=1$). We then need to check if $\Exp_{x \sim q_{\Psi_2}}[|\la \psi_2 | W_x | \psi_2 \ra|^2] \geq  \varepsilon^6$, which we can estimate using Lemma~\ref{lem:est_gowers3_norm_states}, but we require the ability to access copies of $\ket{\psi_2} = \ket{\phi_1^\perp}$ to continue. We observe that this can be obtained via the usual $\LCU$ lemma (in Lemma~\ref{lemma:lcu_basic}) on the unitary $V_2 = \frac{1}{r_1} \left(U_1 - c_1 W_1\right)$ since $\ket{\psi_2} = V_2 \ket{0^n}$. Furthermore, the probability of success of $\LCU$ is 
$$
\left(\frac{\norm{V_2 \ket{0^n}}}{\frac{1}{r_1} + \frac{|c_1|}{r_1}}\right)^2 = \left(\frac{r_1}{1 + |c_1|}\right)^2 \geq \frac{\varepsilon^2}{4},
$$
where we used that $V_2 \ket{0^n}$ produces a normalized quantum state $\ket{\psi_2}$ by construction and $|c_1| \leq 1$, $|r_1|^2 \geq \varepsilon$. We will assume sample access to $\ket{\psi_2}$ from now onwards (with an overhead of $1/\varepsilon^2$ in sample complexity). 

Suppose $\Exp_{x \sim q_{\Psi_2}}[|\la \psi_2 | W_x | \psi_2 \ra|^2] \geq  \varepsilon^6$. We then again carry out $\Selfcorrection$ on $\ket{\psi_2}$ to obtain the stabilizer state $\ket{\phi_2}$ such that $|\la \phi_2 | \psi_2 \ra|^2 \geq \eta$. Also, we can determine a Clifford unitary $W_2$ that prepares $\ket{\phi_2}$, i.e., $\ket{\phi_2} = W_2 \ket{0^n}$ using Lemma~\ref{lem:clifford_isotropic_subspace}. 
At this stage, the goal is to estimate the coefficient $\beta_2$ corresponding to $\ket{\phi_2}$ in the decomposition of $\ket{\psi}$. One approach is to determine the inner product $c_2 = \langle \phi_2 | \psi_2 \rangle$ and then set $\beta_2 = c_2 r_1$. If we were to estimate this directly using the Hadamard test, this would require $\textsf{con}W_2$ and $\textsf{con}V_2$. The former $\textsf{con}W_2$ can be constructed from the Clifford unitary $W_2$ by controllizing all the gates in $W_2$. However, to construct $\textsf{con}V_2$ via $\LCU$, we would require \emph{controlled-controlled} unitary $U$ which is not available to the algorithm.

Instead, we estimate $\beta_2$ via
\begin{equation}\label{eq:expression_beta2}
\beta_2 = c_2 r_1 = r_1 \la \phi_2 | \psi_2 \ra = \la \phi_2 | \psi \ra - c_1 \la \phi_2 | \phi_1 \ra, 
\end{equation}
where we have used $\ket{\psi_2} = (\ket{\psi} - c_1 \ket{\phi_1})/r_1$ and definition of $c_2:= \la \phi_2 | \psi_2 \ra$ in the third equality. We can estimate $\la \phi_2 | \psi \ra$ via the Hadamard test using $\textsf{con}W_2$ and $\textsf{con}U_\psi$. We can also estimate $\la \phi_2 | \phi_1 \ra$ classically (and exactly) via Lemma~\ref{lem:inner_prod_stab_states}. Using Eq.~\eqref{eq:expression_beta2}, we would then determine $\beta_2$ and can then set $c_2 = \beta_2/r_1$ with $r_1$ having been determined at the end of iteration $1$. Hence, the decomposition the algorithm obtains at the end of this iteration is
$$
\ket{\psi} = \underbrace{\beta_1}_{=c_1} \ket{\phi_1} + \underbrace{\beta_2}_{=c_2 r_1} \ket{\phi_2} + \underbrace{\alpha_3}_{=r_1 r_2} \ket{\phi_2^\perp},
$$
where $\ket{\phi_2^\perp}$ is a state orthogonal to $\ket{\phi_2}$ and $r_2 = \sqrt{1 - |c_2|^2}$ (absorbing the phase into $\ket{\phi_2^\perp}$). 

\paragraph{For subsequent iterations.} In the next iteration ($t=3$), we only proceed if $|\alpha_3| = |r_1 r_2|^2 \geq \varepsilon$ (which can be checked as $r_2$ was obtained in iteration $t=2$) and if $\Exp_{x \sim q_{\Psi_3}}[|\la \psi_3 | W_x | \psi_3 \ra|^2] \geq  \varepsilon$ (estimated using Lemma~\ref{lem:est_gowers3_norm_states}). To ensure access to copies of $\ket{\psi_3} = \ket{\phi_2^\perp}$, we proceed as we had in iteration $t=2$. We prepare $\ket{\psi_3}=V_3 \ket{0^n}$ via $\LCU$ (Lemma~\ref{lemma:lcu_basic}) on the unitary $V_3 = \frac{1}{\alpha_3} \left(U_1 - \beta_1 W_1 - \beta_2 W_2\right)$. The probability of success of $\LCU$ is 
$$
\left(\frac{\norm{V_3 \ket{0}^n}}{(\frac{1}{\alpha_3} + \frac{|\beta_1|}{\alpha_3} + \frac{|\beta_2|}{\alpha_3})}\right)^2 = \left(\frac{\alpha_3}{1 + |\beta_1| + |\beta_2|} \right)^2 \geq \frac{\varepsilon^2}{9},
$$ 
as $V_2 \ket{0^n}$ produces a normalized quantum state $\ket{\psi_3}$ by construction, $|\beta_1| = |c_1| \leq 1, |\beta_2| = |c_2 r_1| \leq 1$, and $|\alpha_3|^2 \geq \varepsilon$ (checked as part of the stopping condition). We are thus now ready to proceed through iteration $t=3$. We thus keep repeating this iterative process until we reach the $k$th iteration where the $\Exp_{x \sim q_{\Psi_k}}[|\la \psi_k | W_x | \psi_k \ra|^2] < \varepsilon^6$ or $|\alpha_k|^2 = \prod_{t=0}^{k-1} |r_j|^2 < \varepsilon$.

\subsubsection{Analysis (error-free)}
\label{sec:errorfreeSelfcorrection}
In this section, we will analyze the correctness, the number of steps the algorithm proceeds for and complexity of the iterative $\Selfcorrection$ procedure assuming $\LCU$, Hadamard test and Gowers-$3$ norm estimation procedures are \emph{error-free} and assume the subroutines succeed with probability $1$. In the section thereafter we will incorporate the errors introduced by these procedures. 
\paragraph{Stopping condition.} 
We now prove that the algorithm stops after $k\leq 1/\eta^2$ iterations.  Suppose we stop after $k$ iterations i.e., when $\Exp_{x \sim q_{\Psi_{k+1}}}[|\la \psi_{k+1} | W_x | \psi_{k+1} \ra|^2] < \varepsilon^6$ or $\prod_{j=1}^{k} |r_j|^2 < \varepsilon$.

Before we delve into the proofs, we first describe some notation that will be convenient. Recall that for every $t\geq 1$, we defined
$$
\ket{\psi_{t+1}}=\frac{\ket{\psi}-\sum_{j=1}^{t}\beta_j\ket{\phi_j}}{\prod_{j=1}^{t}|r_j|}
$$
by Line $(14)$ of the algorithm. Let us denote
$
\ket{\widehat{\psi}_t} = \sum_{j=1}^{t} \beta_j \ket{\phi_j}
$
where $\beta_j = c_j \prod_{i=1}^{j-1} r_i$. Additionally, recall that $\alpha_t=\prod_{j=1}^t|r_j|$ by Line $(13)$ of the algorithm. Using this notation, observe~that
\begin{align}
\label{eq:psit+1}
\ket{\psi_{t+1}}=\frac{\ket{\psi_t}-c_t\ket{\phi_t}}{|r_t|} \text{ and }\ket{\psi_{t+1}}=\frac{\ket{\psi}-\ket{\widehat{\psi}_t}}{\alpha_t}.
\end{align}
Both the expressions of the same state will be useful in the analysis. 
Below, it will be convenient to work with the unnormalized versions of the states $\ket{\psi_t}$, and to this end, denote the $\ket{\Psi_{t+1}} = \prod_{j=1}^{t} r_j \ket{\psi_{t}}=\alpha_t \ket{\psi_t}$. In the algorithm, observe that $\Psi_1 = \ket{\psi}$. Often, we will denote $\ket{\Psi_t}$ as simply $\Psi_t$. We first make a few observations about the iterative $\Selfcorrection$ procedure. Observe that for every $t$, the state $\ket{\psi_{t+1}}$ is orthogonal to the state $\ket{\phi_{t}}$. This is because
\begin{align}
\label{eq:psidecompositiononestep}
\ket{\psi_{t+1}} = \frac{1}{|r_{t}|}\left(\ket{\psi_{t}} - c_t \ket{\phi_{t}} \right),
\end{align}
and its inner product with $|\phi_{t}\rangle$ equals $\langle \phi_t| \psi_{t+1} \rangle = (\langle \phi_{t}|\psi_{t}\rangle -c_{t})/r_t$. Now recall that we let $c_t = \beta_t/\alpha_t$ during the algorithm which can be seen to be
\begin{align}
\label{eq:reformulationofct}
c_{t}=\frac{\beta_t}{\prod_{j=0}^{t-1}r_t}=\frac{\la \phi_t|\big(| \psi \ra - \sum_{j=1}^{t-1} \beta_j | \phi_j\big) \ra}{\prod_{j=0}^{t-1}r_t} = \frac{\la \phi_t|\big(| \psi \ra -| \widehat{\psi}_{t-1}\rangle\big)}{\alpha_{t-1}}=\langle \phi_{t}|\psi_{t}\rangle,
\end{align}
where we used the definitions of $\alpha_t,\ket{\widehat{\psi}_t}$ that we established above and Eq.~\eqref{eq:psit+1}.  Hence $\langle \phi_t| \psi_{t+1} \rangle=0$.
 Alternatively the decomposition in Eq.~\eqref{eq:psidecompositiononestep} be rewritten as
\begin{align}
\label{eq:psit+1psit}
\frac{1}{\prod_{j=1}^{t} r_j}\ket{\Psi_{t+1}} = \frac{1}{|r_{t}|}\left(\frac{1}{\prod_{j=1}^{t-1} r_j}\ket{\Psi_{t}} - c_t \ket{\phi_{t}} \right) \implies
\ket{\Psi_{t+1}} = \ket{\Psi_{t}} - c_t \prod_{j=1}^{t-1} r_j \ket{\phi_t},
\end{align}
where we used above that one can absorb the phase of $r_t$ into $\Psi_{t+1}$ without loss. Using this, one can recursively write Eq.~\eqref{eq:psit+1psit} as 
\begin{align}
\label{eq:defnofPsit}    
\ket{\Psi_t} = \ket{\psi} - \ket{\widehat{\psi}_{t-1}}.
\end{align}
In particular, Eq.~\eqref{eq:reformulationofct} implies that $|c_t|\leq 1$ and since $|r_j|=\sqrt{1-|c_j|^2}\leq 1$ for all $j$, we have that $|\beta_t|\leq 1$ for all $t$. Writing the decomposition in Eq.~\eqref{eq:psidecompositiononestep} out iteratively one can observe that we have decomposed the quantum state $\ket{\psi}$ as~follows 
$$
\ket{\psi} = c_1 \ket{\phi_1} + r_1 c_2 \ket{\phi_2} + r_1 r_2 c_3 \ket{\phi_3} + \cdots +\prod_{i=1}^k r_i \ket{\phi_k^\perp},
$$
where each $|r_j|^2 = 1 - |c_j|^2 \leq 1 - \eta$ since $|c_j|^2 \geq \eta$ by the promise of $\Selfcorrection$ in Theorem~\ref{thm:self_correction}. More concisely one can write $\ket{\psi}$ above as 
\begin{align}
\label{eq:psidecomposition}
\ket{\psi} = \sum_{t=1}^k \beta_t \ket{\phi_t} + \alpha_{k+1} \ket{\phi_k^\perp},
\end{align}
where $\beta_t = c_t \prod_{j=1}^{t-1} r_t$ for $t \leq k$ and $\alpha_{k+1} = \prod_{i=1}^k r_i$. Note that for all $t\leq k$, we have that
$$
|\beta_t| = |c_t \prod_{j=1}^{t-1} r_t| = |c_t| \prod_{j=1}^{t-1} |r_t| \geq \eta
$$
since $|c_t| \geq \sqrt{\eta}$ and $\prod_{j=1}^{t-1} |r_t| \geq \sqrt{\varepsilon} \geq \sqrt{\eta}$ (since $\eta = \varepsilon^{6C}$ for $C>1$) for all $t\leq k$ (by step $(6)$ requirement of the algorithm).

We are now ready to prove  two claims that will give our upper bound on the stopping criterion.

\begin{claim}\label{claim:diff_l2_norm_residuals_iterSC_errorfree}
For $t \leq k$, we have that
$$
\norm{\Psi_t}_2^2 - \norm{\Psi_{t+1}}_2^2 \geq \eta^2
$$    
\end{claim}
\begin{proof}
Firstly, we note that 
\begin{align*}
    \norm{\Psi_t}_2^2 - \norm{\Psi_{t+1}}_2^2 &= \norm{\Psi_t - \Psi_{t+1} + \Psi_{t+1}}_2^2 - \norm{\Psi_{t+1}}_2^2 \\
    &= \norm{\Psi_t - \Psi_{t+1}}_2^2 + \norm{\Psi_{t+1}}_2^2 - \norm{\Psi_{t+1}}_2^2\\
    &= \norm{\Psi_t - \Psi_{t+1}}_2^2,
\end{align*}
where the second equality follows from $\|a+b\|_2^2=\|a\|_2^2+\|b\|_2^2+2\textsf{Re}(\langle a,b\rangle)$ and noting that $\Psi_t - \Psi_{t+1}$ is orthogonal to $\Psi_{t+1}$, i.e., 
\begin{align}
\label{eq:psitandt+1orthogonal}
\la \Psi_t - \Psi_{t+1} | \Psi_{t+1} \ra=\la \widehat{\psi}_{t} - \widehat{\psi}_{t-1} | \Psi_{t+1} \ra = \beta_{t}^*\langle \phi_t | \Psi_{t+1} \ra=\beta_t^* \prod_{j=1}^{t} r_j^\star \cdot \la \phi_t | \psi_{t+1} \ra = 0
\end{align}
as $\la \phi_t | \psi_{t+1} \ra = 0$ by construction (see around Eq.~\eqref{eq:psidecompositiononestep} for a proof).
To conclude, we observe that
$$
\norm{\Psi_t}_2^2 - \norm{\Psi_{t+1}}_2^2 = \norm{\Psi_t - \Psi_{t+1}}_2^2 = \norm{c_t^\star \prod_{j=1}^{t-1} r_j^\star \ket{\phi_t}}_2^2 = |c_t|^2 \prod_{j=1}^{t-1} |r_j|^2 \geq \eta^2
$$
where we used $|c_t|^2 \geq \eta$ by the promise of self-correction and $\prod_{j=1}^{t-1} |r_j|^2 \geq \eta$ for all $t\leq k$.
\end{proof}

\begin{claim}\label{claim:ub_k_iterSC_errorfree}
$k \leq 1/\eta^2$.    
\end{claim}
\begin{proof}
Using Claim~\ref{claim:diff_l2_norm_residuals_iterSC_errorfree} and adding over $t=1,\ldots,k$ gives us $\norm{\Psi_1}_2^2 - \norm{\Psi_{k+1}}_2^2 \geq k \eta^2$. Furthermore we have that
$$
\norm{\Psi_1}_2^2 - \norm{\Psi_{k+1}}_2^2 =\|\ket{\psi_1}\|_2^2-\prod_{j=1}^{k+1} |r_j|^2 \|\ket{\psi_{k+2}}\|_2^2 \leq \|\ket{\psi_1}\|_2^2 =1,
$$
where we used that $\Psi_1=\ket{\psi_1}$. Hence, we have that $k \eta^2 \leq 1$ and thus $k \leq 1/\eta^2$.
\end{proof}
This concludes the proof that the algorithm runs for at most $k\leq 1/\eta^2$ steps.

\textbf{Correctness.}
The correctness of the algorithm is immediate. In each step, $\Selfcorrection$ produces a stabilizer state (hence $\ket{\phi_1}.\ldots,\ket{\phi_k}$ are stabilizer states) and at the final step we stop because either $\Exp_{x \sim q_{\Psi_{k+1}}}[|\la \psi_{k+1} | W_x | \psi_{k+1} \ra|^2] < \varepsilon^6$ \emph{or} $|\alpha_{k+1}|^2 = \prod_{t=1}^{k} |r_t|^2 < \varepsilon$. In either case, we have 
$$
|\alpha_{k+1}|^2 \cdot \calF_{\calS}(\ket{\psi_{k+1}}) \leq \prod_{t=1}^{k} |r_t|^2 \cdot \Big(\Exp_{x \sim q_{\Psi_{k+1}}}[|\la \psi_{k+1} | W_x | \psi_{k+1} \ra|^2]\Big)^{1/6} \leq \varepsilon,
$$
where we have used Fact~\ref{fact:lower_bound_stabilizer_fidelity}. This proves the correctness of the theorem statement.

\textbf{Complexity of subroutines.} We now show that the overall cost of implementing $\LCU$, Hadamard test and estimating $\Exp_{x \sim q_\Psi}[|\la \psi | W_x | \psi \ra|^2]$ is $\poly(n,1/\varepsilon)$. With this and the complexity of $\Selfcorrection$ from Theorem~\ref{thm:self_correction}, the overall complexity of (error-free) iterative $\Selfcorrection$~follows. 
In the $t$th iteration we will use Corollary~\ref{cor:LCUfinal} in order to analyze the cost of $\LCU$ to prepare the residual state $\ket{\psi_{t+1}}$ and thereby the unitary $V_{t} = (U_\Psi - \sum_{j=1}^{t-1} \beta_j W_j)/\alpha_{t}$ where $W_j$ is a Clifford unitary preparing the stabilizer state $\ket{\phi_j}$. The probability of success of \LCU~is then
\begin{equation}\label{eq:prob_succ_LCU_iter}
\left(\frac{\norm{V_t \ket{0}^n}}{\frac{1}{\alpha_t} + \sum_{j=1}^{t-1} \frac{|\beta_j|}{\alpha_t}}\right)^2 = \left(\frac{\alpha_t}{1 + \sum_{j=1}^{t-1} |\beta_j|} \right)^2 \geq \frac{\varepsilon}{t^2} \geq \eta^4 \cdot \varepsilon,    
\end{equation}
where we  used $|\alpha_t|^2 \geq \varepsilon$ or we would not have proceeded to the $t$th iteration, $|\beta_i| \leq 1$ for all $i \in [t]$, and $t \leq k \leq 1/\eta^2$ from Claim~\ref{claim:ub_k_iterSC_errorfree}. 
Accounting for the probability of success of $\LCU$~in ~Eq.~\eqref{eq:prob_succ_LCU_iter} and summing over all iterations gives an overall query complexity of $O(1/(\eta^6 \cdot \varepsilon))$ to $\textsf{con}U_\Psi$ and time complexity of  $O(1/(\eta^6 \cdot \varepsilon))\cdot \poly(n)$. 
Using  $\eta=\poly(\varepsilon)$, the overall time complexity (including the cost of $\Selfcorrection$ at each step) is $\poly(n,1/\varepsilon)$.

\subsection{Errors in iterative self-correction}
\label{sec:errorSCsection}
So far we presented an iterative $\Selfcorrection$ algorithm assuming all the subroutines therein are error-free. In this section, we now take into account the errors that occur in the different subroutines and give guarantees on the overall algorithm. This will result in us proving the Theorem~\ref{thm:iterSC_gen} regarding the output of the iterative $\Selfcorrection$ procedure. The resulting algorithm is presented in Algorithm~\ref{alg:iterSC_main}. In order to give more intuition for the algorithm with errors, we structure this section as follows: first we give a high-level idea as to why the error-free algorithm that we described above cannot be \emph{trivially} ``robustified", next we give the new algorithm that can handle errors in the subroutines and describe the first two iterations like before and finally prove the correctness and analysis for the final algorithm.

\textbf{Why its not immediate to generalize the error-free case.} Recall that the  non-trivial aspect in the error-free case analysis was proving that Algorithm~\ref{alg:iterSC_error_free} terminates in $k\leq 1/\eta^2$ many steps. Now, consider an algorithm that incorporates the errors in the estimation steps. Suppose an algorithm with errors terminates in $k^*$ many steps. In Algorithm~\ref{alg:iterSC_error_free}, one of the errors arose from the computation of the $\beta_t$ (in particular $\langle \phi_t|\psi\rangle$). Naively, suppose one estimates these quantities  upto error $\varepsilon/k^*$ (with the  hope that one can use a triangle inequality over $k^*$ many iterations and get an eventual error of at most $\varepsilon$). The main issue here is we \emph{don't know what is $k^*$}.

Recall that in the analysis of the error-free case, we had to \emph{infer} an upper bound on $k^*$ based on the progress that one made in each step (i.e., Claim~\ref{claim:diff_l2_norm_residuals_iterSC_errorfree}). One could still try to replicate this analysis. Ofcourse, one could
 simply assume $k^*\leq O(1/\eta^2)$ like in the error-free case.  Now, if we incorporated errors in computation of $\beta_t$s and consequently $c_t,r_t$s in the proof of Claim~\ref{claim:diff_l2_norm_residuals_iterSC_errorfree}, we lose one property that we used crucially there, i.e., $\Psi_t-\Psi_{t+1}$ is orthogonal to $\Psi_{t+1}$ (in Eq.~\eqref{eq:psitandt+1orthogonal}).  With errors we can only say this quantity has inner product at most $t\varepsilon\cdot \eta^2$ (since need to do a triangle inequality over $(t-1)$ errors that accumulated until step $t$). With this one can only show Claim~\ref{claim:diff_l2_norm_residuals_iterSC_errorfree} is at least $\eta^2-t\varepsilon\eta^2$ and in the following claim, when we sum over $k^*$ iterations, we'd get $k^*\eta^2-(k^*)^2\varepsilon\eta^2$ but this quantity is \emph{always} at most $1$ (for constant $\varepsilon$), so one cannot show \emph{any upper bound} on the number of iterations $k^*$ that the algorithm runs for.

\textbf{Handling errors in the final algorithm.} Intuitively, the simple solution that we use to circumvent the issue mentioned above is, we don't assume \emph{any upper bound} on $k^*$ and instead in each round \emph{treat it as if it's the last round}. By this we mean that, if we are in round $t$, estimate all the quantities upto error $\varepsilon/t$, but if it so happens that the in the following round stopping conditions of $\gowers{\psi_t}{3}$ and $\prod_j|r_j|^2$ were not met (i.e, these quantities were large), we'd go to round $t+1$ and estimate \emph{all} the quantities $\beta_{1},\ldots,\beta_t$ once again with error $\varepsilon/(t+1)$. In particular, this means that in each iteration we'd create an ``error schedule", i.e., in $t$-th round ensure all the relevant quantities are estimated upto error $\varepsilon/(t+1)$.  With the above high-level idea we  show the following:
\begin{enumerate}[$(i)$]
    \item  Recomputing the quantities $\beta_j,c_j,r_j$, the deviation in the $t$-th and $(t')$-th iteration is small. 
    \item  One can show that the procedure above will eventually stop in $k\leq \poly(1/\eta)$ steps,
    \item  The sample and time complexity is polynomially worse than  the error-free case,
    \item Recall that above we only discussed the errors in estimating $\beta_i$s, but the similar idea can be used to show that the new algorithm above is also robust to all the errors that occur in Gowers-$3$ norm estimation, $\Selfcorrection$, Hadamard test and $\LCU$.
\end{enumerate}
Putting all these moving parts together is rather intensive calculations which we will present next. 

\subsubsection{Error-robust algorithm}
We give a brief intuition of the error-robust algorithm before we present the  algorithm itself. Without redescribing the algorithm, we compare it with the error-free $\Selfcorrection$~algorithm. Like in the error-free case, in the $t$th step, the algorithm checks if a $\varepsilon^6/2$-approximation of $\Exp_{x \sim q_{\Psi_t}}[|\la \psi_{t} | W_x | \psi_{t} \ra|^2]$ is at most $\varepsilon^6$ or $\delta$-approximation of $|\alpha_t|^2$ is at most $\varepsilon$. If so, the algorithm breaks. If neither of these conditions are met, we proceed with the $(t+1)$th iteration. At this point, we do the following update step, which will be the main difference between the error-free case and here. Recall that $\ket{\widehat{\psi}_t}=\sum_j\widetilde{\beta}_j^{(t)}\ket{\phi_i}$ and we had that $|\widetilde{\beta}_j^{(t)}-\beta_j|\leq \delta/t$). Now, the algorithm uses more copies of $\ket{\psi}$ and \emph{re-estimates} the values of $\beta_j$ (whose approximations are $\widetilde{\beta}_j$) upto error now $\delta/(t+1)$, i.e., $|\widetilde{\beta}_j^{(t+1)}-\beta_j|\leq \delta/(t+1)$). These new approximations are now referred to as  $\widetilde{\beta}_j^{(t+1)}$ and using these, we also recompute $\widetilde{r}_j^{(t+1)},\alpha_j^{(t+1)}$ and so on. In each step of these re-evaluations, we incorporate the errors coming from applications of $\LCU$, Hadamard test and $\Selfcorrection$.  We do this recomputation step in every iteration since it then allows us to show that the robust $\Selfcorrection$ procedure will terminate after $\poly(1/\eta)$ steps.  Proving that this recomputation step of all the $\beta_j,c_j,r_j,\alpha_j$ still solves the $\Selfcorrection$ task is technical and which we prove below. Let us first describe the different errors and how they will be circumvented in Algorithm~\ref{alg:iterSC_main} and the analysis.

\begin{myalgorithm} \setstretch{1.35}
\begin{algorithm}[H]
    \caption{Robust iterative $\Selfcorrection$}\label{alg:iterSC_main}
    \setlength{\baselineskip}{1.5em} 
    \DontPrintSemicolon 
    \KwInput{$\varepsilon, \in (0,1)$, copies of $n$-qubit state $\ket{\psi}$, access to $U_\psi$ and $\textsf{con}U_\psi$}
    \KwOutput{List of stabilizer states $L = \{\ket{\phi_i}\}_{i \in [k]}$, coefficients $B = \{\beta_i\}_{i \in [k]}$ for some $k \in \mathbb{N}$}
    
    Set $\eta = C_1 (\varepsilon^6/2)^{C_2}$ (with constants $C_1, C_2$ as defined in Theorem~\ref{thm:self_correction}) \\[0.5mm]
    Let $\ket{\psi_1}=\ket{\psi}$, $V_1 = U$, $r_0=1$. \\[0.5mm]
    Set $t_{\mathrm{max}} = \poly(1/\eta)$. Set $\delta = \eta^3/12$. \\
    Set $L=\varnothing$, $B=\varnothing$,  $\widetilde{r}_0=1$.\\[0.5mm]
    \For{$t = 1:t_{\mathrm{max}}$}{
        \lIf{$|\widetilde{\alpha}_t^{(t-1)}|^2 < \varepsilon$}{Let $\ket{\phi^\perp} = \ket{\psi_t}$ and break from loop} 
        Run $\LCU$ to prepare $\ket{\psi_{t}}=V_{t}\ket{0^n}$ where $V_{t} = (U - \sum_{j=1}^t \widetilde{\beta}_j^{(t-1)} W_j)/\widetilde{\alpha}_{t}^{(t-1)}$, for $t\geq 2$ and set $\ket{\psi_t} = \ket{\psi}$ for $t=1$. 
        \label{line_algo:residual_state_prep_iterSC_main}
        \\[0.5mm]
        Obtain $\varepsilon^6/2$-approximate estimate $\nu_t$ of $\Exp_{x \sim q_{\Psi_t}}[|\la \psi_{t} | W_x | \psi_{t} \ra|^2]$ given copies of $\ket{\psi_t}$ using Lemma~\ref{lem:est_gowers3_norm_states} \\[0.5mm]
        \lIf {$\nu_t < \varepsilon^6$}{Let $\ket{\phi^\perp} = \ket{\psi_t}$ and break from loop} 
        Run $\Selfcorrection$ using copies of $\ket{\psi_t}$ to learn $\ket{\phi_t} \in \textsf{Stab}$ s.t. $|\la \phi_t | \psi_t \ra|^2 \geq \eta$. \\[0.5mm]
        Obtain the Clifford unitary $W_t$ that prepares $\ket{\phi_t} = W_t \ket{0^n}$ using Lemma~\ref{lem:clifford_isotropic_subspace}.\\[0.5mm]
        \For{$j=1,\ldots,t$}{
        Estimate $\langle \phi_j|\psi \rangle$ up to error $\delta/(3 t^4)$ via the Hadamard test in Lemma~\ref{lem:hadamardtest} using $\textsf{con}W_j$ and $\textsf{con}U_\psi$. Call this estimate $\zeta_j^{(t)}$.\\[0.5mm]
        Compute $\la \phi_t | \phi_j \ra$  for all $j \in [t-1]$ classically (and exactly) using Lemma~\ref{lem:inner_prod_stab_states}. \\[0.5mm]
        Set $$\widetilde{\beta}_j^{(t)} = \zeta_j^{(t)} - \sum_{i=1}^{j-1} \widetilde{\beta}_i^{(t)} \la \phi_j | \phi_i \ra.$$ \\[0.5mm]
        Set $$\widetilde{c}_j^{(t)} = \widetilde{\beta}_j^{(t)}/\left(\prod_{i=0}^{j-1} \widetilde{r}_i^{(t)}\right), \quad   \widetilde{r}_j^{(t)}=\sqrt{1-|\widetilde{c}_j^{(t)}|^2},\quad \widetilde{\alpha}_{j+1}^{(t)} = \prod_{i=1}^j |\widetilde{r}_i^{(t)}|.$$ \\[0.5mm]
        }
    }
    \Return List of $k \leq O(1/\eta^2)$ stabilizer states $L=\{\ket{\phi_i}\}_{i \in [k]}$ along with their coefficients $B=\{\beta_\ell=c_\ell \prod_{t=1}^{\ell-1} r_t\}_{\ell\in [k]}$ held classically and a quantum state preparation of $\ket{\phi^\perp}$.
\end{algorithm}
\end{myalgorithm}

\subsubsection{Correctness} 
 For notational convenience, we will denote $\widetilde{\beta}_i^{(t)}$ to be the approximation of $\beta_i$ at the $t$th iteration (note that in each iteration, we will improve the error with which we estimate $\beta_i$).  
In more detail, in iteration $t$, we recompute all estimates $\la \phi_j | \psi \ra$ for all $j \in [t]$ up to error $\delta_t = \delta/(3t^4)$ and also re-approximate (with a smaller error) the value of $\beta_i$ which in the error-free case satisfies $\beta_i = c_i \prod_{j=1}^{i-1} r_i$ for $i \leq k$. We will accordingly update our stabilizer rank decomposition $\ket{\widehat{\psi}_j^{(t)}} = \sum_{i=1}^{j-1} \widetilde{\beta}_i^{(t)} \ket{\phi_i}$. The residual states at step $t$ can also then be defined using these new coefficients $\{\widetilde{\beta}_j^{(t)}\}_{j\in [t]}$ as
\begin{align}
\widetilde{\Psi}^{(t)}_j = \ket{\psi} - \sum_{i=1}^{j-1} \widetilde{\beta}_i^{(t)} \ket{\phi_i}, \quad \ket{ \widetilde{\psi}^{(t)}_j} = \frac{\widetilde{\Psi}^{(t)}_j}{\prod_{i=1}^{j-1} \widetilde{r}_i^{(t)}}=\frac{\widetilde{\Psi}^{(t)}_j}{\widetilde{\alpha}_{j}^{(t)}},
\end{align}
where $\widetilde{\Psi}_j^{(t)}$ (resp.~$\widetilde{\psi}_j$) is now the noisy version of $\Psi_j$ (resp.~$\psi_j$) which we had dealt with in the error-free case (Section~\ref{sec:errorfreeSelfcorrection}) as
$$
\Psi_j = \ket{\psi} - \sum_{i=1}^{j-1} \beta_i \ket{\phi_i}, \quad \psi_j = \frac{\Psi_j }{\prod_{i=1}^{j-1} r_i}.
$$
In the notation just introduced, the subscript is used to denote the state of interest from the specified iteration $j \leq t$ and the superscript $t$ is to denote the iteration number in which the re-evaluation is done. In particular, observe that in the $t$th iteration, we will work with the residual state $\widetilde{\Psi}_j^{(t)}$ for all $j\leq t$ by re-approximating these $\beta$s again (this time with error $\delta_t$). We should emphasize that these re-approximated residual states will not actually be reconstructed as part of Algorithm~\ref{alg:iterSC_main} but will be useful during the analysis. In line~\ref{line_algo:residual_state_prep_iterSC_main} of Algorithm~\ref{alg:iterSC_main}, we only construct the most recent residual state required $\ket{\widetilde{\psi}_t^{(t-1)}}$.
With this notation, observe that
\begin{align}
\label{eq:reformulationofctildet}
\widetilde{c}_{j}^{(t)}=\frac{\widetilde{\beta}_j^{(t)}}{\prod_{i=0}^{j-1}\widetilde{r}_i^{(t)}}=\frac{\zeta_j^{(t)}- \sum_{i=1}^{j-1} \widetilde{\beta}_i^{(t)} \langle \phi_j| \phi_i\rangle }{\prod_{i=0}^{j-1}\widetilde{r}_i^{(t)}},
\end{align}
where $\zeta_j^{(t)}$ is the $\delta_t$-approximate estimate of $\la \phi_j | \psi \ra$. We now have the following statement relating the noisy estimates $\widetilde{\beta}_j^{(t)}$, $\widetilde{c}_j^{(t)}$, and $\widetilde{r}_j^{(t)}$.
\begin{claim}\label{claim:c_r_expressions_beta}
Let $t\geq 1$ and $j\leq t$, then we have that
$$
\widetilde{c}_{j+1}^{(t)} = \frac{\widetilde{\beta}_{j+1}^{(t)}}{\sqrt{1 - \sum_{i=1}^j |\widetilde{\beta}_i^{(t)}|^2}}, \quad \left(\widetilde{r}_{j+1}^{(t)}\right)^2 = \frac{1 - \sum_{i=1}^{j+1} |\widetilde{\beta}_i^{(t)}|^2}{1 - \sum_{i=1}^j |\widetilde{\beta}_i^{(t)}|^2}
$$    
\end{claim}
\begin{proof}
Note that $\widetilde{c}^{(t)}_1 = \widetilde{\beta}_1^{(t)}$ and $\widetilde{r}_1^{(t)} = \sqrt{1 - |\widetilde{\beta}_1^{(t)}|^2}$. We will now show the above is true for $\widetilde{c}^{(t)}_2$ and $\widetilde{r}_2^{(t)}$. The general result will then follow from induction.  By the definitions of these quantities in the algorithm, we have that 
$$
\widetilde{c}^{(t)}_2 = \frac{\widetilde{\beta}_2^{(t)}}{\widetilde{r}_1^{(t)}} = \frac{\widetilde{\beta}_2^{(t)}}{\sqrt{1 - |\widetilde{\beta}_1^{(t)}|^2}}, \quad \left(\widetilde{r}^{(t)}_2\right)^2 = 1 - |\widetilde{c}^{(t)}_2|^2 = \frac{1 - |\widetilde{\beta}_1^{(t)}|^2  - |\widetilde{\beta}_2^{(t)}|^2}{1 - |\widetilde{\beta}_1^{(t)}|^2},
$$
which is the claim statement for $j=1$. Assuming the claim statement is true for up to $j = \ell-1$, we will now show this is true for $j=\ell$:
$$
\widetilde{c}^{(t)}_\ell = \frac{\widetilde{\beta}_\ell^{(t)}}{\prod_{i=1}^{\ell-1} \widetilde{r}_{i}^{(t)}} = \frac{\widetilde{\beta}_\ell^{(t)}}{\sqrt{\prod_{m=1}^{\ell-1} \frac{1 - \sum_{i=1}^{m} |\widetilde{\beta}_i^{(t)}|^2}{1 - \sum_{i=1}^{m-1} |\widetilde{\beta}_i^{(t)}|^2}}} = \frac{\widetilde{\beta}_\ell^{(t)}}{\sqrt{1 - \sum_{i=1}^{\ell-1}|\widetilde{\beta}_i^{(t)}|^2}},
$$
where we used that the denominator is a telescoping sum wherein only the first term and last term are not canceled. The result for $(\widetilde{r}_{\ell}^{(t)})^2 = 1 - |\widetilde{c}^{(t)}_\ell|^2$ then follows immediately. This concludes the  induction step.
\end{proof}

We will now comment on the guarantees of how close our estimates are due to the use of an error schedule. Recall from Section~\ref{sec:errorfreeSelfcorrection} (in particular Line~\ref{eq:defnofbetainalgo11} of the algorithm), in our error-free algorithm, $\beta_t$ were given by
\begin{align}
\label{eq:defnofbetainerrorfreecase}
\beta_t = \la \phi_t| \psi \ra - \sum_{j=1}^{t-1} \beta_j \la \phi_t | \phi_j \ra.
\end{align}
We now show how the error schedule ensures that $\widetilde{\beta}_j^{(t)}$ is close to $\beta_j$ across iterations.

\begin{lemma}[Properties of $\beta$s]
\label{lem:properties_beta}
If in Algorithm~\ref{alg:iterSC_main}, we proceed to iteration $t \geq 1$, the following properties are true about $\widetilde{\beta}s$ for all $j \leq t$
\begin{enumerate}[(a)]
    \item $|\widetilde{\beta}^{(t)}_j - \beta_j|\leq \delta/(3t^2)$
    \item $ \Big|\sum_{i=1}^{j-1}|{\beta}_i|^2 - \sum_{i=1}^{j-1}|\widetilde{\beta}_i^{(j-1)}|^2\Big|\leq \delta/t$
    \item $\sum_{i=1}^{j-1}|{\beta}_i|^2\leq 1-\eta +\delta/t $
\end{enumerate}
\end{lemma}
\begin{proof}
Let the current iteration be $t$ and $\delta_t = \delta/(3t^4)$. 

$(a)$ For all $j \in [t]$, denote the error in the estimate $\la \phi_j | \psi \ra$ as $\varepsilon_j \in \mathbb{C}$ and $|\varepsilon_j| \leq \delta_t$, then
\begin{align}
\label{eq:tildebetainequality1}
& \widetilde{\beta}^{(t)}_j = \zeta_j^{(t)}  - \sum_{i=1}^{j-1} \widetilde{\beta}^{(t)}_i \la \phi_j | \phi_i \ra \implies |\widetilde{\beta}^{(t)}_j - \beta_j| \leq \delta_t + \sum_{i=1}^{j-1} |\widetilde{\beta}^{(t)}_i - \beta_i|,
\end{align}
where $|\zeta_j^{(t)}-\langle \phi_j|\psi\rangle|\leq \varepsilon_j$ and in the implication, we used Eq.~\eqref{eq:defnofbetainerrorfreecase}  and $|\la \phi_j | \phi_i \ra| \leq 1$. Noting that $|\widetilde{\beta}^{(t)}_1 - \beta_1| = |\varepsilon_1| \leq \delta_t$ and $|\widetilde{\beta}^{(t)}_2 - \beta_2| \leq |\varepsilon_2| + |\widetilde{\beta}^{(t)}_1 - \beta_1| \leq 2\delta_t$. It can be shown by recursion that $|\widetilde{\beta}^{(t)}_i - \beta_i| \leq i \cdot \delta_t$. We thus have that
$$
|\widetilde{\beta}^{(t)}_j - \beta_j| \leq \delta_t + \sum_{i=1}^{j-1} i \cdot \delta_t \leq j^2 \delta_t \leq {\delta/3t^2}  \text{ for all } j \in [t]
$$
That concludes the proof of item $(a)$.

$(b)$ First note that $|\beta_j|\leq 1$ for all $j$ (see the description below Eq.~\eqref{eq:defnofPsit}), which implies that $|\widetilde{\beta}_j^{(t)}|\leq 1+\delta/(3t^2)$ for all $j\leq t$. Hence, we have that
\begin{align}
    \Big||\widetilde{\beta}_i^{(t-1)}|^2 - |{\beta}_i|^2 \Big|\leq  (2+\delta/(3t^2))\cdot   \Big||\widetilde{\beta}_i^{(t-1)}|- |{\beta}_i|\Big|&\leq (2+\delta/(3t^2))\cdot |\widetilde{\beta}_i^{(t-1)}- {\beta}_i| \\
    &\leq (2\delta/3+\delta^2/(9t^2))/(t-1)^2 \leq \delta / (t(t-1)),
\end{align}
where the second inequality used the reverse triangle inequality of $\big||a|-|b|\big|\leq |a-b|$ and the third inequality follows from item $(a)$ proved earlier. This implies
$$
\Big|\sum_{i=1}^{t-1}|\widetilde{\beta}_i^{(t-1)}|^2 - \sum_{i=1}^{t-1}|{\beta}_i|^2 \Big|=\Big|\sum_{i=1}^{t-1}\big(|\widetilde{\beta}_i^{(t-1)}|^2 - |{\beta}_i|^2\big) \Big|\leq \sum_{i=1}^{t-1}\Big||\widetilde{\beta}_i^{(t-1)}|^2 - |{\beta}_i|^2 \Big| \leq \delta /t,
$$
so we have that
\begin{align}
    \label{eq:betatildesumminusbetasum}
    \sum_{i=1}^{t-1}|{\beta}_i|^2 + \delta/t \geq  \sum_{i=1}^{t-1}|\widetilde{\beta}_i^{(t-1)}|^2\geq \sum_{i=1}^{t-1} |\beta_i|^2 -\delta/t   
\end{align}

$(c)$ As we have not stopped before reaching iteration $t$, we have that
\begin{align}
\label{eq:errorinbetatildeeta}
\eta \leq \prod_{i=1}^{t-1} \left(\widetilde{r}_i^{(t-1)}\right)^2 = 1 - \sum_{i=1}^{t-1} |\widetilde{\beta}_i^{(t-1)}|^2.
\end{align}
where we used Claim~\ref{claim:c_r_expressions_beta} for the equality.
Using item $(b)$ proved earlier, we have that
\begin{align}
    \label{eq:implicationofsumbetalarge}
    &\eta \leq \prod_{i=1}^{j-1} \left(\widetilde{r}_i^{(j-1)}\right)^2 = 1 - \sum_{i=1}^{j-1 } |\widetilde{\beta}_i^{(j-1)}|^2\leq 1-\sum_{i=1}^{j-1}|{\beta}_i|^2 +\delta/t \implies \sum_{i=1}^{j-1}|{\beta}_i|^2\leq 1-\eta +\delta/t.
\end{align}
This concludes the proof of the lemma.
\end{proof}

We have the following immediate corollary regarding the estimates of $\widetilde{r}_j^{(t)}$.

\begin{corollary}[Properties about $\widetilde{r}$s]
\label{corr:properties_r}
If in Algorithm~\ref{alg:iterSC_main}, we proceed to iteration $t \geq 1$, the following is true about $\widetilde{r}_j^{(t)}$ for all $j \leq t$
$$
\left|\prod_{i=1}^j \left(\widetilde{r}_i^{(t)}\right)^2 - \prod_{i=1}^j r_i^2 \right| \leq \delta/t
$$
\end{corollary}
\begin{proof}
The proof follows from Lemma~\ref{lem:properties_beta}$(b)$:
$$
\left|\prod_{i=1}^j \left(\widetilde{r}_i^{(t)}\right)^2 - \prod_{i=1}^j r_i^2 \right| = \left|1 - \sum_{i=1}^j |\widetilde{\beta}_i^{(t)}|^2 - 1 + \sum_{i=1}^j |\beta_i|^2 \right| = \left|\sum_{i=1}^j |\widetilde{\beta}_i^{(t)}|^2 - \sum_{i=1}^j |\beta_i|^2 \right| \leq \delta/t,
$$
proving the corollary statement.
\end{proof}

We now have the following statement regarding the noisy estimates of $\widetilde{c}_j^{(t)}$.
\begin{lemma} [Properties about $\widetilde{c}$]
\label{claim:approxtildec}If in Algorithm~\ref{alg:iterSC_main}, we proceed to iteration $t \geq 1$, the following is true regarding $\widetilde{c}_j^{(t)}$ for all $j \leq t$
$$
\left| |\widetilde{c}_j^{(t)}| - |c_j| \right|\leq \eta^{1.5}/(2t).
$$
\end{lemma}
\begin{proof}
Using Claim~\ref{claim:c_r_expressions_beta}, we can write
\begin{align}
\label{eq:ctildeminusc}
|\widetilde{c}_j^{(t)}| - |c_j| &=  \frac{|\widetilde{\beta}_j^{(t)}|}{\sqrt{1 - \sum_{i=1}^{j-1} |\widetilde{\beta}_i^{(t)}|^2}} - \frac{|{\beta}_j|}{\sqrt{1 - \sum_{i=1}^{j-1} |{\beta}_i|^2}}
\end{align}
Using Lemma~\ref{lem:properties_beta}(b), we can bound the denominator in the first term on the right hand side as
\begin{align*}
    \sqrt{1 - \sum_{i=1}^{j-1} |\widetilde{\beta}_i^{(t)}|^2} \geq \sqrt{1 - \sum_{i=1}^{j-1} |\beta_i|^2 - \delta/t} = \sqrt{1 - \sum_{i=1}^{j-1} |\beta_i|^2} \sqrt{1 - \frac{ \delta/t}{1 - \sum_{i=1}^{j-1} |\beta_i|^2}}
\end{align*}
Substituting the above equation into Eq.~\eqref{eq:ctildeminusc} gives us
\begin{align}
|\widetilde{c}_j^{(t)}| - |c_j| 
& \leq \frac{|\widetilde{\beta}_j^{(t)}|}{\sqrt{1 - \sum_{i=1}^{j-1} |\beta_i|^2}} \frac{1}{\sqrt{1 - \frac{ \delta/t}{1 - \sum_{i=1}^{j-1} |\beta_i|^2}}} - \frac{|{\beta}_j|}{\sqrt{1 - \sum_{i=1}^{j-1} |{\beta}_i|^2}} \\
& \leq \frac{|{\beta}_j| + \delta/t^2}{\sqrt{1 - \sum_{i=1}^{j-1} |\beta_i|^2}} \frac{1}{\sqrt{1 - \frac{ \delta/t}{1 - \sum_{i=1}^{j-1} |\beta_i|^2}}} - \frac{|{\beta}_j|}{\sqrt{1 - \sum_{i=1}^{j-1} |{\beta}_i|^2}} \\
& \leq \frac{|{\beta}_j| + \delta/t^2}{\sqrt{1 - \sum_{i=1}^{j-1} |\beta_i|^2}} \left({1 + \frac{ \delta/t}{1 - \sum_{i=1}^{j-1} |\beta_i|^2}}\right) - \frac{|{\beta}_j|}{\sqrt{1 - \sum_{i=1}^{j-1} |{\beta}_i|^2}} \\
& = \frac{\delta/t^2}{\sqrt{1 - \sum_{i=1}^{j-1} |\beta_i|^2}}+  \frac{\delta/t \cdot (|{\beta}_j| + \delta/t^2)}{({1 - \sum_{i=1}^{j-1} |\beta_i|^2})^{3/2}}\\
&\leq \frac{\delta/t^2}{\sqrt{(\eta-\delta/t)}}+  \frac{\delta/t \cdot (|{\beta}_j| + \delta/t^2)}{(\eta-\delta/t)^{3/2}}  \\
&\leq \frac{\eta^3/(12t^2)}{\sqrt{(\eta-\eta^3/(12t))}}+  \frac{2\eta^3/(12t) }{(\eta-\eta^3/(12t))^{3/2}}  \\
&= \frac{\eta^{2.5}/(12t^2)}{\sqrt{(1-\eta^2/(12t))}}+  \frac{2{\eta}^{1.5}/(12t) }{(1-\eta^2/(12t))^{3/2}}  \\
&\leq {\eta}^{1.5}/(2t),
\end{align}
where the first inequality used Lemma~\ref{lem:properties_beta}$(a)$. The third inequality used that $1/\sqrt{1-x} \leq (1+x)$ for $x \in (0,1/2]$ and we implicitly used that 
\begin{align}
\label{eq:validtaylor}
\frac{\delta/t}{1-\sum_i|\beta_i|^2}\leq \frac{\delta/t}{\eta-\delta/t}=\frac{\eta^3/(12t)}{\eta-\eta^3/(12t)}=\frac{\eta^2/(12t)}{1-\eta^2/(12t)}\leq 1/2,
\end{align}
where first inequality above used  Lemma~\ref{lem:properties_beta}$(c)$, first equality used $\delta=\eta^3/12$ (as fixed in the algorithm) and the final inequality used $\eta=\varepsilon^C\leq 2^{-C}$ by assumption that $\varepsilon\leq 1/2$). The fifth inequality again used Lemma~\ref{lem:properties_beta}$(c)$, sixth inequality used that $\delta=\eta^3/12$ and the final inequality used that $\eta \leq 1/2$ again. 

To complete the proof, we now give a lower bound on $|\widetilde{c}_j^{(t)}| - |c_j|$. We now use Lemma~\ref{lem:properties_beta}(b) to lower bound
\begin{align*}
    \sqrt{1 - \sum_{i=1}^{j-1} |\widetilde{\beta}_i^{(t)}|^2} \leq \sqrt{1 - \sum_{i=1}^{j-1} |\beta_i|^2 + \delta/t} = \sqrt{1 - \sum_{i=1}^{j-1} |\beta_i|^2} \sqrt{1 + \frac{ \delta/t}{1 - \sum_{i=1}^{j-1} |\beta_i|^2}}
\end{align*}

After substituting the above equation into Eq.~\eqref{eq:ctildeminusc}, we have that
\begin{align}
|\widetilde{c}_j^{(t)}| - |c_j| &\geq \frac{|\widetilde{\beta}_j^{(t)}|}{\sqrt{1 - \sum_{i=1}^{j-1} |\beta_i|^2}} \frac{1}{\sqrt{1 + \frac{ \delta/t}{1 - \sum_{i=1}^{j-1} |\beta_i|^2}}} - \frac{|{\beta}_j|}{\sqrt{1 - \sum_{i=1}^{j-1} |{\beta}_i|^2}} \\
& \geq \frac{|{\beta}_j| - \delta/t^2}{\sqrt{1 - \sum_{i=1}^{j-1} |\beta_i|^2}} \frac{1}{\sqrt{1 + \frac{ \delta/t}{1 - \sum_{i=1}^{j-1} |\beta_i|^2}}} - \frac{|{\beta}_j|}{\sqrt{1 - \sum_{i=1}^{j-1} |{\beta}_i|^2}} \\
& \geq \frac{|{\beta}_j| - \delta/t^2}{\sqrt{1 - \sum_{i=1}^{j-1} |\beta_i|^2}} \Big({1 - \frac{ \delta/(2t)}{1 - \sum_{i=1}^{j-1} |\beta_i|^2}}\Big) - \frac{|{\beta}_j|}{\sqrt{1 - \sum_{i=1}^{j-1} |{\beta}_i|^2}} \\
& = -\frac{\delta/t^2}{\sqrt{1 - \sum_{i=1}^{j-1} |\beta_i|^2}}-  \frac{\delta/(2t) \cdot (|{\beta}_j| - \delta/t^2)}{({1 - \sum_{i=1}^{j-1} |\beta_i|^2})^{3/2}}\\
& \geq -\frac{\delta/t}{\sqrt{\eta-\delta/t}}-  \frac{\delta/t }{({\eta-\delta/t})^{3/2}}\\
& = -\frac{\eta^3/(12t)}{\sqrt{\eta-\eta^3/(12t)}}-  \frac{\eta^3/(12t) }{({\eta-\eta^3/(12t)})^{3/2}}\\
& = -\frac{\eta^{2.5}/(12t)}{\sqrt{1-\eta^2/(12t)}}-  \frac{\eta^{1.5}/(12t) }{({1-\eta^2/(12t)})^{3/2}}\\
&\geq -\eta^{1.5}/(2t).
\end{align}
where we used Lemma~\ref{lem:properties_beta}$(a)$ in the second line, and $1/\sqrt{1+x}\geq 1-x/2$ for $x\leq 1/2$ (again using Eq.~\eqref{eq:validtaylor} to argue that $x\leq 1/2$ along with the fact that $\eta\leq 1/2$) in the third line. We used Lemma~\ref{lem:properties_beta}$(c)$ in the fifth line to argue 
$$
1 - \sum_{i=1}^{j-1} |\beta_i|^2 \geq \eta - \delta/t \implies -1/\sqrt{1 - \sum_{i=1}^{j-1} |\beta_i|^2} \geq -1/\sqrt{\eta-\delta/t},
$$
used $\delta=\eta^3/12$ in the sixth line, and that $\eta \leq 1/2$ again in the final line. 

We thus have shown
$$
-\eta^{1.5}/(2t)\leq |\widetilde{c}_j^{(t)}| - |c_j| \leq \eta^{1.5}/(2t) \implies  \left| |\widetilde{c}_j^{(t)}| - |c_j| \right|\leq \eta^{1.5}/(2t),
$$
proving the desired result.
\end{proof}

Using the above claim, we will show the promise on $\Selfcorrection$ in any iteration $j$ carries over to future iterations $t > j$. As we estimate $\Exp_{x \sim q_{\Psi_j}}[|\la \psi_{j} | W_x | \psi_{j} \ra|^2]$ for the residual state $\ket{\psi_j}$ (set to $\ket{\widetilde{\psi}_j^{(j-1)}}$) up to error $\varepsilon^6/2$ in iteration $j$, the true value satisfies
$$
\Exp_{x \sim q_{\Psi_j}}[|\la \psi_{j} | W_x | \psi_{j} \ra|^2] \geq \varepsilon^6/2.
$$
This in turn implies for the residual state $\ket{\widetilde{\psi}_j^{(j-1)}}$ that we would run with, we learn a stabilizer state $\ket{\phi_j}$ with fidelity $\eta = \Omega((\varepsilon^6/2)^C)$ for some constant $C>1$ as defined in Algorithm~\ref{alg:iterSC_main}. The promise is now formally stated below.

\begin{corollary}\label{claim:self_correction_promise}
Consider iteration $j \geq 1$. Suppose the algorithm continues to iteration $t > j$. If in iteration $j$, we obtained via $\Selfcorrection$ the stabilizer state $\ket{\phi_j}$ such that $|\la \phi_j | \widetilde{\psi}_j^{(j-1)} \ra|^2 \geq \eta$, then $|\widetilde{c}_j^{(t)}|^2 \geq \eta/2$ for all $t > j$.
\end{corollary}
\begin{proof}
First recall that in iteration $j$, we use $\Selfcorrection$ on the residual state $\ket{\widetilde{\psi}_j^{(j-1)}}$ to learn the stabilizer state $\ket{\phi_j}$ such that $|\la \phi_j | \widetilde{\psi}_j^{(j-1)} \ra|^2 \geq \eta$. Using Eq.~\eqref{eq:reformulationofctildet}, we can bound
\begin{align}
\left|\widetilde{c}_{j}^{(j)} - \la \phi_j | \widetilde{\psi}_j^{(j-1)}\ra \right| 
&= \left| \frac{\zeta_j^{(j)}- \sum_{i=1}^{j-1} \widetilde{\beta}_i^{(j)} \langle \phi_j| \phi_i\rangle - \la \phi_j | \widetilde{\Psi}_j^{(j-1)}\ra}{\widetilde{\alpha}_t^{(t-1)}} \right| \nonumber\\
&\leq \frac{1}{\widetilde{\alpha}_t^{(t-1)}}\left[ |\zeta_j^{(j)} - \la \phi_j | \psi \ra| + \sum_{i=1}^{j-1}|\widetilde{\beta}_i^{(j)} - \widetilde{\beta}_i^{(j-1)}| \right] \nonumber \\
&\leq \frac{1}{\widetilde{\alpha}_t^{(t-1)}}\left[ |\zeta_j^{(j)} - \la \phi_j | \psi \ra| + \sum_{i=1}^{j-1}|\widetilde{\beta}_i^{(j)} - \beta_i| + \sum_{i=1}^{j-1}|\beta_i - \widetilde{\beta}_i^{(j-1)}| \right] \nonumber \\
&\leq \frac{1}{\sqrt{\eta}}\left[ \delta/(3j^4) + 2\delta/(3(j-1)) \right] \nonumber \\
&\leq \eta^{2.5}/12,
\label{eq:interim_error_tildec}
\end{align}
where we used $\ket{\widetilde{\psi}_j^{(j-1)}} = \widetilde{\Psi}_j^{(j-1)}/\widetilde{\alpha}_j^{(j-1)}$ in the first line, used the expression of $\ket{\widetilde{\psi}_j^{(j-1)}}$ along with the triangle inequality in the second line, and noted that $\zeta_j^{(j)}$ is a $\delta/(3j^4)$-approximate estimate of $\la \phi_j | \psi \ra$ along with using Lemma~\ref{lem:properties_beta}$(a)$ in the fourth line and that $\widetilde{\alpha}_t^{(t-1)} \geq \sqrt{\eta}$ as we have not stopped. Finally, we noted the choice of $\delta = \eta^3/12$. We can now bound 
\begin{align}
\left| |\widetilde{c}_j^{(t)}|^2 - |\la \phi_j | \widetilde{\psi}_j^{(j-1)}\ra|^2 \right| \leq 2\left| |\widetilde{c}_j^{(t)}| - |c_j| \right| + 2\left| |\widetilde{c}_j^{(j)}| - |c_j| \right| + 2 \left| |\widetilde{c}_{j}^{(j)}| - |\la \phi_j | \widetilde{\psi}_j^{(j-1)}\ra| \right| \leq 13 \eta^{1.5}/6,
\end{align}
where we used Lemma~\ref{claim:approxtildec} for bounding the first two terms and Eq.~\eqref{eq:interim_error_tildec} for bounding the third term in the final inequality. This implies that 
$$
|\widetilde{c}_j^{(t)}|^2 \geq |\la \phi_j | \widetilde{\psi}_j^{(j-1)}\ra|^2 - 13 \eta^{1.5}/6 \geq \eta - 13 \eta^{1.5}/6 \geq \eta/2,
$$
where in the last inequality we used that $\varepsilon \in (0,1/2)$.
\end{proof}

We now prove the following claim which ensures that, assuming we are at iteration $t+1$, then the residual state after this round still has large $\ell_2$ norm. The reason we need to prove this is because we will eventually be running $\Selfcorrection$ on the residual state, so we require that the norm of this state is large enough for the output of $\Selfcorrection$ to to make sense.

\begin{claim}\label{claim:high_norm_residual_state}
Recall that $\widetilde{\Psi}_{t+1}^{(t)} = \ket{\psi} - \sum_{i=1}^t\widetilde{\beta}_i^{(t)}\ket{\phi_i}$ and $\ket{\widetilde{\psi}_{t+1}^{(t)}} = \widetilde{\Psi}_{t+1}^{(t)}/\widetilde{\alpha}_{t}^{(t)}$ be the residual state (possibly unnormalized) at the beginning of iteration $t+1$ after having checked the stopping criteria. Then, for each iteration $t \leq k-1$, $\norm{\psi_{t+1}^{(t)}}_2^2 \geq 7/8$.    
\end{claim}
\begin{proof}
Recall that in the error-free case $\ket{\psi_{t+1}}$ would have been the exact residual state at the end of iteration $t$ and satisfies $\norm{\ket{\psi_{t+1}}}_2 = 1$. We can bound the difference of the norm of the residual state obtained as part of the algorithm $\ket{\widetilde{\psi}_{t+1}^{(t)}}$ from the exact residual state $\ket{\psi_{t+1}}$ as
\begin{equation}
\label{eq:norm_diff_residual_states}
\left|\norm{\ket{\widetilde{\psi}_{t+1}^{(t)}}}_2 - \norm{\ket{\psi_{t+1}}}_2 \right| \leq \norm{\ket{\widetilde{\psi}_{t+1}^{(t)}} - \ket{\psi_{t+1}}}_2 = \norm{ \sum_{i=1}^t \left( \frac{\widetilde{\beta}_i^{(t)}}{\widetilde{\alpha}_{t}^{(t)}} - \frac{\beta_i}{\alpha}\right) \ket{\phi_i}}_2 \leq \sum_{j=1}^t \left| \frac{\widetilde{\beta}_j^{(t)}}{\widetilde{\alpha}_t^{(t)}} - \frac{\beta_j}{\alpha} \right|    
\end{equation}
where we have denoted $\alpha = \prod_{j=1}^t r_j$, $\widetilde{\alpha}_t^{(t)} = \prod_{j=1}^t \widetilde{r}_j^{(t)}$, we used the reverse triangle inequality of $\left| \norm{a} - \norm{b} \right| \leq \norm{a - b}$ in the first inequality and the triangle inequality of $\norm{\cdot}_2$ in the last inequality. Now, observe that for all $i \leq t$:
\begin{align}\label{eq:interim_error_coeffs_phi}
    \left|\frac{\widetilde{\beta}_i^{(t)}}{\widetilde{\alpha}_t^{(t)}} - \frac{\beta_i}{\alpha} \right| = \left|\frac{\beta_i + b_i}{\alpha + a} - \frac{\beta_i}{\alpha} \right| = \left|\frac{\alpha \beta_i + \alpha b_i - \alpha \beta_i - a \beta_i}{\alpha \widetilde{\alpha}_t^{(t)} } \right| \leq \frac{|b_i|}{|\widetilde{\alpha}_t^{(t)}|} + \frac{|a|}{|\alpha| \cdot |\widetilde{\alpha}_t^{(t)}|}
\end{align}
where we have let the error $\widetilde{\beta}_i^{(t)} - \beta_i = b_i$ and  and $\widetilde{\alpha}_t^{(t)} - \alpha = a$ for some $a,b_i \in \mathbb{C}$ which we will comment on shortly. Using Claim~\ref{lem:properties_beta}, we have that $|b_i| \leq \delta/(3t^2)$ and using Corollary~\ref{corr:properties_r}, we have that $|a| \leq \delta/t$. Moreover, as we have proceeded to iteration $t+1$ without stopping, we have that $\widetilde{\alpha}^{(t)} \geq \sqrt{\eta}$. Corollary~\ref{corr:properties_r} implies that $\alpha \geq \widetilde{\alpha}_t^{(t)} - a \geq \sqrt{\eta} - \delta/t \geq \sqrt{\eta}/2$ (since $\delta=O(\eta^3)$). Substituting in Eq.~\eqref{eq:interim_error_coeffs_phi} gives us
\begin{equation}
    \left|\frac{\widetilde{\beta}_i^{(t)}}{\widetilde{\alpha}_t^{(t)}} - \frac{\beta_i}{\alpha} \right| \leq \frac{\delta}{t^2 \sqrt{\eta}} + \frac{2 \delta }{t \eta} \leq \frac{\eta^2}{8 t},
\end{equation}
for the choice of $\delta=\eta^3/24$. Substituting the above result in the right hand side of Eq.~\eqref{eq:norm_diff_residual_states} gives~us
\begin{equation}
    \left|\norm{\tilde{\psi}_{t+1}^{(t)}}_2 - \norm{\psi_{t+1}}_2 \right| \leq \frac{1}{8} \eta^2 \implies \norm{\widetilde{\psi}_{t+1}^{(t)}}_2 \geq 1 - \frac{1}{8} \eta^2 \geq 7/8,
\end{equation}
proving the claim statement.
\end{proof}

We are now ready to argue about the progress on makes in iteration $t$ between the norms of the $\widetilde{\psi}_j^{(t)}$ across different $j$s.

\begin{claim}\label{claim:diff_l2_norm_residuals_iterSC}
Define $$\widetilde{\psi}^{(t)}_j = \ket{\psi} - \sum_{i=1}^{j-1} \widetilde{\beta}_i^{(t)} \ket{\phi_i}$$ and define $\psi_j = \ket{\psi} - \sum_{i=1}^{j-1} \beta_i \ket{\phi_i}$. In iteration $t$ for all $j \leq t$, we have that
\begin{align}
\norm{\widetilde{\psi}^{(t)}_j}_2^2 - \norm{\widetilde{\psi}^{(t)}_{j+1}}_2^2 \geq \eta^2/9.
\end{align}    
\end{claim}
\begin{proof}
Firstly, we note that 
\begin{align}
    \norm{\widetilde{\psi}^{(t)}_j}_2^2 - \norm{\widetilde{\psi}^{(t)}_{j+1}}_2^2 &= \norm{\widetilde{\psi}^{(t)}_j - \widetilde{\psi}^{(t)}_{j+1} + \widetilde{\psi}^{(t)}_{j+1}}_2^2 - \norm{\widetilde{\psi}^{(t)}_{j+1}}_2^2  \\
    &= \norm{\widetilde{\psi}^{(t)}_j - \widetilde{\psi}_{j+1}^{(t)}}_2^2 + 2 \textsf{Re}(\langle \widetilde{\psi}^{(t)}_j - \widetilde{\psi}^{(t)}_{j+1} | \widetilde{\psi}^{(t)}_{j+1} \rangle) + \norm{\widetilde{\psi}^{(t)}_{j+1}}_2^2 - \norm{\widetilde{\psi}^{(t)}_{j+1}}_2^2 \\
    &\geq \norm{\widetilde{\psi}_j^{(t)} - \widetilde{\psi}^{(t)}_{j+1}}_2^2 - 2 |\langle \widetilde{\psi}^{(t)}_j - \widetilde{\psi}^{(t)}_{j+1} | \widetilde{\psi}^{(t)}_{j+1} \ra|,
    \label{eq:L2_norm_diff_states}
\end{align}
where the second equality follows from $\|a+b\|_2^2=\|a\|_2^2+\|b\|_2^2+2\textsf{Re}(\langle a|b\rangle)$. Recall that in the ideal case when all the estimates of $\beta_j$ for all $j \in [k]$ match the true values, we observed in Claim~\ref{claim:diff_l2_norm_residuals_iterSC_errorfree} that $\psi_j^{(t)} - \psi_{j+1}^{(t)}$ is orthogonal to $\psi_{j+1}$ which is no longer true in the noisy case. To this end, observe~that for every $j\in [t]$, we have
\begin{align}
    \left|\Big\langle \widetilde{\psi}_j^{(t)}-\widetilde{\psi}_{j+1}^{(t)}| \widetilde{\psi}_{j+1}^{(t)}\Big\rangle \right| = \left|{\widetilde{\beta}_j^{(t)}} \la \phi_j | \widetilde{\psi}_{j+1}^{(t)} \ra \right| &\leq \left| \la \phi_j |\big( |\psi \ra - \sum_{i=1}^j \widetilde{\beta}_i^{(t)} \ket{ \phi_i}\big)\right| \\
    &=\left| \la \phi_j |\big( |\psi_{j+1} \ra +\sum_{i=1}^j\beta_i\ket{\phi_i} - \sum_{i=1}^j \widetilde{\beta}_i^{(t)} \ket{ \phi_i}\big)\right| \\
    &=\left| \sum_{i=1}^j\beta_i\la \phi_j\ket{\phi_i} - \sum_{i=1}^j \widetilde{\beta}_i^{(t)} \la \phi_j\ket{  \phi_i}\big)\right| \\
    &\leq \sum_{i=1}^j \left|\widetilde{\beta}_i^{(t)} - \beta_i \right| \\
    & \leq \frac{j \delta}{t^2} \leq \frac{2 \delta}{t},
\end{align}
where in the second inequality  we used $\widetilde{\psi}^{(t)}_{j+1} = \ket{\psi} - \sum_{i=1}^{j} \widetilde{\beta}_i^{(t)} \ket{\phi_i}$, $|{\widetilde{\beta}_j^{(t)}}|\leq 1$ and the second equality used that ${\psi}_{j+1} = \ket{\psi} - \sum_{i=1}^{j} {\beta}_i\ket{\phi_i}$, fourth equality that $\ket{\phi_j}$ and $\ket{\psi_{j+1}}$ are orthogonal (see the discussion around Eq.~\eqref{eq:psit+1})  and in the second-to-last inequality we used Lemma~\ref{lem:properties_beta}$(a)$ to comment on the error in $\beta_i$ estimate. Substituting the above into Eq.~\eqref{eq:L2_norm_diff_states} gives us
\begin{align}
\norm{\widetilde{\psi}_j^{(t)}}_2^2 - \norm{\widetilde{\psi}_{j+1}^{(t)}}_2^2 &\geq \norm{\widetilde{\psi}_j^{(t)} - \widetilde{\psi}_{j+1}^{(t)}}_2^2 - 2 |\langle \widetilde{\psi}_j^{(t)} - \widetilde{\psi}_{j+1}^{(t)} | \widetilde{\psi}_{j+1}^{(t)} \ra| \\
&\geq \norm{\widetilde{\beta}_{j}^{(t)}\ket{\phi_j}}_2^2 - \frac{2\delta}{t} \\
&= |\widetilde{c}_{j}^{(t)}|^2 \prod_{i=1}^{j-1} |\widetilde{r}_i^{(t)}|^2 - \frac{2\delta}{t} \\
&\geq \frac{\eta^2}{3} - \frac{2\delta \eta}{3t} - \frac{2\delta}{t} \\
& \geq \frac{\eta^2}{9},
\end{align}
where we used $|\widetilde{c}_j^{(t)}|^2 \geq \eta/2 \norm{\widetilde{\psi}_j^{(j-1)}} > \eta/3$ by the promise of self-correction~(Claim~\ref{claim:self_correction_promise}) and that the norm of the residual state is high~(Claim~\ref{claim:high_norm_residual_state}), and 
$$\prod_{i=1}^{j-1} |\widetilde{r}_i^{(t)}|^2 \geq \prod_{i=1}^{j-1} |r_i|^2 - \delta/t \geq \prod_{i=1}^{j-1} |\widetilde{r}_i^{(j-1)}|^2 - 2\delta/t \geq \eta - 2 \delta/t
$$
follows from the fact that we did not stop before proceeding through iteration $t$ which in particular includes iteration $j$ and using Corollary~\ref{corr:properties_r}. Finally, we conclude with the choice of $\delta = \eta^3/12$.
\end{proof}

With the progress bound in the previous claim, we now give an upper bound on the number of iterations the protocol needs to run before stopping.

\begin{claim}\label{claim:k_ub_iterSC†}
For $\delta=\eta^3/12$, we have $k \leq 9/\eta^2$.    
\end{claim}
\begin{proof}
Now, suppose we stop our algorithm after $k$ iterations. Then, in iteration $k$, using Claim~\ref{claim:diff_l2_norm_residuals_iterSC} and adding over $j=1,\ldots,k$ gives us $\norm{\widetilde{\psi}_1^{(k)}}_2^2 - \norm{\widetilde{\psi}_{k+1}^{(k)}}_2^2 \geq k \eta^2/9$.
Furthermore we have that
$$
\norm{\widetilde{\psi}_1^{(k)}}_2^2 - \norm{\widetilde{\psi}_{k+1}^{(k)}}_2^2 =\|\psi_1\|_2^2-\prod_{j=1}^{k} |\widetilde{r}^{(k)}_j|^2 \|\widetilde{\psi}^{(k)}_{k+1}\|_2^2 \,\, \leq 1,
$$
where we used that $\ket{\widetilde{\psi}^{(k)}_1}=\ket{\psi_1}=\ket{\psi}$ and $\norm{\ket{\psi}}=1$. Hence, we have that
$$
k \cdot \eta^2/9 \leq 1 \implies k \leq 9/\eta^2.
$$
\end{proof} 

\paragraph{Correctness of output.}
The correctness of the algorithm is immediate. In each step, Algorithm~\ref{alg:iterSC_main} produces a stabilizer state (hence $\ket{\phi_1}.\ldots,\ket{\phi_k}$ are stabilizer states) and at the final step we stop because either $\varepsilon^6/2$-estimate of $\Exp_{x \sim q_{\Psi_{k+1}}}[|\la \psi_{k+1} | W_x | \psi_{k+1} \ra|^2]$ is $< \varepsilon^6$ \emph{or} $\delta_k$-estimate of $|\alpha_{k+1}|^2 = \prod_{t=1}^{k} |r_t|^2$ is $< \varepsilon$. In either case, we have 
$$
|\alpha_{k+1}|^2 \cdot \calF_{\calS}(\ket{\psi_{k+1}}) \leq \prod_{t=1}^{k} |r_t|^2 \cdot \Big(\Exp_{x \sim q_{\Psi_{k+1}}}[|\la \psi_{k+1} | W_x | \psi_{k+1} \ra|^2]\Big)^{1/6} < 2 \varepsilon,
$$
where we have used Fact~\ref{fact:lower_bound_stabilizer_fidelity}. Thus, running with $\varepsilon/2$ everywhere would prove the correctness of the theorem statement.

\paragraph{Success probability.}
Let us now comment on the success probability of preparing the state $\ket{\widetilde{\psi}_t^{(t-1)}}$ on which we perform self-correction in iteration $t$. The state $\ket{\widetilde{\psi}_t^{(t-1)}}$ can be 
expressed as
$$
\ket{\widetilde{\psi}_t^{(t-1)}} = \frac{\ket{\psi} - \sum_{j=1}^{t-1} \widetilde{\beta}_j^{(t-1)} \ket{\phi_j}}{\widetilde{\alpha}_t^{(t-1)}},
$$
where $\widetilde{\alpha}_t^{(t-1)} = \prod_{j=1}^{t-1} \widetilde{r}_j^{(t-1)}$. Suppose $V_t$ is the corresponding state preparation unitary i.e., $V_t \ket{0}^n = \ket{\widetilde{\psi}_t^{(t-1)}}$ which we build via $\LCU$. In particular, $V_t$ takes the following form
$$
V_t = \frac{1}{\widetilde{\alpha}_t^{(t-1)}} U_{\Psi} - \sum_{j=1}^{t-1} \frac{\widetilde{\beta}_j^{(t-1)}}{\widetilde{\alpha}_t^{(t-1)}} W_j,
$$
where $W_j$ is the state preparation unitary of $\ket{\phi_j}$ and which we can be constructed using Lemma~\ref{lem:clifford_isotropic_subspace} in $O(n^2)$ time. The success probability of preparing $\ket{\widetilde{\psi}_t^{(t-1)}}$ via $\LCU$ using Corollary~\ref{cor:LCUfinal} is then
\begin{equation}\label{eq:prob_succ_LCU_iter_error}
\left( \frac{\|V_t \ket{0}\|_2}{\frac{1}{\widetilde{\alpha}_t^{(t-1)}}+\sum_{j=1}^{t-1}\frac{|\widetilde{\beta}_j^{(t-1)}|}{\widetilde{\alpha}_t^{(t-1)}}}\right)^2= \left( \frac{\widetilde{\alpha}_t^{(t-1)} \|\widetilde{\psi}_t^{(t-1)}\|_2}{1 + \sum_{j=1}^{t-1} |\widetilde{\beta}_j^{(t-1)}| }\right)^2 \geq \frac{\eta}{4t^2} \geq \frac{\eta^5}{324},    
\end{equation}
where we have used that $\|\widetilde{\psi}_t^{(t-1)}\|_2^2\geq 1/2$ for all $t$ from  Claim~\ref{claim:high_norm_residual_state} and that $\widetilde{\alpha}_t^{(t-1)}\geq \sqrt{\eta}$ from Corollary~\ref{corr:properties_r} since we have run the algorithm for $t$ rounds (so it must have been that this condition is satisfied for all $\widetilde{\alpha}_t^{(j)}$ for $j\leq t$). The last inequality follows from the upper bound on $\kappa$ from Claim~\ref{claim:k_ub_iterSC†}.

\subsubsection{Complexity}
Recall that in iteration $t$, we will recompute all the $\widetilde{\beta}$s from before as well as the current one
$$
\widetilde{\beta}_j^{(t)} = \zeta_j^{(t)} - \sum_{i=1}^{j-1} \widetilde{\beta}_i^{(t)} \la \phi_t | \phi_i \ra, 
$$
which in turn requires $\delta_t=\delta/t^4$-approximate estimates of $\langle \phi_i|\psi\rangle$. So in total that uses $O(t/\delta_t^2)$ of $U_\psi,\textsf{con}U_\psi$ and the Clifford circuits $W_i$ (which produced the stabilizer states $\ket{\phi_i}$). So in total across all iterations the overall cost is
$$
\sum_{t=1}^k t/\delta_t^2=\sum_{t=1}^k t^9/\delta^2\leq k^{10}/\delta^2\leq O(1/\eta^{26})=\poly(1/\varepsilon),
$$
where we used that $k\leq O(1/\eta^2)$ and $\delta=O(\eta^3)$. So this is the cost of getting all the $\widetilde{\beta}$s throughout the protocol.  We also need to consider the contribution to time complexity for preparation of the residual states $\ket{\psi_t}$ for all $t \in [\kappa]$. Accounting for the probability of success of $\LCU$~in ~Eq.~\eqref{eq:prob_succ_LCU_iter_error} and summing over all $O(1/\eta^2)$ iterations gives an overall query complexity of $O(1/\eta^7)$ to $\textsf{con}U_\Psi$ and time complexity of $O(1/\eta^7)\cdot \poly(n)$. 
Using  $\eta=\poly(\varepsilon)$, the overall time complexity (including the cost of $\Selfcorrection$ at each step) is $\poly(n,1/\varepsilon)$.

Putting together all the costs throughout all the iterations, the overall  complexity (which includes uses of the $U_\psi,\textsf{con}U_\psi$, gate complexity) is $\poly(n,1/\varepsilon,\log (1/\upsilon))$ proving our main theorem. 
\subsection{Iterative Stabilizer Bootstrapping}
\label{sec:boostrappingSC}
Depending on the choice of the algorithm $\calA$ in Theorem~\ref{thm:iterSC_gen}, the following are then true.
\begin{restatable}{theorem}{iterSCbootstrap}
\label{thm:iterSC_stab_bootstrap}(Iterative Stabilizer Bootstrapping)
Let $\varepsilon,\upsilon\in (0,1)$, $\eta(\varepsilon)=\varepsilon/2$. Let $\ket{\psi}$ be an unknown $n$-qubit quantum state. There is an algorithm that with probability $\geq 1-\upsilon$, satisfies the following:
Given access to $U_\psi,\textsf{con}U_\psi$, outputs 
$\beta \in \calB_\infty^k$ and stabilizer states $\{\ket{\phi_i}\}_{i\in [k]}$ such that one can write $\ket{\psi}$ as
$$
\ket{\psi}=\sum_{i\in [k]} \beta_i \ket{\phi_i}+\beta_{k+1}\ket{\phi^\perp},
$$
where the residual  state $\ket{\phi^\perp}$ satisfies $|\beta_{k+1}|^2\cdot \stabfidelity{\ket{\phi^\perp}} < \varepsilon$ and $k \leq O(1/\varepsilon^2)$. This algorithm uses the \textsf{Stabilizer Bootstrapping} algorithm (Theorem~\ref{thm:sitanbootstrapping}) $k$ times and the total runtime is 
$$
\poly\big(n,(1/\varepsilon)^{\log(1/\varepsilon)},\log(1/\upsilon)\big).
$$
\end{restatable}
In the previous sections, we used the $\Selfcorrection$ protocol iteratively to learn a structured decomposition of $\ket{\psi}$ leading to the proof. However, instead of using the $\Selfcorrection$ protocol as our base algorithm, we could have used the agnostic learner of stabilizer states in Theorem~\ref{thm:sitanbootstrapping} based on stabilizer bootstrapping~\cite{chen2024stabilizer}. The promise in each iteration is then $\eta(\varepsilon) = \varepsilon/2$ (where we set the error corresponding to Theorem~\ref{thm:sitanbootstrapping} to $\varepsilon/2$). Note that this promise is better than that of $\Selfcorrection$ which was $\varepsilon^C$ (for $C > 1$) but comes at the cost of a quasi-polynomial time complexity of $\poly(n) (1/\varepsilon)^{O(1/\varepsilon)}$. It can then be shown that all the statements made in Section~\ref{sec:errorSCsection} go through with this modified $\eta$ and leads to the proof of Theorem~\ref{thm:iterSC_stab_bootstrap}.

\section{Applications}
\label{sec:applications}
We have so far discussed how to learn a structured decomposition $\ket{\widetilde{\psi}}$ of any $n$-qubit quantum state $\ket{\psi}$ (Section~\ref{sec:iteratedselfcorrection}). There are two natural questions, what is the utility of such a decomposition and secondly, what if $\ket{\psi}$ itself was \emph{structured}, could we say anything about the properties of $\ket{\widetilde{\psi}}$? We answer both these questions below by giving a few applications of our main result.
\begin{enumerate}
    \item We first show the following: suppose we want to compute the inner product of an arbitrary state $\ket{\psi}$ with $\ket{\phi}$ where $\ket{\phi}$ has low stabilizer extent, then one can approximate this inner product by  using the ``mimicking state" $\ket{\widetilde{\psi}}$ in place of $\ket{\psi}$. 
    \item We give an algorithm to learn states with low stabilizer extent (up to constant trace distance).
\end{enumerate}
We state these results explicitly in the subsections below and give their proofs. 

\subsection{Mimicking state for  estimating stabilizer-extent fidelities}
\label{sec:mimickingfidelity}
We essentially show that for all states $\ket{\phi}$ with low stabilizer extent, if $\ket{\widetilde{\psi}}$ is the output of iterative $\Selfcorrection$ on input $\ket{\psi}$, then $\langle \phi|\psi\rangle$ is close to $\langle \phi |\widetilde{\psi}\rangle$ and also the stabilizer fidelities of $\ket{\psi}$ and $\ket{\widetilde{\psi}}$ are close. This is formalized in the lemma below.
\begin{lemma}
\label{cor:low_stab_extent_inner_prod}
Let $\xi > 0, \varepsilon' \in (0,1]$, $\eta=(\varepsilon'/\xi)^{2C}$, and $\calC(\xi)$ be the set of $n$-qubit states with stabilizer extent at most $\xi$. Assume the algorithmic $\PFR$ conjecture. For every $n$-qubit  state $\ket{\psi}$, we can find a~stabilizer-rank $(1/\eta^2)$ state~$\ket{\widetilde{\psi}}$ (up to a normalization) using $\poly(n,\xi,1/\varepsilon')$ queries to $U_\psi,\textsf{con}U_\psi$, copies of $\ket{\psi}$ and in $\poly(n,k,1/\varepsilon')$ time such~that
\begin{itemize}
    \item[$(i)$] $\left|\langle \phi|\psi\rangle - \langle \phi|\widetilde{\psi}\ra\right| \leq \varepsilon'$ for every  $\ket{\phi} \in \calC(\xi)$, 
    \item[$(ii)$] $\left|\calF_{\calC(\xi)}(\ket{\widetilde{\psi}}) - \calF_{\calC(\xi)}(\ket{\psi}) \right| \leq 3 \varepsilon',$
\end{itemize}
where $\calF_{\calC(\xi)}(\ket{\psi}) = \max_{\ket{\phi} \in \calC(\xi)} |\la \phi | \widetilde{\psi} \ra|^2$.
\end{lemma}
\begin{proof}
We will run iterative-$\Selfcorrection$  (Theorem~\ref{thm:iterSC_gen}) with $\varepsilon$ therein instantiated as $\varepsilon = (\varepsilon')^2/\xi^2$ and the corresponding value of $\eta = \varepsilon^C = (\varepsilon'/\xi)^{2C}$ as stated in the theorem statement.
The output is a stabilizer-rank $k\leq 1/\eta^2=\poly(\xi/\varepsilon')$ state $\ket{\widetilde{\psi}}$ (up to normalization)  such that
$$
\ket{\psi} = \ket{\widetilde{\psi}} + \beta_{k+1} \ket{\phi^\perp}
$$
and where $|\beta_{k+1}|^2 \cdot \calF_\calS(\ket{\phi^\perp}) \leq \varepsilon$. Now, for any stabilizer state $\ket{s} \in \textsf{Stab}$, we have
\begin{align}
\label{eq:interim_stab_fidelity}
    \la s | \psi \ra &= \la s | \widetilde{\psi} \ra + \beta_{k+1} \la s | \phi^\perp \ra \implies 
    \left|\la s | \psi \ra-\la s | \widetilde{\psi} \ra  \right| =  \left| \beta_{k+1} \right| \left| \la s | \phi^\perp \ra \right|\leq \sqrt{\varepsilon}=\varepsilon'/\xi, \end{align}
since $|\beta_{k+1}|^2 \cdot \stabfidelity{\ket{\phi^\perp}}=|\langle s|\phi^\perp\rangle|^2\leq \varepsilon$.  
Consider $\ket{\phi} \in \calC(\xi)$ expressed as
$$
\ket{\phi} = \sum_{i=1}^m c_i \ket{s_i},
$$
for some $m \in \mathbb{N}$, stabilizer states $\{\ket{\phi_i}\}_{i \in [m]}$, and coefficients $c_i \in \mathbb{C}$ for all $i \in [m]$ such that $\sum_{i=1}^m |c_i| \leq \xi$. We then observe
\begin{align}\label{eq:item1_stab_extent_ip_relations}
    \left| \la \psi | \phi \ra - \la \widetilde{\psi} | \phi \ra \right| = \left| \sum_{i=1}^m c_i \left( \la \psi | s_i \ra - \la \widetilde{\psi} | s_i \ra \right)  \right| \leq \sum_{i=1}^m |c_i| \cdot \left| \la \psi | s_i \ra - \la \widetilde{\psi} | s_i \ra \right| \leq \sum_{i=1}^m |c_i| \cdot \varepsilon'/\xi \leq \varepsilon',
\end{align}
where we used the triangle inequality in the second inequality, Eq.~\eqref{eq:interim_stab_fidelity} in the third inequality and that $\sum_{i=1}^m |c_i| \leq \xi$ in the fourth inequality. This proves item $(i)$ of the corollary.

We now prove item $(ii)$. Suppose $\ket{\phi_1} \in \calC(\xi)$ maximizes the fidelity with $\ket{\psi}$ among all states from $\calC(\xi)$ and $\ket{\phi_2} \in \calC(\xi)$ maximizes the fidelity with $\ket{\widetilde{\psi}}$ i.e.,
\begin{align}
    \ket{\phi_1} = \argmax_{\ket{\phi} \in \calC(\xi)} |\la \phi | \psi \ra|^2, \quad \ket{\phi_2} = \argmax_{\ket{\phi} \in \calC(\xi)} |\la \phi | \widetilde{\psi} \ra|^2.
\end{align}
In other words, $|\langle \phi_1 |\psi\rangle|=\sqrt{\calF_{\calC(\xi)}{\ket{\psi}}}$ and $|\langle \phi_2 |\widetilde{\psi}\rangle|=\sqrt{\calF_{\calC(\xi)}{\ket{\widetilde{\psi}}}}$. We then have that
\begin{equation}\label{eq:interim_phi1}
\sqrt{\calF_{\calC(\xi)}{\ket{\psi}}}-|\langle \phi_1|\widetilde{\psi}\rangle|= |\langle \phi_1 |\psi\rangle|-|\langle \phi_1 |\widetilde{\psi}\rangle|\leq \Big||\langle \phi_1|\psi\rangle|-|\langle \phi_1 |\widetilde{\psi}\rangle|\Big|\leq  \left|\la \phi_1 | \psi \ra-\la \phi_1 | \widetilde{\psi} \ra  \right|\leq \varepsilon', 
\end{equation}
where the second inequality used reverse triangle inequality and the last inequality used item $(i)$ (Eq.~\eqref{eq:item1_stab_extent_ip_relations}). Now, note that 
$$
|\langle \phi_1|\widetilde{\psi}\rangle|\leq |\langle \phi_2 |\widetilde{\psi}\rangle|=\sqrt{\calF_{\calC(\xi)}{\ket{\widetilde{\psi}}}},
$$
by definition of $\ket{\phi_2}$, which implies 
$$
\sqrt{\calF_{\calC(\xi)}{\ket{\psi}}}-\sqrt{\calF_{\calC(\xi)}{\ket{\widetilde{\psi}}}}\leq \sqrt{\calF_{\calC(\xi)}{\ket{\psi}}}-|\langle \phi_1|\widetilde{\psi}\rangle|\leq \varepsilon'.
$$
Similarly one can show $\sqrt{\calF_{\calC(\xi)}{\ket{\widetilde{\psi}}}}-\sqrt{\calF_{\calC(\xi)}{\ket{\psi}}}\leq \varepsilon'$ by starting with $\ket{\phi_2}$ in Eq.~\eqref{eq:interim_phi1}. Thus, we~have
$$
\left|\sqrt{\calF_{\calC(\xi)}{\ket{\psi}}}-\sqrt{\calF_{\calC(\xi)}{\ket{\widetilde{\psi}}}}\right| \leq \varepsilon'.
$$
Since $\ket{\widetilde{\psi}}$ is not a normalized state, we now upper  bound its maximal fidelity with $\calC(\xi)$. Consider a quantum state $\ket{\phi_2} \in \calC(\xi)$ written as $\ket{\phi_2} = \sum_{j=1}^m c_i \ket{s_i}$ for some $m \in \mathbb{N}$, stabilizer states $\{\ket{s_i}\}_{i\in [m]} \in \calS$ and corresponding coefficients $c_i \in \mathbb{C}$ such that $\sum_{i=1}^m |c_i| \leq \xi$. We now obtain
\begin{align}\label{eq:interim_ub_stab_extent_fidelity_tildepsi}
\sqrt{\calF_{\calC(\xi)}{\ket{\widetilde{\psi}}}} = \left|\la \phi_2 | \widetilde{\psi} \ra \right| 
&= \left| \la \phi_2 | \psi \ra - \beta_{k+1} \la \phi_2 | \phi^\perp \ra \right| \\
&\leq \left| \la \phi_2 | \psi \ra \right| + |\beta_{k+1}| \cdot \left| \sum_{i=1}^m c_i \la \phi^\perp | s_i \ra \right| \\
&\leq \left| \la \phi_2 | \psi \ra \right| + \sum_{i=1}^m |c_i| |\beta_{k+1}| \cdot \left|  \la \phi^\perp | s_i \ra \right| \\
&\leq \left| \la \phi_2 | \psi \ra \right| + \sum_{i=1}^m |c_i| \cdot (\varepsilon'/\xi) \\
&\leq 1 + \varepsilon',    
\end{align}
where we have used the triangle inequality and the expression for $\ket{\phi_2}$ in the second line. In the third line, we use the triangle inequality again and then used the fact that $|\beta_{k+1}| \cdot \left| \la s_i | \phi^\perp \ra \right| \leq |\beta_{k+1}| \cdot \sqrt{\stabfidelity{\ket{\phi^\perp}}} \leq \sqrt{\varepsilon} = \varepsilon'/\xi$ from the promise of iterative \Selfcorrection~(Theorem~\ref{thm:iterSC_gen}). Finally, we use that $\sum_{i=1}^m |c_i| \leq \xi$. We now observe\footnote{Here, we use $|a^2 - b^2| = |a+b|\cdot |a-b|$.}
\begin{align*}
\left|\calF_{\calC(\xi)}{\ket{\widetilde{\psi}}} - \calF_{\calC(\xi)}{\ket{\psi}} \right| 
&= \left|\sqrt{\calF_{\calC(\xi)}{\ket{\widetilde{\psi}}}} + \sqrt{\calF_{\calC(\xi)}{\ket{\psi}}} \right| \cdot \left|\sqrt{\calF_{\calC(\xi)}{\ket{\widetilde{\psi}}}} - \sqrt{\calF_{\calC(\xi)}{\ket{\psi}}} \right| \leq 3 \varepsilon',    
\end{align*}
where we used $\calF_{\calC(\xi)}{\ket{\psi}} \leq 1$ and Eq.~\eqref{eq:interim_ub_stab_extent_fidelity_tildepsi} in the second inequality. This proves the lemma.
\end{proof}

\subsection{Learning states with low stabilizer extent}\label{sec:learn_stab_extent}
In this section,  we first show that if $\ket{\psi}$ is promised to have low stabilizer extent i.e., $\xi(\ket{\psi}) \leq \xi$, then $\ket{\widetilde{\psi}}$ is a $\poly(n,\xi/\varepsilon)$ stabilizer rank state that is $(1/2-\varepsilon)$-close to the unknown state $\ket{\psi}$, thereby solving the task of state tomography for such states $\ket{\psi}$ up to constant trace distance. This result is formally stated below:
\begin{theorem}
\label{thm:learn_low_stab_extent_states}
Let $\varepsilon'\in (0,1),\xi\geq 1$. Let $\ket{\psi}$ be an unknown $n$-qubit  state such that $\xi(\ket{\psi}) \leq \xi$. Assuming algorithmic $\PFR$ conjecture, there is a $\poly(n,\xi,1/\varepsilon)$-time algorithm that, given access to  $U_\psi$, $conU_\psi$, copies of $\ket{\psi}$,  outputs $\ket{\phi}$ with $\chi(\ket{\phi}) = \poly(\xi,1/\varepsilon)$ such that
$$
|\la {\phi} | \psi \ra|^2 \geq 1/2 - \varepsilon',
$$
The algorithm also gives a $\poly(n,\xi,1/\varepsilon')$-sized circuit that prepares~$\ket{\phi}$.
\end{theorem}
\begin{proof}[Proof of Theorem~\ref{thm:learn_low_stab_extent_states}]
We will run iterative-\Selfcorrection~(Theorem~\ref{thm:iterSC_gen}) on $\ket{\psi}$ with $\varepsilon$ therein instantiated as $\varepsilon = (\varepsilon'/(2\xi))^2$ and the corresponding value of $\eta = \varepsilon^C = (\varepsilon'/(2\xi))^{2C}$. The output state $\ket{\widetilde{\psi}}$ of iterative-\Selfcorrection~satisfies the following:
\begin{enumerate}[$(i)$]
    \item its stabilizer extent $\chi(\ket{\widetilde{\psi}}) = \poly(1/\eta) = \poly(\xi/\varepsilon')$
    \item the algorithm makes $\poly(n,\xi,1/\varepsilon')$ queries to $U_\psi,\textsf{con}U_\psi$ and runs in  time $\poly(n,\xi,1/\varepsilon')$ and 
    \item $ \left|\langle \varphi|\psi\rangle - \langle \varphi|\widetilde{\psi}\ra\right| \leq \varepsilon'/2$ for all $\ket{\varphi} \in \calC(\xi)$ as a consequence of Corollary~\ref{cor:low_stab_extent_inner_prod}. 
\end{enumerate}
Now,  since our known $\ket{\psi}$ is promised to have stabilizer extent $\xi$, i.e., $\ket{\psi} \in \calC(\xi)$, item $(iii)$ above in particular implies that
\begin{equation}
    \varepsilon'/2 \geq \left|\langle \psi|\psi\rangle - \langle \psi|\widetilde{\psi}\ra\right|=\left|1 - \langle \psi|\widetilde{\psi}\ra\right| \geq \left| 1 - |\la \psi | \widetilde{\psi} \ra| \right| \geq 1 - |\la \psi | \widetilde{\psi} \ra|,
\end{equation}
where the second inequality used reverse triangle inequality. The above now implies $ |\la \psi | \widetilde{\psi} \ra|^2 \geq (1-\varepsilon'/2)^2 \geq 1 - \varepsilon'$. 
Furthermore, this implies that the norm of $\ket{\widetilde{\psi}}$ can be lower bounded as
\begin{align}
\label{eq:normofpsitilde}
1 - \varepsilon'\leq |\la \psi | \widetilde{\psi} \ra|^2\leq \norm{\ket{\widetilde{\psi}}}^2 \cdot \norm{\ket{\psi}}^2  =
\norm{\ket{\widetilde{\psi}}}^2,
\end{align}
where the second inequality follows from Cauchy-Schwartz. Also, the norm of $\ket{\widetilde{\psi}}$ can be bounded from above as  follows: since $\ket{\widetilde{\psi}}=\ket{\psi}-\beta_{k+1}\ket{\phi^\perp}$, we have
$$
\norm{\ket{\widetilde{\psi}}}^2 = 1 + |\beta_{k+1}|^2 - 2 Re(\beta_{k+1} \la \psi | \phi^\perp \ra) \leq 1 + |\beta_{k+1}|^2 + 2 |\beta_{k+1}| \cdot |\la \psi | \phi^\perp \ra| \leq 2 + 2 \varepsilon'.
$$
Consider the state $\ket{\phi} = \ket{\widetilde{\psi}}/\norm{\ket{\widetilde{\psi}}}$ (which is known to the algorithm since $\ket{\widetilde{\psi}}$ is produced by the algorithm explicitly). We then have
$$
|\la \phi | \psi \ra|^2 = \frac{|\la \widetilde{\psi} | \psi \ra|^2}{\norm{\ket{\widetilde{\psi}}}^2} \geq \frac{1 - \varepsilon'}{2(1 + \varepsilon')} \geq \frac{(1-\varepsilon')^2}{2} \geq \frac{1}{2} - \varepsilon',
$$
which proves the desired result.

It remains to show that one can prepare this $\ket{\phi}$. Recall that $\ket{\widetilde{\psi}}=\sum_i\beta_i\ket{s_i}$ where $\ket{s_i}$ are stabilizer states. Now, we can prepare the Clifford unitary that prepares $\ket{s_i}$ (from the all-zero state) in time $O(n^2)$ (using Lemma~\ref{lem:inner_prod_stab_states}), call it $C_i$. If we apply Corollary~\ref{cor:LCUfinal} for the set of Cliffords $\{C_i\}_i$ and coefficients $\{\beta_i\}_i$, then $\LCU$ lemma prepares $\ket{\phi}=\ket{\widetilde{\psi}}/\|\ket{\widetilde{\psi}}\|$ with probability 
$$
\Big(\|\ket{\widetilde{\psi}}\|/\sum_i|\beta_i|\Big)^2\geq \big((1-\varepsilon')/\sum_i|\beta_i|\big)^2 \geq  \big((1-\varepsilon')/k\big)^2\geq \Omega(\eta^4),
$$
where the first inequality used Eq.~\eqref{eq:normofpsitilde}, second inequality used that $|\beta_i|\leq 1$ for all $i\in [k]$ and $k\leq 1/\eta^2$. So the algorithm now repeats the $\LCU$ lemma $O(1/\eta^4)$ many times and when we succeed we have prepared the state corresponding to $\ket{\phi}$ (recall that we know in Lemma~\ref{lemma:lcu_basic} when we have succeeded). The total number of gates used in this procedure is $\poly(1/\eta,n)$. 
\end{proof}
In the above theorem, instead of using iterative $\Selfcorrection$, we could have applied iterative stabilizer bootstrapping algorithm from Theorem~\ref{thm:iterSC_stab_bootstrap} (which does not assume algorithmic $\PFR$ conjecture). This leads to a learning algorithm whose time complexity is $\poly(n,(\xi/\varepsilon)^{\log(\xi/\varepsilon)},\log(1/\delta))$ which is summarized in Result~\ref{result:learn_stab_extent}.

The above theorem implies the following corollary for $\kappa$-rank stabilizer states where we leverage Theorem~\ref{thm:ub_stab_extent_stab_rank_states} which bounds the stabilizer extent of such states and allowing us to use Theorem~\ref{thm:learn_low_stab_extent_states}.
\begin{corollary}
Let $\varepsilon'\in (0,1),\kappa\geq 1$. Let $\ket{\psi}$ be an unknown $n$-qubit  state with stabilizer rank $\chi(\ket{\psi}) \leq \kappa$. Assuming the algorithmic $\PFR$ conjecture, there is a $\poly(n,\kappa^{O(\kappa)},1/\varepsilon')$-time algorithm that, given access to  $U_\psi$, $conU_\psi$, copies of $\ket{\psi}$ and  outputs $\ket{\phi}$ with $\chi(\ket{\phi}) = \poly(\kappa^{O(\kappa)},1/\varepsilon)$ such~that
$$
|\la {\phi} | \psi \ra|^2 \geq 1/2 - \varepsilon',
$$
The algorithm also gives a $\poly(n,\kappa^{O(\kappa)},1/\varepsilon')$-sized circuit that prepares~$\ket{\phi}$.
\end{corollary}
\begin{proof}
This is immediate from Theorem~\ref{thm:learn_low_stab_extent_states} and application of Theorem~\ref{thm:ub_stab_extent_stab_rank_states} which gives an upper bound on the stabilizer extent of $\kappa$ stabilizer-rank states as $\xi(\ket{\psi}) \leq \sqrt{e} \cdot (2 \kappa)^{O(\kappa)}$.
\end{proof}

We can improve the above result when we know that the unknown state $\ket{\psi}$ is produced by a Clifford circuit with few non-Clifford $T$ gates by utilizing Lemma~\ref{lem:stab_extent_clifford_T_circs} which bounds the stabilizer extent for such states. Note that this result could have been obtained using \cite{grewal2023efficient}.
\begin{corollary}
Let $\varepsilon'\in (0,1/2),\kappa\geq 1$. Let  $\ket{\psi}$ is an unknown $n$-qubit quantum state produced by a circuit with Clifford gates and $t \in \mathbb{N}$ many $T$ gates. Assuming the algorithmic $\PFR$ conjecture, there is a $\poly(n,2^t,1/\varepsilon')$-time algorithm that, given access to  $U_\psi$, $conU_\psi$, copies of $\ket{\psi}$,  outputs $\ket{\phi}$ with $\chi(\ket{\phi}) = \poly(2^t,1/\varepsilon)$ such that
$$
|\la {\phi} | \psi \ra|^2 \geq 1/2 - \varepsilon',
$$
The algorithm also gives a $\poly(n,2^t,1/\varepsilon')$-sized circuit that prepares~$\ket{\phi}$.
\end{corollary}
\begin{proof}
This follows from Theorem~\ref{thm:learn_low_stab_extent_states} and the fact above which gives an upper bound on the stabilizer extent of states $\ket{\psi}$ produced by Clifford circuits with $t$ many $T$ gates as $\xi(\ket{\psi}) \leq 2^{t}$.
\end{proof}
We remark that except the work of~\cite{van2023quantum}, we are not aware of works that have looked at quantum state tomography of an unknown $\ket{\psi}$ given access to $U_\psi,\textsf{con}U_\psi$. It is an interesting open question if this access is necessary for proving the results above.

\subsection{Learning decompositions of high stabilizer-dimension states}
We now show how the iterative $\Selfcorrection$ protocol (Algorithm~\ref{alg:iterSC_main}) of Theorem~\ref{thm:iterSC_gen} can also be utilized to express an arbitrary $n$-qubit state as a \emph{structured} decomposition over stabilizer states and an \emph{unstructured} state which  has low fidelity with $\calS(n-t)$ (note that in the usual iterative $\Selfcorrection$ we demanded the same with $t=0$ and here we strengthen it for all $t\leq n$). In particular,  we show the following.

\begin{restatable}{theorem}{iterativeSC_high_stab_dim}
\label{thm:iterSC_high_stab_dim}
Let $\varepsilon,\upsilon\in (0,1)$, $\eta=\left(2^{-t}\varepsilon^6\right)^C$ for a constant~$C>1$. Assuming the algorithmic $\PFR$ conjecture, there is an algorithm that with probability $\geq 1-\upsilon$, satisfies the following: given access to $U_\psi,\textsf{con}U_\psi$, copies of $\ket{\psi}$, outputs 
$\beta \in \calB_\infty^{k},\alpha\in \calB_\infty$ for $i \in [k]$ and stabilizer states $\{\ket{\phi_i}\}_{i\in [k]}$ such that one can write $\ket{\psi}$ as
$$
    \ket{\psi}=\sum_{i\in [k]} \beta_i \ket{\phi_i}+\alpha\ket{\phi^\perp},
$$
where the residual state satisfies $|\alpha|^2\cdot \calF_{\calS(n-t)}{\ket{\phi^\perp}} \leq \varepsilon$. This algorithm uses $\Selfcorrection$ procedure $k$ times and the total runtime is $\poly(n,2^t,1/\varepsilon,\log(1/\upsilon))$.
\end{restatable}
\begin{remark} The algorithm used in the above theorem is that of  Algorithm~\ref{alg:iterSC_main} with the modified stopping condition: we stop at an iteration $k$ if $\Exp_{x \sim q_\Psi}\left[ |\la \phi^\perp | W_x | \phi^\perp \ra|^2 \right] < 2^{-2t} \varepsilon^6$ or $\prod_{j=1}^{k-1} (\widetilde{r}_j^{(k-1)})^2 \leq \varepsilon$. We do not include the proof here since it is exactly the same as Theorem~\ref{thm:iterSC_gen} with replaced~parameters. 
\end{remark}
The above result assumes algorithmic $\PFR$. Without assuming this conjecture, we could have used the algorithmic for agnostic learning states with stabilizer dimension at least $n-t$ from \cite{chen2024stabilizer} and which would have given us a $\poly(n,(2^t/\varepsilon)^{O(\log(1/\varepsilon))},\log(1/\upsilon))$-time algorithm. We now have the following claim regarding the structured state $\ket{\widehat{\psi}}=\sum_{i\in [k]} \beta_i \ket{\phi_i}$ in Theorem~\ref{thm:iterSC_high_stab_dim}, which is the analogue of Lemma~\ref{cor:low_stab_extent_inner_prod} to the case where $t>0$.
\begin{claim} \label{claim:stabdim_interim_result}
Let $\varepsilon\in (0,1),\eta= \left(2^{-2t} \varepsilon^6 \right)^C$. For every $\ket{\psi}$, in $\poly(n,1/\varepsilon,1/\eta)$ time we can find a~stabilizer-rank $(1/\eta^2)=\poly(2^{2t}/\varepsilon)$ state~$\ket{\widehat{\psi}}$ (up to a normalization) using Algorithm~\ref{alg:iterSC_main} such that for any state $\ket{\varphi} \in \calS(n-t)$ we~have
$$
|\langle \varphi|\psi\rangle| - |\langle \varphi|\widehat{\psi}\ra| \leq \sqrt{\varepsilon}.
$$
\end{claim}
\begin{proof}
We run Algorithm~\ref{alg:iterSC_main} with $\eta= \left(2^{-2t} \varepsilon^6 \right)^C$ as stated in the theorem and wherein we stop at an iteration $k$ if $\Exp_{x \sim q_\Psi}\left[ |\la \phi^\perp | W_x | \phi^\perp \ra|^2 \right] < 2^{-2t} \varepsilon^6$ or $\prod_{j=1}^{k-1} (\widetilde{r}_j^{(k-1)})^2 \leq \varepsilon$. Invoking Theorem~\ref{thm:iterSC_high_stab_dim}, we would then find a stabilizer-rank $k\leq 1/\eta^2=\poly(2^{2t}/\varepsilon)$ state (up to some normalization) $\ket{\widehat{\psi}}$ such~that
$$
\ket{\psi} = \ket{\widehat{\psi}} + \alpha \ket{\phi^\perp}.
$$
Furthermore by the promise of Theorem~\ref{thm:iterSC_high_stab_dim}, either $\Exp_{x \sim q_\Psi}\left[ |\la \phi^\perp | W_x | \phi^\perp \ra|^2 \right] < 2^{-2t} \varepsilon^6$ (which by Lemma~\ref{lemma:completeness_stabdim} implies $\calF_{\calS(n-t)}(\ket{\phi^\perp}) < \varepsilon$)  or $|\alpha|^2 \leq \varepsilon$. Now, for any state $\ket{\varphi} \in \calS(n-t)$, we have
\begin{align*}
    \la \varphi | \psi \ra &= \la \varphi | \widehat{\psi} \ra + \alpha \la \varphi | \phi^\perp \ra \\
    \left|\la \varphi | \psi \ra \right| &\leq \left|\la \varphi | \widehat{\psi} \ra \right| + \left| \alpha \right| \left| \la \varphi | \phi^\perp \ra \right|\\
    \left|\la \varphi | \psi \ra \right| - \left|\la \varphi | \widehat{\psi} \ra \right| &\leq  \sqrt{\varepsilon},
\end{align*}
since $|\langle \varphi |\phi^\perp\rangle| \leq \sqrt{\calF_{\calS(n-t)}{\ket{\phi^\perp}}} = \sqrt{\varepsilon}$ or $|\alpha| \leq \sqrt{\varepsilon}$. This proves the claim.
\end{proof}

Similarly, one can generalize the proof of Theorem~\ref{thm:learn_low_stab_extent_states} to the setting where $t>0$, i.e., $\ket{\psi}$ is a sum of states whose stabilizer dimension is $n-t$. The proof of this corollary is exactly the same as Theorem~\ref{thm:learn_low_stab_extent_states}, except that one uses Claim~\ref{claim:stabdim_interim_result} in the proof.
\begin{corollary}
Let $t \in \mathbb{N}$ and $t < n$. Suppose $\ket{\psi}$ is an unknown $n$-qubit quantum state such that $\ket{\psi}=\sum_i c_i\ket{\phi_i}$ where each $\ket{\phi_i}$ has stabilizer dimension $(n-t)$ i.e., $\ket{\phi_i} \in \calS(n-t)$ and $\sum_i|c_i|\leq \xi$. Assuming the algorithmic $\PFR$ conjecture,  there exists an algorithm that outputs $\ket{\widetilde{\psi}}$ such that
$$
|\la \widetilde{\psi} | \psi \ra|^2 \geq \frac{1}{2} - \varepsilon,
$$
using queries to the state preparation unitary $U_\psi,\textsf{con}U_\psi$, copies of $\ket{\psi}$ and in $\poly(n,\xi,2^t,1/\varepsilon)$ time. The algorithm also outputs a circuit that prepares $\ket{\widetilde{\psi}}$ and with gate complexity $\poly(n,\xi,2^t,1/\varepsilon)$.
\end{corollary}


\bibliographystyle{alpha}
\bibliography{references}

\appendix

\section{Proof of correctness of BSG test}\label{appsec:proof_bsg}
In this section, we prove Theorem~\ref{lemma:analysis_bsg}. For convenience, we first define the distribution $\Dpsi$ defined on the set $S$ as follows
$$
\Dpsi(v)=\frac{q_\Psi(v) 2^np_\Psi(v)}{\sum_{x \in S} q_\Psi(x) 2^np_\Psi(x)}.
$$
For a vertex $u\in \calV$, recall that we define the following sets 
\begin{align*}
N_\zeta(u) &= \left\{v \in S : (u,v) \in \calE_\zeta \right\}\\
Q_\zeta(u) &= \left\{v \in N_\zeta(u) : \Pr_{v_1\sim \Dpsi}\left[v_1 \in N_\zeta(u) \text{ and } \Pr_{v_2\sim \Dpsi} \left[v_2 \in N_\zeta(v) \cap N_\zeta(v_1) \right] \leq \rho_1 \right] > \rho_2 \right\}\\
T_\zeta(u) &= \left\{v \in N_\zeta(u) : \Pr_{v_1\sim \Dpsi}\left[v_1 \in N_\zeta(u) \text{ and } \Pr_{v_2\sim \Dpsi} \left[v_2 \in N_\zeta(v) \cap N_\zeta(v_1) \right] \leq \rho_1 \right] \leq \rho_2 \right\}.
\end{align*}
We now state the lemma, whose proof we adapt from~\cite{tulsiani2014quadratic} and prove under the Bell sampling distribution

\bsgtestlemma*

In order to prove the lemma, we first prove the claim that $|N_\zeta(u)|$ is large. To do this, we will require the following result regarding the number of length-$3$ paths in an approximate subgroup, which we will then use to show that the mass of $D_\Psi$ over $N_\zeta(u)$ is $\Omega(\poly(\gamma))$.

\begin{lemma}[{\cite[Lemma~7.13.11]{zhao2023graph}}]\label{lem:length3_paths}
Let $S$ be the approximate subgroup, i.e., $\Pr_{x,y\in S}[x+y\in S]\geq \Delta$. There exists $A' \subseteq S$ such that $|A'| \geq (\Delta^2/8) \cdot |S|$ and for every pair $(x,y) \in A' \times A'$, there are at least $(\Delta^6/1600) \cdot |S|^2$ pairs of points $(a,b) \in S \times S$ such that ``$xaby$" form a path of length $3$.
\end{lemma}

\sizeNu*
\begin{proof}
To see this lower bound, we first express the LHS of the inequality as follows
\begin{align*}
&\sum_{u,v\in S} D_\Psi(u)D_\Psi(v)[v\in N_{\gamma/4}(u)]\\
&=    \sum_{u,v\in S}[u+v\in S]\cdot  D_\Psi(u)D_\Psi(v)\\
&\geq 2^{2n}   \sum_{u,v\in S} [u+v\in S]\cdot \sum_{a}p_\psi(a)p_\psi(a+u)p_\psi(u)\cdot \sum_{b}p_\psi(b)p_\psi(b+v)p_\psi(v)\\
&\geq 2^{2n}   \sum_{u,v\in S} [u+v\in S]\cdot \sum_{a\in S}[a+u\in S]p_\psi(a)p_\psi(a+u)p_\psi(u)\cdot \sum_{b\in S}[b+v\in S]p_\psi(b)p_\psi(b+v)p_\psi(v)\\
&\geq (\gamma/4)^6\cdot 2^{-4n}\sum_{a,b\in S}\sum_{u,v\in S}[a+u,u+v,v+b\in S]\\
&\geq  (\gamma/4)^6\cdot 2^{-4n}\sum_{a,b\in A'\times A'}\sum_{u,v\in S}[auvb \text{ forms a  length-}3 \text{ path }]\\
&\geq  (\gamma/4)^6\cdot 2^{-4n}\sum_{a,b\in A'\times A'} (\gamma^{30}/(1.6 \times 10^9)) \cdot |S|^2\\
&=  (\gamma/4)^6\cdot 2^{-4n} \cdot (|A'|^2/2) \cdot (\gamma^{30}/(2^{10} \times 10^8)) \cdot |S|^2\\
&\geq  (\gamma/4)^6\cdot 2^{-4n}\gamma^{50} |S|^4/(2^{15} \times 10^{12}) \\
&\geq  \gamma^{64}/(2^{39} \times 10^{15}),
\end{align*}
where the third line follows from the definition of $D_\Psi$ (Eq.~\eqref{eq:induced_dist_SAMPLE}. Throughout, we used the fact that the graph defined on the vertices in $S$ contains an edge between $(a,b)$ if $a+b\in S$. In the sixth line, we consider the set $A' \subseteq S$ of Lemma~\ref{lem:length3_paths} which satisfies $|A'| \geq (\gamma^{10}/400) \cdot |S|$ as the density of the graph $\G(S,\E)$ is $\geq \gamma^5/20$. In the seventh line, we used Lemma~\ref{lem:length3_paths} to comment on the number of length-$3$ paths and the number of different pairs in $A'$ is $\binom{|A'|}{2} \geq |A'|^2/2$ in the eighth line. The final inequality follows from noting that $|S| \geq \gamma^2 \cdot 2^n/80$ from Lemma~\ref{lem:sample_bds}.

The result regarding the size of $N(u)$ is now immediate. We know that $N_{\gamma/4}(u) \subseteq N_\zeta(u)$ for $\zeta \in [0,\gamma/4]$. So it is enough to give a lower bound on the size of $N_{\gamma/4}$ which we do as follows. In Lemma~\ref{lem:Dpsilowerbound}, we have shown that
$$
\Exp_{u\sim D_\Psi}\Big[\sum_{v\in N_{\gamma/4}(u)}D_\Psi(v)\Big]\geq  \gamma^{64}/(2^{39} \times 10^{15}).
$$
Using the upper bound of $D_\Psi(x) \leq (2^{10} \cdot 10^2)/(\gamma^{10} \cdot |S|), \,\, \forall x \in S$ from Fact~\ref{fact:D_ub}, we then obtain
$$
\Exp_{u\sim D_\Psi}[|N_{\gamma/4}(u)|]\geq  \gamma^{74}/(2^{49} \times 10^{17}),
$$
which completes the proof.
\end{proof}

We now start proving claims directly corresponding to parts of Theorem~\ref{lemma:analysis_bsg}. We will use the values of the parameters as described in Section~\ref{sec:algo_BSG} and as used in Algorithm~\ref{algo:bsg_test}. Let us now introduce some useful claims that we will need. We first observe that for $\zeta_1',\zeta_2',\zeta_3',\rho_1',\rho_2'$, one has the following inclusion for any $u \in \calV$
$$
T(u,\zeta_1,\zeta_2,\zeta_3,\rho_1,\rho_2) \subseteq T(u, \zeta_1 - \zeta_1', \zeta_2 + \zeta_2', \zeta_3 - \zeta_3', \rho_1 - \Delta\rho_1, \rho_2 + \Delta\rho_2)
$$
For the instantiations of $\zeta_1=\zeta_3=\zeta+\mu$, $\zeta_2 = \zeta-\mu$, $\zeta_1'=\zeta_2'=\zeta_3'=\mu$ and $=\Delta\rho_1=\rho_1/10$ (i.e., $\rho_1'=9/10\cdot\rho_1$), $\Delta\rho_2=\rho_2/10$ (i.e., $\rho_2'=11/10\cdot\rho_2$), it is then clear that
$A^{(1)}(u) \subseteq A^{(2)}(u)$. 
We first item $(ii)$ of Lemma~\ref{lemma:lb_size_A1} regarding a lower bound on the size of $A^{(1)}(u)$ (as mentioned in the statement of Theorem~\ref{lemma:analysis_bsg}), which will be the core part of the main lemma proof. We restate and prove it below.
\lemmacoreBSG*
\begin{proof}
The proof of items $(i),(iii)$ were provided in the main text. We give the proof of item $(ii)$ here.
 Consider parameters $\rho_1,\rho_2 \in (0,1)$ to be chosen later. Let us define the set $Q'(u)$ as
\begin{align}\label{eq:set_Qprime_u}
    &Q'(u) \\
    &:= N_{\zeta + \mu}(u) \setminus T(u, \zeta + \mu,\zeta-\mu,\zeta+\mu,\rho_1,\rho_2) \nonumber \\
    &= \left\{ v \in N_{\zeta + \mu}(u) : \Pr_{v_1 \sim \Dpsi}\left[v_1 \in N_{\zeta - \mu}(u) \text{ and } \Pr_{v_2 \sim \Dpsi} \left[v_2 \in N_{\zeta+\mu}(v) \cap N_{\zeta+\mu}(v_1) \right] \leq \rho_1 \right] \geq \rho_2 \right\}\nonumber.   
\end{align}
Moreover, let us denote $H_{\gamma'}(u)$ corresponding to the parameter $\gamma' \in (0,1)$ (to be fixed later) as
\begin{equation}\label{eq:set_Hu1}
    H_{\gamma'}(u) := \{ v \in N_{\zeta + \mu}(u) : D_\Psi(v) \geq \gamma' \cdot |N_{\zeta + \mu}|^{-1} \}.
\end{equation}
First, note that by definition of $A^{(1)}(u)$  in Eq.~\eqref{eq:set_Qprime_u},  we have
\begin{equation}
\label{eq:exp_size_A1}
\Exp_{u \sim \Dpsi}\left[ |A^{(1)}(u) \cap H_{\gamma'}(u)| \right] = \Exp_{u \sim \Dpsi}\left[ |(N_{\zeta + \mu}(u) \cap H_{\gamma'}(u)) \setminus Q'(u)|\right] \geq \Exp_{u \sim \Dpsi}\left[ |H_{\gamma'}(u)|\right] - \Exp_{u \sim \Dpsi}\left[ |Q'(u)|\right],
\end{equation}
where the third inequality follows from noting that $H(u) \subseteq N_{\zeta + \mu}(u)$ from Eq.~\eqref{eq:set_Hu1}.
Let us bound the two terms on the RHS of Eq.~\eqref{eq:exp_size_A1} separately. We first upper bound the expected size of $Q'(u)$. To that end, we call a pair $(v,v_1)$ $\emph{bad}$ if \footnote{Note that this definition of a bad pair is different from that typically used in the classical setting e.g., in \cite{viola2011combinatorics,tulsiani2014quadratic}. If $D(x)=1/|S|$, the given definitions reduces to the one used in the mentioned works.}
\begin{equation}\label{eq:bad_pair}
\sum_{x \in S} D(x) [x \in N_{\zeta + \mu}(v) \cap N_{\zeta + \mu}(v_1)] \leq \rho_1.
\end{equation}
Indeed, we defined a bad pair as above so as to suggest the following expression for the set $Q'(u)$ from Eq.~\eqref{eq:set_Qprime_u}:
$$
Q'(u) = \left\{ v \in N_{\zeta + \mu}(u) : \Pr_{v_1 \sim \Dpsi}\left[v_1 \in N_{\zeta - \mu}(u) \text{ and } (v,v_1) \text{ is a bad pair } \right] \geq \rho_2 \right\}.
$$
We now have the following observation relating the number of bad pairs to $|Q'(u)|$:
\begin{align}
&\quad \left|\{\text{bad } (v,v_1) : v \in N_{\zeta+\mu}(u), v_1 \in N_{\zeta - \mu}(u)\}\right| \\
&= \sum_{v, v_1 \in S} [v \in N_{\zeta + \mu}(u)] [v_1 \in N_{\zeta - \mu}(u)] [(v,v_1) \text{ is a bad pair}] \\
&\geq \sum_{v \in Q'(u)} \sum_{v_1 \in S} [v_1 \in N_{\zeta - \mu}(u)] [(v,v_1) \text{ is a bad pair}] \\
&= \frac{\gamma^{10}|S|}{C_1}\cdot \sum_{v \in Q'(u)} \sum_{v_1 \in S} \frac{C_1}{\gamma^{10} \cdot |S|} \left[v_1 \in N_{\zeta - \mu}(u) \text{ and } (v,v_1) \text{ is a bad pair } \right]\\
&\geq \frac{\gamma^{10}|S|}{C_1}\cdot  \sum_{v \in Q'(u)} \sum_{v_1 \in S} D(v_1) \left[v_1 \in N_{\zeta - \mu}(u) \text{ and } (v,v_1) \text{ is a bad pair } \right]\\
&\geq  \frac{\gamma^{10} |S|}{C_1}\cdot \rho_2 \cdot |Q'(u)|, \label{eq:interim_size_Qprime_u}
\end{align}
where in the third line, we only consider the sum over the subset of $v \in S$ which are also in $Q'(u)$  (and by definition in Eq.~\eqref{eq:set_Qprime_u} this also satisfies  $v \in N_{\zeta + \mu}(u)$), in the second-to-last inequality we used that $D(x) \leq C_1/(\gamma^{10} \cdot |S|)$ by Fact~\ref{fact:D_ub} with the constant set as $C_1 = 2^{10} \cdot 10^2$ and the final inequality used the definition of $Q'(u)$ to simplify the inner and hence outer sum.

We now use the above as an intermediate step to give an \emph{upper bound} on $\Exp_{u \sim \Dpsi}[|Q'(u)|]$ by upper bounding the following
\begin{align}\label{eq:ub_size_Su}
& \left|\{\text{bad } (v,v_1) : v \in N_{\zeta+\mu}(u), v_1 \in N_{\zeta - \mu}(u)\}\right| \nonumber \\
&= \left|\{\text{bad } (v,v_1) : v \in N_{\zeta+\mu}(u), v_1 \in N_{\zeta + \mu}(u)\}\right| \nonumber \\ 
& \quad + \left|\{\text{bad } (v,v_1) : v \in N_{\zeta+\mu}(u), v_1 \in N_{\zeta - \mu}(u) \setminus N_{\zeta + \mu}(u) \}\right| \nonumber \\
&\leq \underbrace{\left|\{\text{bad } (v,v_1) : v \in N_{\zeta+\mu}(u), v_1 \in N_{\zeta + \mu}(u)\}\right|}_{(*)} + \underbrace{|S| \cdot |N_{\zeta - \mu}(u) \setminus N_{\zeta + \mu}(u)|}_{(**)},
\end{align}
where we used a simple upper bound of the second term in the third line considering all $v \in S$ and $v_1$ that satisfies the mentioned condition. We now analyze $\Exp_{u \sim \Dpsi}[(*)]$ and $\Exp_{u \sim \Dpsi}[(**)]$ separately. To bound $\Exp_{u \sim \Dpsi}[(*)]$, we make the following observation.
\begin{align}
\Exp_{u \sim \Dpsi}[(*)] 
&= \sum_{u \in S} D(u) \sum_{(v,v_1) \text{ bad pairs in } S^2} [ v \in N_{\zeta+\mu}(u), v_1 \in N_{\zeta - \mu}(u)] \\
&= \sum_{(v,v_1) \text{ bad pairs in } S^2} \sum_{u \in S} D(u) [ v \in N_{\zeta+\mu}(u), v_1 \in N_{\zeta - \mu}(u)] \\
&\leq \sum_{(v,v_1) \text{ bad pairs in } S^2} \rho_1 \\
&\leq \rho_1 |S|^2/2,
\label{eq:exp_star1}
\end{align}
where we used Eq.~\eqref{eq:bad_pair} for bad pairs $(v,v_1)$ in the third line and that there are at most $\binom{|S|}{2} \leq |S|^2/2$ bad pairs in the final inequality. 

To bound $\Exp_{u \sim \Dpsi}[(**)]$, we analyze 
$$
\Exp_{u \sim \Dpsi}[|N_{\zeta - \mu}(u) \setminus N_{\zeta + \mu}(u)|] = \Exp_{u \sim \Dpsi}[|N_{\zeta - \mu}(u)|] - \Exp_{u \sim \Dpsi}[|N_{\zeta + \mu}(u)|].
$$ 
Let $\rho_3 \in (0,1)$ be a parameter to be fixed later. We note that $|N_\zeta(u)|$ is monotonically decreasing in $\zeta$ for all $u \in S$. This implies that $\Exp_{u \sim \Dpsi}[|N_\zeta(u)|]$ is monotonically decreasing in $\zeta$. Recall from our choice of parameters for the $\BSG$ test (Section~\ref{sec:algo_BSG}), $|N_\zeta(u)| \leq |S|$ for $\zeta \in [0,\gamma/4]$ and in particular for $\zeta = \gamma/180$. We divide the interval $[\gamma/180,\gamma/18]$ into $1/\rho_3$ equally-sized consecutive sub-intervals of size $(\gamma \cdot \rho_3)/20$ each. Then, by pigeonhole principal at least one sub-interval (which we denote as $[\zeta-\mu,\zeta+\mu]$) in $[\gamma/180,\gamma/18]$ satisfies $\Exp_{u \sim \Dpsi}[|N_{\zeta - \mu}(u)|] - \Exp_{u \sim \Dpsi}[|N_{\zeta + \mu}(u)|] \leq (\gamma \cdot \rho_3/20)\cdot |S|$. As each sub-interval is equally likely to be chosen and there are $1/\rho_3$ many sub-intervals, we would have chosen the sub-interval $[\zeta-\mu,\zeta+\mu]$ with probability at least $\rho_3$. Thus, with probability at least $\rho_3$, we have
\begin{equation}\label{eq:exp_star2}
\Exp_{u \sim \Dpsi}[(**)] = |S| \cdot (\Exp_{u \sim \Dpsi}[|N_{\zeta - \mu}(u)|] - \Exp_{u \sim \Dpsi}[|N_{\zeta + \mu}(u)|]) \leq (\gamma \cdot \rho_3/20) \cdot |S|^2.
\end{equation}
Substituting Eqs.~\eqref{eq:exp_star1},\eqref{eq:exp_star2} into Eq.~\eqref{eq:ub_size_Su} and then using \eqref{eq:interim_size_Qprime_u} gives us that with probability at least $\rho_3$ (or for the choice of the sub-interval $[\zeta-\mu,\zeta+\mu]$) that
\begin{equation}
\label{eq:exp_size_Qu}
    \Exp_{u \sim \Dpsi}[|Q'(u)|] \leq \frac{C_1}{\rho_2 \gamma^{10}} \left(\frac{\rho_1}{2} + \frac{\rho_3 \gamma}{20} \right) \cdot |S|.
\end{equation}
Plugging Eq.~\eqref{eq:exp_size_Qu} and item $(i)$ of the lemma back (with constants explicitly described as part of the proof there) into Eq.~\eqref{eq:exp_size_A1}, for the sub-interval $[\zeta-\mu,\zeta+\mu]$ and using the definitions of $\zeta,\mu$, we have that 
\begin{align}
\Exp_{u \sim \Dpsi}[|A^{(1)}(u) \cap H_{\gamma'}(u)|] 
&\geq \Exp_{u \sim \Dpsi}[|H_{\gamma'}(u)|] - \Exp_{u \sim \Dpsi}[|Q'(u)|] \\
&\geq \left[\frac{\gamma^{138}}{4C_1 C_2^2} - \frac{C_1}{\rho_2 \gamma^{10}} \left(\frac{\rho_1}{2} + \frac{\rho_3 \gamma}{20} \right)\right] \cdot |S|.
\end{align}
By substituting and choosing the parameters in the above expression as
$$
C_1=2^{10}\cdot 10^2,\enspace C_2 = 2^{39} \times 10^{15}, \enspace \rho_1 = \frac{\gamma^{350}}{10240 C_1^3 C_2^5},\enspace\rho_2 = \frac{9\gamma^{202}}{2560 C_1 C_2^3}, \enspace \rho_3 = \frac{\gamma^{349}}{2560 C_1^3 C_2^5}
$$
we get that
\begin{equation}
    \Exp_{u \sim \Dpsi}[|A^{(1)}(u)|] \geq \gamma^{138} \cdot |S|/(8 C_1 C_2^2).
\end{equation}
In particular, this implies that $\Exp_{u \sim \Dpsi}[|A^{(1)}(u)|] \geq (\gamma^{138}/(8 C_1 C_2^2)) \cdot |S|$ with probability at least $\rho_3 = \gamma^{349}/(2560 C_1^3 C_2^5)$. Using Fact~\ref{fact:lowerboundexpectation}, we have that $|A^{(1)}(u)| \geq (\gamma^{138}/(16 C_1 C_2^2)) \cdot |S|$ with probability at least $\gamma^{138}/(16 C_1 C_2^2)$ over the choice of $u \sim D_\Psi$ and with probability at least $\gamma^{349}/(2560 C_1^3 C_2^5)$ over the choice of $\zeta,\mu$. Overall, the desired result then occurs with probability at least $\Omega(\gamma^{487})$ over the choice of $u \sim D_\Psi$ and parameters $\zeta,\mu$ (which are in turn used to define $\zeta_1,\zeta_2,\zeta_3$). This concludes the proof of item $(ii)$.
\end{proof}

An implication of Lemma~\ref{lemma:lb_size_A1} is that the set $A^{(2)}(u)$ is also pretty large with high probability over the choice of $u$ as $A^{(1)}(u) \subseteq A^{(2)}(u)$. To comment on the small doubling of the set $A^{(2)}(u)$ as stated in Theorem~\ref{lemma:analysis_bsg}, we adapt the result of \cite[Claim~4.6]{hatami2018cubic} to our setting.
\smalldoublingAtwo*
\begin{proof}
Recall that $\rho_1' = 10 \rho_1/11,\rho_2' = 10 \rho_2/9$. Suppose $a_1,a_2\in A^{(2)}(u)=T(u,\zeta,\zeta,\zeta,\rho_1',\rho_2')$ where $\zeta$ is the parameter corresponding to the interval for which Lemma~\ref{lemma:lb_size_A1} holds.  By definition of $T(u,\ldots)$, we then have
$$
\sum_{v_1 \in A^{(2)}(u)}D_\Psi(v_1) \Big[\underbrace{\sum_{v_2\in S}D_\Psi(v_2)[v_2\in N(a_i)\cap N(v_1)]\leq \rho_1'}_{E(a_i,v_1)}\Big] \leq \rho_2', \enspace \forall i \in \{1,2\},
$$
where we used the fact that $A^{(2)}(u)\subseteq N(u)$. Now, define the inner indicator as an events $E(a_1,v_1), E(a_2,v_1)$ for $a_1,a_2$ respectively. By a union bound we have that $[E(a_1,v_1)\vee E(a_2,v_1)]\leq [E(a_1,v_1)]+[E(a_2,v_1)]$, so we have that
\begin{align*}
\sum_{v_1 \in A^{(2)}(u)}D_\Psi(v_1) [E(a_1,v_1)\vee E(a_2,v_1)]\leq \sum_{b\in \{0,1\}}\sum_{v_1 \in A^{(2)}(u)}D_\Psi(v_1) [E(a_b,v_1)]\leq 2\rho_2'.
\end{align*}
In particular, the negation of the inequality above implies 
\begin{align}
\sum_{v_1 \in A^{(2)}(u)}D_\Psi(v_1) [\overline{E(a_1,v_1)}\wedge \overline{E(a_2,v_1)}]\geq \left(\sum_{x\in A^{(2)}(u)} D_\Psi(x) \right)-2\rho_2' \geq \frac{\gamma^{202}}{C_4},    
\end{align}
where we used Lemma~\ref{lemma:lb_size_A1} and definition of $\rho_2'$ in the second inequality (and $C_4>1$ is the appropriate constant by taking this difference). Using the upper bound of $D_\Psi(x) \leq C_1/(\gamma^{10} \cdot |S|)$ for all $x \in S$ (with $C_1 = 2^{10} \cdot 10^2)$ from Fact~\ref{fact:D_ub}, we then obtain
\begin{align}
\sum_{v_1 \in A^{(2)}(u)} \Big[\overline{E(a_1,v_1)}\wedge \overline{E(a_2,v_1)}\Big] \geq \frac{\gamma^{212}}{C_1 C_4} |S|,
\end{align}
i.e., for $\Omega(\gamma^{212})|S|$ many $w\in A^{(2)}(u)$, we have that
$$
\rho_1' \leq \sum_{x \in S} D(x)[x\in N(a_i)\cap N(w)] \leq \frac{C_1}{\gamma^{10} \cdot |S|} \sum_x [x\in N(a_i)\cap N(w)], \enspace \forall i \in \{1,2\},
$$
where the second inequality follows from noting $D(x) \leq C_1/(\gamma^{10} \cdot |S|)$ for all $x$ by Fact~\ref{fact:D_ub}. In particular, this implies that 
$$
|N(a_i)\cap N(w)|\geq \frac{\rho_1' \gamma^{10}}{C_1} |S|, \enspace \forall i \in \{1,2\}.
$$
For any such $w$ and $w_i \in N(a_i) \cap N(w)$, $(a_1,w_1,w,w_2,a_2)$ is a length-$4$ path in $\G(S,\E_{\eta + \mu})$. Thus, we have at least
$$
\frac{{\rho_1'}^2 \gamma^{20}}{C_1^2} \cdot \frac{\gamma^{212}}{C_1 C_4} \cdot |S|^3\geq \Omega(\gamma^{932})|S|^3
$$
distinct paths of length $4$ from $v_1$ to $v_2$ each differing in at least one vertex, where we substituted the definition of $\rho_1'$ above. For any such path $(v_1,w_1,w,w_2,v_2)$, we also have that
$$
v_1 + v_2 = (v_1 + w_1) + (w_1 + w) + (w + w_2) + (w_2 + v_2),
$$
where each of the collected terms on the right are in $S$ by the definition of edges in $\G(S,\E)$. We then have
\begin{equation}\label{eq:interim_small_doubling}
|A^{(2)}(u)+A^{(2)}(u)| \cdot \Omega(\gamma^{932}) \cdot |S|^3 \leq |A^{(2)}(u)|^4.
\end{equation}
Noting that $|A^{(2)}(u)| \leq |S|$ and substituting the above into Eq.~\eqref{eq:interim_small_doubling}, we get that 
$$
|A^{(2)}(u)+A^{(2)}(u)| \leq O(1/\gamma^{932}) \cdot |A^{(2)}(u)|.
$$
This proves the desired result.
\end{proof}
\end{document}